\documentclass[11pt]{article}

\usepackage[margin=1in]{geometry}
\usepackage{amsmath,amssymb,amsthm,mathtools}
\usepackage{booktabs}
\usepackage{enumitem}
\usepackage{hyperref}
\usepackage[nameinlink,capitalise,noabbrev]{cleveref}

\usepackage{bbm}

\usepackage{complexity}

\usepackage{multirow}
\usepackage{tabularx}
\usepackage{makecell}
\usepackage{threeparttable}

\usepackage{tikz}
\usetikzlibrary{calc}
\usetikzlibrary{shapes.geometric}
\usetikzlibrary{arrows.meta}
\usetikzlibrary{decorations.pathreplacing}

\usepackage[table]{xcolor}
\definecolor{niceMagenta}{HTML}{C2185B}

\hypersetup{
  colorlinks=true,
  linkcolor=niceMagenta, 
  citecolor=niceMagenta, 
  urlcolor=niceMagenta
}

\theoremstyle{plain}
\newtheorem{theorem}{Theorem}[section]
\newtheorem{lemma}{Lemma}[section]
\newtheorem{claim}{Claim}[section]
\newtheorem{proposition}{Proposition}[section]
\newtheorem{corollary}{Corollary}[section]
\newtheorem{informaltheorem}{Main Result}

\theoremstyle{definition}
\newtheorem{definition}{Definition}[section]
\newtheorem{assumption}{Assumption}[section]
\newtheorem{example}{Example}[section]

\theoremstyle{remark}
\newtheorem{remark}{Remark}[section]

\usepackage{thm-restate}

\usepackage{MnSymbol}
\usepackage[mathlines]{lineno}

\newcommand{\OPT}{\operatorname{OPT}}
\newcommand{\PoA}{\operatorname{PoA}}
\newcommand{\PoAPNE}{\operatorname{PoA}_{\mathsf{PNE}}}
\newcommand{\PoACCE}{\operatorname{PoA}_{\mathsf{CCE}}}
\newcommand{\argmax}{\operatorname*{arg\,max}}
\renewcommand{\E}{\mathbb{E}}
\renewcommand{\R}{\mathbb{R}}

\newcommand{\cA}{\mathcal{A}}
\newcommand{\cG}{\mathcal{G}}

\renewcommand{\vec}[1]{\boldsymbol{#1}}

\newcommand{\optin}{\mathsf{in}}
\newcommand{\optout}{\mathsf{out}}

\newcommand{\defeq}{\coloneqq}

\newcommand{\solcon}{\mathsf{Eq}}

\newcommand{\pMC}{p^{\textrm{mc}}}
\newcommand{\pSH}{p^{\textrm{Sh}}}
\newcommand{\gameMC}{\Gamma^{\textrm{mc}}}
\newcommand{\costcon}{\kappa}
\newcommand{\cutvalue}{\operatorname{Cut}}

\newcommand{\N}{\mathbb{N}}

\usepackage{natbib}
\let\oldcitet\citet
\let\oldcitep\citep
\renewcommand{\citet}[1]{\oldcitet*{#1}}
\renewcommand{\citep}[1]{\oldcitep*{#1}}

\usepackage{xcolor}
\definecolor{niceRed}{RGB}{190,38,38}
\definecolor{blueGrotto}{HTML}{059DC0}
\definecolor{royalBlue}{HTML}{057DCD}
\definecolor{navyBlue}{HTML}{0B579C}
\definecolor{limeGreen}{HTML}{81B622}
\definecolor{nicePurple}{HTML}{9c27b0}
\definecolor{lightRoyalBlue}{HTML}{def2ff} 
\definecolor{gold}{HTML}{ffa300}

\usepackage{color-edits} 
\addauthor[Grigoris]{gv}{gold}
\addauthor[Ioannis]{ia}{limeGreen} 
\addauthor[Chris]{cl}{purple} 
\addauthor[Weiqiang]{wz}{blue} 
\addauthor[Kshipra]{kl}{nicePurple}

\title{Compensation Design\thanks{Authors are ordered alphabetically.}}

\author{
Ioannis Anagnostides%
\thanks{Corresponding authors:
\texttt{ianagnos@cmu.edu} and \texttt{weiqiang.zheng@yale.edu}. Work performed while at Google Research.}\\
CMU
\and
Kshipra Bhawalkar\\
Google Research
\and
Christopher Liaw\\
Google Research
\and
Aranyak Mehta\\
Google Research
\and
Renato Paes Leme\\
Google Research
\and
Yifeng Teng\\
Google Research
\and
Grigoris Velegkas\\
Google Research
\and
Weiqiang Zheng\footnotemark[2]\\
Yale University
}

\date{}

\begin{document}
\pagenumbering{roman}
\pagestyle{empty}
\maketitle

\begin{abstract}
Motivated by current developments in digital ecosystems, we introduce \emph{compensation design}, the problem of designing payment rules that incentivize high-quality contributions in decentralized environments.
Here, a budget-constrained principal (\emph{e.g.}, a platform) with a monotone submodular value function aims to design a payment rule, while agents (\emph{e.g.}, content creators) decide whether to opt in or out depending on their \emph{private cost}. We show that a simple \emph{cost-oblivious} and anonymous marginal-contribution payment rule guarantees that pure Nash equilibria always exist and attain a price of anarchy (PoA)---relative to the omniscient budget-feasible optimum---of at most $(2-\lambda)/(1-\lambda)=2+o_{\lambda}(1)$ in the large-market regime ($\lambda \to 0$) where each individual cost is at most a $\lambda$ fraction of the budget. We further show that the limiting factor $2$ for the PoA is unavoidable among deterministic cost-oblivious rules, even for additive values. In stark contrast, and surprisingly, we identify a counterexample showing that a payment rule based on the \emph{Shapley value}---which is ubiquitous in related profit sharing problems in game theory---may admit \emph{no} pure Nash equilibria.

We then extend our scope to coarse correlated equilibria---outcomes attained by no-regret learning algorithms. This is further motivated by our intractability result: although a pure Nash equilibrium always exists, computing one is \PLS-complete. In this context, we establish that coarse correlated equilibria also attain a PoA bound of at most $2+o_\lambda(1)$, and this guarantee in fact extends even under the payment rule induced by the Shapley value.

Moreover, we move beyond monotone submodular value functions and binary actions. First, for (monotone) XOS valuations, we show that no \emph{oracle-efficient} payment rule can attain a PoA bound of $O(n^{1/2 - \epsilon})$, for any constant $\epsilon > 0$. Second, for submodular but non-monotone valuations, we show that a broad class of natural payment rules fails to guarantee a bounded PoA. Finally, we extend compensation design to the setting where each agent has a \emph{combinatorial action set}. We provide randomized payment rules with logarithmic PoA guarantees for \emph{subadditive values}, and matching lower bounds that apply even in the single-agent additive-value setting.
\end{abstract}
\thispagestyle{empty}

\clearpage

\tableofcontents
\thispagestyle{empty}

\clearpage

\pagenumbering{arabic}
\pagestyle{plain}

\section{Introduction}

The creator economy constitutes one of the workhorses of the burgeoning digital ecosystem, with its value projected to exceed half a trillion dollars within the next few years~\citep{ResearchMarkets26:Creator}. Platforms such as YouTube and Spotify derive much of their value from content created by independent artists, and therefore invest considerable resources to encourage the production of high-quality and diverse content, thereby stimulating user engagement and attracting new users. The same design problem arises in many other digital domains, where platforms may seek to incentivize users or firms to share valuable datasets, contribute computational resources, or participate in federated learning protocols (\Cref{sec:furtherapps}).

Taken together, these applications bring to the fore a central incentive-design problem: how should a budget be distributed so as to encourage valuable participation by external contributors? In each of these highly decentralized settings, the system designer derives value from the collective set of contributions, but \emph{participation cannot be mandated}~\citep{Marden13:distributed,Marden14:generalized}. The challenge is to design payments that induce voluntary participation by a subset of contributors that optimizes system utility, subject to the budget constraint. Moreover, contributors invest considerable time, capital, and other resources when they decide to participate, thereby incurring a \emph{private cost}. Since the platform cannot force participation, each contributor independently decides whether to opt in or out. In particular, a contributor will be inclined to participate only if the compensation offered by the platform outweighs the private cost of opting in.

At the same time, mechanisms that rely on complex elicitation of private information or computationally intractable primitives---such as VCG-type schemes---are often impractical at scale, while highly discriminatory payment rules may lead to perceptions of unfairness and discourage participation. There is, therefore, an imperative to design payment rules that are \emph{anonymous}---equal contributors are treated equally---and \emph{cost oblivious}, in the sense that payments are determined solely by observable measures of utility rather than reported costs.

A natural assumption in this setting is that system utility exhibits \emph{diminishing marginal returns}: the addition of a new contributor can bring considerable value to a nascent platform, but provides less marginal benefit to a platform already saturated with similar content or data. As a result, the platform's objective is naturally modeled as a monotone submodular function. If the platform had full knowledge of contributors' private costs, the problem would reduce to a standard algorithmic problem---namely, submodular maximization subject to a budget constraint. In practice, however, these costs are private and participation decisions are decentralized. The platform must therefore design payment rules that steer the resulting strategic system toward near-optimal equilibria.

\subsection{Compensation design: our basic model}

We introduce \emph{compensation design}: the
problem of designing computable payment rules that induce high-value decentralized participation under a fixed budget and private participation costs. We begin by describing the key aspects of our basic model, and then highlight the differences between our framework and other related models in game theory (\Cref{sec:comparison}).

The central object in the paper is the following game. There is a set $N$ of agents, a budget $B$, and a normalized value function $v:2^N\to\R_{\geq0}$ for the principal. Each agent $i$ has a private cost $c_i$ and chooses whether to opt in or opt out. Denote $S \subseteq N$ the set of agents that opt in. A payment rule assigns nonnegative payments $p_i(S)$ to the agents
who opt in, subject to the budget constraint $\sum_{i\in S}p_i(S)\le B$ for every realized set $S$. Crucially, this payment depends only on the binary signals the principal receives from the agents, as well as the public value function. The payment rule induces a game among agents in which the utility of each agent $i$ is $p_i(S)-c_i$ if $i\in S$, and zero if $i\notin S$. The goal is to design a budget-feasible payment rule that is \emph{cost oblivious} (\textit{i.e.}, it cannot depend on the private
cost vector\footnote{Otherwise, the problem reduces to pure algorithmic optimization, namely submodular maximization subject to a budget constraint, as we explain in~\Cref{appendix:targeted-implementation}.}) and maximizes the principal's value. The benchmark is the full-information budgeted optimum $\OPT:=\max_{S \subseteq N}\{v(S): \sum_{i \in S} c_i \le B\}$, which is the maximum value of any set whose total cost does not exceed the budget $B$. Performance is measured through the price of anarchy (PoA) with respect to the equilibrium concept under consideration, mainly pure Nash equilibria and \emph{coarse correlated equilibria} (\Cref{def:CCE}). A key parameter is the large-market ratio $\lambda=\max_{i\in N} c_i/B$, which measures the highest individual cost as a fraction of the total budget~\citep{Mehta07:Adwords,Anari14:Mechanism}. In the sequel, we also consider compensation design without the large-market assumption and the \emph{combinatorial} setting where agents choose from a set of actions (\Cref{sec:intro-comb}).

\subsubsection{Comparison with existing models}
\label{sec:comparison}

Compensation design is closely related to several existing models. In what follows, we highlight the similarities and the key differences with those models, with a summary given in \Cref{tab:design-comparison}. A more detailed comparison with each of these lines of work is deferred to~\Cref{sec:related}.

\definecolor{cdblue}{HTML}{1F4E79}
\definecolor{cdlight}{HTML}{EAF2F8}
\definecolor{headergray}{HTML}{F3F4F6}

\begin{table}[t]
    \centering
    \small
    \renewcommand{\arraystretch}{1.35}
    \setlength{\tabcolsep}{6pt}

    \caption{Comparison of incentive-design problems by information and participation structure.}
    \label{tab:design-comparison}

    \begin{tabularx}{\linewidth}{
        >{\raggedright\arraybackslash}p{0.45\linewidth}
        >{\centering\arraybackslash}X
        >{\centering\arraybackslash}X
        >{\centering\arraybackslash}X
    }
        \toprule
        \rowcolor{headergray}
        \textbf{Problem setting}
        & \textbf{Actions}
        & \textbf{Costs}
        & \textbf{Participation} \\
        \midrule

        Procurement auction~\citep{Singer10:Budget}
        & Observed
        & Private
        & Centralized \\

        Multi-agent contracts
        & \multirow{2}{*}{Hidden}
        & \multirow{2}{*}{Public}
        & \multirow{2}{*}{Decentralized} \\
        \citep{Dutting26:Multi}
        & & & \\

        \rowcolor{cdlight}
        \textbf{\textcolor{cdblue}{Compensation design [This paper]}}
        & \textbf{\textcolor{cdblue}{Observed}}
        & \textbf{\textcolor{cdblue}{Private}}
        & \textbf{\textcolor{cdblue}{Decentralized}} \\

        \bottomrule
    \end{tabularx}
\end{table}

\begin{itemize}
    \item \emph{Budget-feasible mechanism design}~\citep{Singer10:Budget} is a classic framework for optimizing allocation under private costs: a principal designs a procurement auction such that the agents first report their private costs, and the mechanism then determines an allocation (\emph{i.e.,} which agents are selected) and payments (\emph{i.e.,} how to pay the agents). The key difference is that mechanism design rests on bidding and \emph{centralized allocation}, whereas compensation design promotes voluntary \emph{decentralized participation} without bidding or centralized allocation.
    \item Compensation design is closely related to the problem of \emph{multi-agent contract design}~\citep{Dutting26:Multi}, where a principal aims to incentivize efforts from agents with voluntary participation. Contract design assumes that \emph{the principal knows agents' costs but does not observe their actions}. In contrast, compensation design assumes that \emph{the principal observes agents' actions but must design cost-oblivious payment rules} that perform well for any unknown cost vector. Compensation design is more natural in content platforms where agents' actions (\emph{e.g.}, creating content) are publicly known, yet the costs (\emph{e.g.}, the time and monetary costs of creation) are private and not available to the platform.

    \item Compensation design can be understood through the general framework of \emph{utility design}: agents control actions that determine the social welfare; the principal designs a welfare-sharing mechanism that assigns a utility function to each agent, with the goal of ensuring that agents' local optimization behavior (\emph{i.e.,} the resulting pure Nash equilibria) leads to near-optimal global social welfare. There is a long line of work on utility design across various settings, including profit sharing~\citep{Vetta02:Nash,Augustine11:dynamics,Gollapudi17:profit}, distributed welfare games~\citep{Marden13:distributed}, and distributed resource allocation~\citep{Marden14:generalized}. However, the setting with private costs and a budget constraint has not been studied.
\end{itemize}

Taken together, there are fundamental differences between our setting and prior work in budgeted mechanism design, contract design, and utility design in the context of profit sharing. Our paper revolves around the following set of questions.

\begin{enumerate}
    \item First, can one design efficient, anonymous, budget-feasible, and cost-oblivious payment rules that guarantee equilibrium existence and near-optimal price of anarchy bounds?\label{item:first}
    \item Second, how robust are such guarantees to weaker equilibrium concepts such as coarse correlated equilibria and to broader valuation classes beyond monotone submodular?\label{item:second}
    \item Third, what changes when each agent has a combinatorial action space rather than a single, binary participation decision?\label{item:third}
\end{enumerate}

\subsection{Our results}

The high-level quantitative picture painted by our results is summarized in~\Cref{tab:valuation-action-bounds}. The remainder of this section describes our main results and the key ideas behind them.

\subsubsection{Binary participation under monotone submodular valuations}

Perhaps the most natural payment rule is based on the \emph{Shapley value}~\citep{Shapley53:Value}, a cornerstone of cooperative game theory that is ubiquitous in practice and theory alike. In particular, let $\phi_i(S)$ denote the \emph{Shapley value} of agent $i$ in the cooperative game obtained by restricting $v$ to $S$; that is,
\begin{equation}
    \label{eq:Shapleyval-intro}
\phi_i(S)
=
\sum_{T\subseteq S\setminus\{i\}}
\frac{|T|!(|S|-|T|-1)!}{|S|!}
\bigl(v(T\cup\{i\})-v(T)\bigr).
\end{equation}
The \emph{Shapley payment rule} is defined as
\begin{equation}
    \label{eq:intro-Shapley}
    \pSH_i(S)=
    \begin{cases}
    B\dfrac{\phi_i(S)}{v(S)} & \text{if } i\in S\text{ and }v(S)>0,\\[1.5ex]
    0 & \text{otherwise.}
    \end{cases}
\end{equation}
This payment rule is budget feasible, cost oblivious, and anonymous (\Cref{lem:shapley-efficiency-covering}). Surprisingly, we show that the induced game may admit no pure Nash equilibria.

\begin{table}[t]
\centering
\small
\renewcommand{\arraystretch}{1.35}
\setlength{\tabcolsep}{6pt}

\begin{threeparttable}
\caption{Summary of our main PoA results concerning compensation design.}
\label{tab:valuation-action-bounds}

\begin{tabular}{@{}lccc@{}}
\toprule
& \multicolumn{3}{c}{\textbf{Setting}} \\
\cmidrule(lr){2-4}
\textbf{Valuation class}
&
\makecell{\textbf{Multi-agent}\\\textbf{binary actions}}
&
\makecell{\textbf{Single-agent}\\\textbf{combinatorial actions}}
&
\makecell{\textbf{Multi-agent}\\\textbf{combinatorial actions}}
\\
\midrule

Additive
&
\multirow{2}{*}{\makecell{\(\frac{2-\lambda}{1-\lambda}=2+o_\lambda(1)^\dagger\)\\[-1pt]
{\scriptsize (\Cref{thm:poa-pne})}}}
&
\multirow{4}{*}{\makecell{\(\Theta(\log m)\)\\[-1pt]
{\scriptsize (\Cref{thm:single-agent-upper-basic})}\\[-1pt]
{\scriptsize (\Cref{thm:single-agent-lower-basic})}}}
&
\multirow{2}{*}{\makecell{\(\Theta(\log m)\)\\[-1pt]
{\scriptsize (\Cref{thm:multi-agent-upper})}}}
\\

Submodular
&
&
&
\\

\cmidrule(lr){1-2}
\cmidrule(lr){4-4}

XOS
&
\multirow{2}{*}{\makecell{\(\Omega(n^{1/2-\epsilon})\)\\[-1pt]
{\scriptsize (\Cref{thm:xos-cce-lb})}}}
&
&
\multirow{2}{*}{\makecell{\(\Omega(n^{1/2-\epsilon})\)\\[-1pt]
{\scriptsize (\Cref{thm:xos-cce-lb})}}}
\\

Subadditive
&
&
&
\\

\bottomrule
\end{tabular}

\begin{tablenotes}
\footnotesize
\item \textit{Notes: $n$ denotes the number of agents; in the combinatorial setting, $m$ denotes the total number of agents' actions; and $\lambda \defeq \max_{i \in N} c_i/B$ denotes the large-market parameter. \emph{$^\dagger$}: the $2 + o_\lambda(1)$ PoA bound for pure Nash equilibria can be extended to CCEs (\Cref{thm:cce-submodular-sharp}). Moreover, without the large-market assumption, there is a randomized rule that has a PoA of $3$ under monotone submodular value functions (\Cref{thm:no-large-market-randomized}).}
\end{tablenotes}

\end{threeparttable}
\end{table}

\begin{informaltheorem}
    \label{informal:Shapley}
    Even when the value function is monotone and submodular and the private costs satisfy $\max_{i \in N} c_i/B \ll 1$, there is an instance in which the game induced by the Shapley value does not have a pure Nash equilibrium.
\end{informaltheorem}

This is unexpected because in many related utility design problems, the Shapley value always guarantees the existence of pure Nash equilibria~\citep{HartMasColell89:Potential, Augustine11:dynamics, Marden13:distributed,Marden14:generalized,Bachrach13:Incentives}.
This shows that our setting is fundamentally different. The argument revolves around \emph{quadratic} submodular functions, in which case Shapley values~\eqref{eq:Shapleyval-intro} admit a more tractable characterization (\Cref{lem:sc-quadratic-shapley}).

\paragraph{Marginal contribution rule.} We then consider a simpler payment rule, whereby the compensation to agent $i$ under a realized set $S$ is proportional to its \emph{marginal contribution}:
\begin{equation}
    \label{eq:mc}
    \pMC_i(S) \defeq 
\begin{cases}
    B \frac{v(S) - v(S \setminus \{i\})}{v(S)} & \text{if } v(S) > 0,\\
    0 & \text{otherwise}.
\end{cases}
\end{equation}
If $v$ is monotone and submodular, marginals are nonnegative and their sum is at most $v(S)$, so the rule is nonnegative and budget feasible (\Cref{prop:budget}). It is also cost
oblivious and anonymous: it uses only the realized set and the value oracle, and treats agents with the same contribution symmetrically. Unlike the payment rule based on Shapley values, we show that this simple payment rule always guarantees the existence of a pure Nash equilibrium.

\begin{informaltheorem}
    \label{informalpure}
For any monotone submodular value function and every cost vector with $c_i<B$ for all agents, the game induced by the marginal contribution rule has a pure Nash equilibrium.
\end{informaltheorem}

In fact, we establish a stronger property: the induced game is an \emph{ordinal potential game}, so every sequence of strict better responses terminates (\Cref{thm:existence}). The argument establishes that the function $\Phi(S) = v(S)\prod_{i\in S}(1-c_i/B)$ serves as a potential. An opt-in deviation is profitable exactly when the gain in value is large enough to offset the multiplicative penalty from the entrant's cost; conversely, an opt-out deviation is profitable exactly when removing the agent increases the same potential.

\paragraph{Price of anarchy.} We then examine the \emph{price of anarchy (PoA)}\footnote{It is common to define the price of anarchy in terms of the social welfare $\sum_{i=1}^n u_i$, but here we define it with respect to the underlying value function $v$.} with respect to pure Nash equilibria of the game induced by the marginal contribution rule. 

\begin{informaltheorem}
    \label{informal:poa}
    Let $\lambda=\max_i c_i/B<1$. Every pure Nash equilibrium $S$ of the marginal-contribution game satisfies $\OPT / v(S)\le (2-\lambda)/(1-\lambda) = 2 + o_\lambda(1)$.
\end{informaltheorem}

This result shows that the principal can guarantee a constant fraction of the optimal value without knowing anything about the agents' costs. Moreover, this guarantee is tight, even for additive value functions (\Cref{thm:poa-pne,prop:tight}). We also remark that in the special case of \emph{additive} value, the same PoA bound has been established for a non-truthful first-price procurement auction~\citep{Bhawalkar25:Equilibrium}. In comparison, we show that a simple marginal-contribution payment rule achieves the same guarantee in the more general monotone submodular value setting.

The proof of~\Cref{informal:poa} is based on a relatively standard argument: if an agent $i$ is absent from equilibrium, then its decision not to enter upper bounds the relative marginal gain $(v(S\cup\{i\})-v(S))/v(S)$ by
$(c_i/B)/(1-c_i/B)$. Submodularity then allows summing these inequalities over agents present in an optimal set.

\paragraph{Optimality among cost-oblivious rules.} It turns out that the factor two is not merely an artifact of the marginal-contribution rule. We prove that no nonnegative,
budget-feasible, and cost-oblivious rule can guarantee PoA with respect to pure Nash equilibria below $2$, even for additive values, and even as $\lambda \to 0$. The rule here may depend arbitrarily on agent identities and on the value function; the only information it cannot use is the private cost vector (\Cref{thm:cost-oblivious-without-anonymity-lb}), for otherwise the problem becomes trivial from an information-theoretic standpoint.

\begin{informaltheorem}
    No deterministic cost-oblivious payment guarantees a PoA with respect to pure Nash equilibria below $2$.
\end{informaltheorem}

The obstacle is information-theoretic: by cost obliviousness, if one alters only the cost vector, the induced payment rule must remain the same. We identify an instance where one cost vector forces the rule to admit a certain equilibrium while the other cost vector makes that same equilibrium far from optimal.

\paragraph{Beyond the large-market assumption.} So far, our results focus on the large-market assumption, \emph{i.e., $c_i < B$}. We also investigate the case in which an agent's cost may be as large as or larger than the budget. We first show two negative results on deterministic payment rules: (1) there exists an instance where the PoA of any deterministic rule is $+\infty$ (\Cref{prop:lambda-one}); this lower bound is immediate: one agent with cost $c_i = B$ may be indifferent between opting in and opting out, and opting out can lead to unbounded PoA. (2) Taking a step further, even under principal-favoring tie-breaking---whereby an indifferent agent chooses the action that maximizes the principal's value---there exists an instance where any deterministic rule suffers a PoA of $\Omega(\sqrt{n})$ (\Cref{thm:lower-deterministic-no-large-market}). Both lower bounds hold even with additive values.

We then demonstrate the power of randomized payment rules, circumventing the aforementioned lower bounds for deterministic rules. Specifically, we show that a randomized rule that chooses between the marginal-contribution rule and a singleton-selection rule achieves \emph{constant} PoA for monotone submodular values without the large-market assumption (\Cref{thm:no-large-market-randomized}).

\begin{informaltheorem}
    \label{informal-poa3}
    For monotone submodular value functions, there exists a randomized payment rule with an expected PoA of $3$, without the large-market assumption.
\end{informaltheorem}

This randomization bypasses the deterministic lower bound in a way reminiscent of how the \emph{smart greedy} algorithm achieves a constant-factor approximation for maximizing a submodular function subject to a budget constraint~\citep{Khuller99:budgeted}. In fact, in the monotone submodular setting, the $3$-approximation obtained in~\Cref{informal-poa3} is better than the state-of-the-art (either randomized or deterministic) \emph{truthful} budget-feasible mechanism that achieves a $3.798$-approximation \citep{han25:improved}.

\paragraph{Computation of equilibria.} The existence of pure Nash equilibria notwithstanding (\Cref{informalpure}), we show that strict better-response dynamics can take exponentially many steps,
even under an arbitrarily strong large-market condition. What is more, computing a pure Nash equilibrium of the marginal-contribution game is
\PLS-complete~\citep{Johnson88:Easy} for monotone submodular value functions (\Cref{theorem:exp-br-lower,thm:pls-complete}). The reduction is from the usual
local version of~\textsc{MaxCut}: we encode the cut objective into a monotone submodular
valuation so that unilateral participation changes simulate \textsc{FLIP}
improvements. On the positive side, for
additive valuations with integral values, better-response paths have pseudo-polynomial length $O(W^2)$, where $W$ is an upper bound on the total value of all agents
(\Cref{cor:fastconvergence}).

\paragraph{Coarse correlated equilibria.} Motivated by the hardness of computing pure Nash equilibria and their non-existence under the Shapley payment rule, we examine the efficiency of \emph{coarse correlated
equilibria (CCE)} (\Cref{def:CCE}). CCE is a relaxed notion of Nash equilibrium that \emph{always exists} and can be approached \emph{efficiently} by no-regret learning algorithms~\citep{Cesa06:prediction, Blum08:Regret}. Without knowledge of the platform's value function and other participants' private costs, it is often impossible for an agent to compute and play a pure Nash equilibrium. However, the agents could employ simple, efficient learning algorithms that require only their own utilities, and the whole system would then converge to the set of CCE. It is thus important to understand the efficiency of outcomes produced by learning dynamics.

For monotone submodular values, we prove that every CCE of the marginal-contribution game attains an asymptotically tight large-market guarantee: the PoA of CCE is at most
$2/(1-2\lambda)$, and hence at most $2+o_{\lambda}(1)$ when $\lambda<1/2$ (\Cref{thm:cce-submodular-sharp}); the PoA of CCE is at most $2 + 1/(1-\lambda)$ when $\lambda < 1$ (\Cref{thm:cce-submodular}). These results in fact hold even for the game induced by the Shapley value (\Cref{cor:cce-marginal-covering}), where pure Nash equilibria may not exist. 

\begin{informaltheorem}
    Let $\lambda=\max_i c_i/B < 1$. For monotone submodular values, every coarse correlated equilibrium $\mu$ of the marginal-contribution game satisfies $\OPT / \E_{S\sim\mu}[v(S)] \leq 2/(1-2\lambda) = 2 + o_{\lambda}(1)$ if $\lambda < 1/2$ and $\OPT / \E_{S\sim\mu}[v(S)] \leq 2+1/(1-\lambda)$ if $\lambda < 1$. The same guarantees also hold for the Shapley payment rule.
\end{informaltheorem}

Notably, our argument does not go through the usual \emph{smoothness framework}~\citep{Roughgarden15:Intrinsic}, which has been the main engine behind many PoA bounds for CCEs~\citep{Syrgkanis13:Composable,Roughgarden17:Price}. 

\paragraph{Price of stability of CCEs.} A price of anarchy analysis revolves around worst-case equilibria. Another natural question concerns the best-case equilibrium, which is captured by the so-called \emph{price of stability}~\citep{Anshelevich08:Price}. This can be motivated when a platform can somehow \emph{steer} agents toward more desirable equilibria.

In this context, we point out that CCEs exhibit an interesting phenomenon that has no analogue for pure equilibria: there can exist a CCE whose expected value is higher than the full-information budget-feasible optimum! In other words, the price of stability with respect to CCEs can be strictly smaller than one. In particular, we find an example where the price of stability is approaching $1/2$, and this is in fact tight within the class of additive valuations (\Cref{sec:additive-cce-ratio-sharp}).

Attaining a price of stability strictly below one may seem at first glance impossible because $\OPT$ is the benchmark that an omniscient principal could achieve subject to the budget constraint. The explanation is that a CCE is a distribution over participation profiles, whereas budget feasibility constrains payments in each realized profile rather than the private cost of the selected agents. This means that a CCE may assign positive probability to coalitions whose realized cost exceeds $B$; the CCE condition only forces cost feasibility \emph{in expectation}. For additive values, this turns the comparison with $\OPT$ into the standard fractional relaxation of knapsack, whose value can exceed the integral optimum by at most a factor of two.

Taken together, these results give a sharp interpretation of the marginal-contribution rule in the binary, monotone-submodular model. It is simple, anonymous, cost oblivious, budget feasible, admits pure Nash equilibria, and achieves the best possible PoA among cost-oblivious rules.

\subsubsection{Beyond monotone submodular valuations}

We next extend our scope beyond monotone submodular functions. For monotone XOS valuations, we prove an oracle lower bound: no payment designer using only polynomially many value queries can guarantee $O(n^{1/2-\epsilon})$ price of anarchy with respect to CCEs for every XOS valuation, even with \emph{known} unit costs and $\lambda=1/\sqrt n$ (\Cref{thm:xos-cce-lb}). The construction is based on the standard information-theoretic lower bound for XOS maximization~\citep{Mirrokni08:Tight}. In particular, a polynomial-query payment designer cannot distinguish a base valuation, on which all approximate CCEs have low value by budget feasibility, from a planted valuation whose optimum is significantly higher. This shows that constant PoA guarantees for CCEs are impossible for oracle-efficient payment design over XOS valuations (\Cref{thm:xos-cce-lb}).

\begin{informaltheorem}
    For every constant $\epsilon>0$ and inverse-polynomial $\eta(n)\le 1/n$, no oracle-efficient payment-rule design algorithm can guarantee, for every monotone XOS valuation, that all $\eta(n)$-CCEs obtain expected value at least $\OPT/O(n^{1/2-\epsilon})$. In particular, a payment rule that guarantees $\PoACCE = O(n^{1/2 - \epsilon})$ requires exponentially many value queries. This holds even with known unit costs and $\lambda=1/\sqrt n$.
\end{informaltheorem}

For non-monotone submodular objectives, the obstruction is different. First, if costs are allowed to be zero, then no nonnegative budget-feasible rule can guarantee finite PoA (\Cref{prop:zero-cost-harmful-impossibility}); this is driven by the fact that certain harmful agents will be inclined to participate to the detriment of the efficiency of the system. Moreover, even when each cost is strictly positive, we show that a broad class of payment rules inevitably incurs unbounded pure price of anarchy (\Cref{thm:nm-equal-marginal-lb}).

\subsubsection{Combinatorial compensation design}
\label{sec:intro-comb}

Finally, we move beyond binary participation. This is motivated by many applications in which a single agent controls many possible actions. For example, a creator may produce any subset of videos, a data provider may contribute any subset of datasets, and a compute
provider may allocate any subset of jobs. The resulting \emph{combinatorial} compensation design problem is much harder than the binary-action case, since the payment rule now must 
incentivize high-quality contributions from an agent without knowing their private costs.

\paragraph{Single-agent combinatorial compensation design.} The combinatorial compensation design problem is already non-trivial in the single-agent setting. In this setting, there is one agent with $m$ actions $A$, where each action $e \in A$ has cost $c_e$. The challenge is that the principal needs to entice high effort with a carefully designed payment rule, while in the binary-action setting the principal could just pay the whole budget to the agent.

We establish that, even for additive value functions, any deterministic rule has price of anarchy at least $\Omega(m)$ (\Cref{thm:single-agent-lower-deterministic}), and any randomized payment rule has price of anarchy at least $\Omega(\log m)$ (\Cref{thm:single-agent-lower-basic}). The intuition for the randomized lower bound is that, given a value function, the payment rule must be fixed for any cost vector. Under uniform costs, the payment induces a \emph{demand curve} for the number of actions the agent chooses. If the cost is $t$, the offline optimum grows as $1/t$, then it can be shown that any single rule is bound to miss some scale by a logarithmic factor.

On the positive side, we show that it is possible to achieve a matching PoA upper bound of $O(\log m)$ using a simple randomized payment rule, even for subadditive value functions (\Cref{thm:single-agent-upper-basic}).

\begin{informaltheorem}
    In the single-agent combinatorial model with $m$ actions and monotone subadditive values, there is a randomized payment rule that has PoA of $O(\log m)$ whenever every action is individually feasible. This logarithmic dependence is tight even for additive values and uniform costs.
\end{informaltheorem}

The key observation is that if the principal knew the optimal value $\OPT$, there is a straightforward payment rule: pay the whole budget to the agent if the value is at least $\OPT$. In view of this, we propose a simple randomized threshold-prize rule that samples a target value on a geometric scale. We show that it attains a PoA of $O(\log m)$ even for monotone \emph{subadditive values}---under the mild assumption that every action is individually feasible (\Cref{thm:single-agent-upper-basic}).\footnote{In fact, logarithmic PoA continues to hold without the assumption that every action is individually feasible or the assumption that the value function is subadditive (we only need it to be monotone). Using the same idea, we can still guarantee a PoA of $O(\log V)$, where $V:= v(A)/\min_e v(\{e\})$ is the largest ratio between the upper and lower bounds on $\OPT$.}

\paragraph{Multi-agent combinatorial compensation design.} We then extend the model to multiple agents, which is the most general problem addressed in our paper. Here, there are $n$ agents and each agent controls a block of actions. We denote by $m$ the total number of actions. The principal has budget $B > 0$ and aims to maximize a monotone value function over the union of chosen actions. When the value function is XOS, we have a lower bound of $\Omega(n^{1/2-\epsilon})$ already in the multi-agent binary-action setting. Thus, we focus on the case of monotone submodular functions.

Since the multi-agent setting generalizes the single-agent setting, we know that it is impossible to beat $\Omega(\log m)$. Whether it is possible to achieve a PoA of $O(\log m)$ is not straightforward. The multi-agent setting is harder because the payment rule must not only incentivize high-quality actions within each agent's action set, but also manage competition among multiple agents, all without knowing the private costs. Nevertheless, we show that it is still possible to achieve $O(\log m)$ (\Cref{thm:multi-agent-upper}).

\begin{informaltheorem}
    In the multi-agent combinatorial setting with $m$ total actions and monotone submodular values, there is a randomized payment rule that has PoA of $O(\log m)$ whenever every action is individually feasible. 
\end{informaltheorem}

The randomized rule combines ideas from both the single-agent combinatorial and multi-agent binary-action settings. We first sample a target value on a geometric scale intended to approximate $\OPT$. Given such a target value $R$, the mechanism then uses the marginal contribution payment based on the \emph{capped payment potential} $F_R(S):= B \min \{\frac{v(S)}{R},1\}$. Here, the cap $R$ incentivizes agents to contribute value at least $R$, while the marginal contribution structure ensures sufficient competition among agents. We show that when $R$ is a constant-factor lower estimate of $\OPT$, this branch achieves constant PoA. Randomization over the geometric scale thus gives a PoA of $O(\log m)$.



\section{Further related work}
\label{sec:related}

This section elaborates on related lines of work.

\subsection{Contract design} 

There is a burgeoning line of work on algorithmic contract theory~\citep{Saig26:Adaptive,Dutting26:Multi,Feldman26:Equal,Dutting26:Contracts,Ezra26:Contract,Dutting25:Combinatorial,Feldman25:One,Dutting24:Ambiguous, Deo24:supermodular, Babaioff06:Combinatorial}. We refer to the excellent surveys of~\citet{Feldman25:Combinatorial} and~\citet{Dutting24:Algorithmic} for further pointers to that line of work.


In the multi-agent binary-action setting, compensation design is closest in spirit to the multi-agent contract problem~\citep{Dutting26:Multi}, in which the principal seeks to incentivize a group of agents to exert effort. More specifically, there is a set of agents $N$ with binary actions (exert effort or not exert effort) and each agent $i$ incurs cost $c_i$ for exerting effort and cost $0$ otherwise. The principal has a monotone utility function $f: 2^N \rightarrow [0,1]$ that maps the set of agents $S \subseteq N$ who exert effort to a reward.\footnote{Strictly speaking, the multi-agent contract model assumes a normalized reward $1$ for success, and $f$ denotes the probability of success. So $f$ denotes the expected reward for the principal.} The principal does not observe the agents' actions but knows the costs $c$ and incentivizes effort using a \emph{linear} contract $\vec{\alpha} \in \mathbb{R}^{N}_{\ge 0}$.\footnote{There are other models for contract design where agents have hidden types and their realized costs may be private~\citep{Guruganesh21:contracts, guruganesh23:power, Alon21:contracts, Alon22:approximate, Castiglioni22:designing, Castiglioni23:multi, castiglioni25:reduction,Gan22:generalized}. These models generally make assumptions that are incomparable to our model of compensation design. We refer to~\citeauthor[Section 6]{Dutting24:Algorithmic} for a detailed survey.}  If $S$ is the set of agents that exert effort, then agent $i$'s expected payment is $\alpha_i f(S)$. This induces a game among agents, where agent $i$'s utility is $\alpha_i f(S) - c_i$ if $i \in S$ and $\alpha_i f(S)$ otherwise. Given $(N, f, c)$, the goal of the principal is to design a linear contract that, under a pure Nash equilibrium $S$, maximizes her \emph{expected profit} $(1 - \sum_{i=1}^n \alpha_i) f(S)$. The multi-agent contract model has also been studied in the budget-feasible setting with more general objectives such as social welfare and principal's value~\citep{Feldman25:budget,Aharoni25:welfare}, under coarse correlated equilibria~\citep{Dutting26:black}, and in the combinatorial-action setting~\citep{Dutting2025:multi}.

The key differences between the multi-agent contract model and compensation design are:
\begin{itemize}
    \item[1.] In the multi-agent contract model, the agents' costs are provided as part of the input, and the principal can design a contract based on them. In compensation design, we assume that costs are private and that the payment rule is cost-oblivious.
    \item[2.] In the multi-agent contract model, it is assumed the principal cannot observe agents' actions, whereas compensation design assumes the principal can. In content platforms such as YouTube and TikTok, we can observe whether each provider creates content.
\end{itemize}
Compensation design is natural for settings such as content platforms, where each provider's action (\emph{e.g.}, creating content) is publicly observable, but the costs (\emph{e.g.}, the time and monetary costs of creation) are often private and not available to the platform.

\subsection{Budget-feasible mechanism design and procurement auctions}

Budget-feasible mechanism design and procurement auctions were first proposed by~\citet{Singer10:Budget}. Here, a buyer with a budget aims to buy from multiple sellers whose items have private costs. Most existing mechanisms work as sealed-bid procurement auctions that first let sellers report their costs and then, based on the reports, decide an allocation and payment. Budget-feasible mechanism design has been extensively studied, focusing on approximation to the full-information optimal value~\citep{Chen11:Approximability, Dobzinski11:Mechanisms, Bei12:budget, Anari18:Budget, Gravin20:Optimal, Jalaly21:Simple, Neogi24:Budget,Gkatzelis25:Procurement} and competitive ratios of sequential posted-price mechanisms~\citep{Badanidiyuru12:Learning,Charalampopoulos25:Competitive}. This line of work also includes budget-feasible mechanisms for non-monotone submodular objectives~\citep{Amanatidis19:Budget}. There are also budget-feasible mechanisms based on \emph{clock auctions}~\citep{Milgrom20:clock, Balkanski22:deterministic, Balkanski25:deterministic, han23:triple, han25:improved, han2025efficient} that have many appealing properties. With value oracle access, a randomized mechanism achieves a tight $2$ approximation for additive values~\citep{Gravin20:Optimal}; a recent work by~\citet{han25:improved} gives a deterministic clock auction with a $3.798$ approximation for monotone submodular values, a randomized auction with $9.742$ for non-monotone submodular values; both of these guarantees are state-of-the-art among all possibly randomized truthful budget-feasible mechanisms. In comparison, one of our contributions is to provide a simple alternative framework that achieves PoA of $2+o_\lambda(1)$ or $3$ (without the large-market assumption) in the monotone submodular value setting.

The underlying full-information problem of maximizing a monotone submodular function subject to a budget constraint is \NP-hard, but its approximability is well understood: the best polynomial-time approximation factor is $1-1/e$~\citep{Nemhauser78:Analysis,Sviridenko04:Note,Badanidiyuru19:Optimization}.

More closely related to our work, \citet{Bhawalkar25:Equilibrium} recently studied budget-feasible first-price (non-truthful) procurement auctions in the \emph{additive} value setting and established equilibrium existence and constant PoA bounds for a greedy allocation rule under the large-market assumption. It is perhaps surprising that, using only payments, we show that compensation design achieves guarantees similar to those of a first-price procurement auction, even in the more general monotone submodular value setting and without the large-market assumption. A more recent work by~\citet{Neogi25:multidimensional} initiates the study of multidimensional budget-feasible mechanism design, similar to the combinatorial-action extension . They show that constant approximation to the full-information benchmark is impossible; they then consider a weaker benchmark and develop constant-approximation mechanisms. It is an interesting future direction to identify meaningful weaker benchmarks that enable constant approximation in the combinatorial compensation design setting.



In principle, one can employ a procurement auction mechanism in our setting with the buyer being the principal and the sellers being the agents. The key difference is that mechanism design assumes the power of \emph{centralized allocation} to enforce the allocation rule, while compensation design promotes voluntary \emph{decentralized participation} without a bidding or allocation process. In particular, procurement auctions are not suitable for the applications that motivate our work. First, platforms cannot simply procure contributions on demand: participation is voluntary and decentralized, taking the form of an opt-in or opt-out decision rather than a bidding process. Moreover, the sheer scale of modern digital ecosystems makes repeated auction-based elicitation impractical. For example, running a fresh procurement auction for each use of content among millions of AI agents is clearly out of reach. These considerations make it clear that modern platforms need to design simple payment rules that can be applied on a continuous basis at scale, without requiring constant elicitation of private costs.



\subsection{Utility design and welfare-sharing games}

Compensation design is also connected to the \emph{utility design} literature, where a designer aims to design utilities so that agents' decentralized optimization leads to high welfare. A central reference point is the class of \emph{valid utility games} of~\citet{Vetta02:Nash}, which gives PoA guarantees for marginal-contribution-type utility sharing in submodular welfare settings. Subsequent work studies related welfare-sharing and profit-sharing models~\citep{Augustine11:dynamics, Gollapudi17:profit}, distributed welfare games~\citep{Marden13:distributed}, distributed resource allocation~\citep{Marden14:generalized}, and collaborative environments with Shapley-based incentives~\citep{Bachrach13:Incentives}. The main distinction is that those models typically treat utilities as shares of welfare chosen by the designer, without a budget constraint and without private participation costs. In our setting, the payment rule must be budget feasible for every realized set and cost oblivious across all possible private costs. Equilibrium behavior in our setting is fundamentally different. For example, the standard sharing rule based on the Shapley value may admit no pure Nash equilibria (\Cref{informal:Shapley}).

\subsection{Further applications}
\label{sec:furtherapps}

An additional real-world example related to compensation design is Parallel---a startup building a search engine for AI agents---which recently launched \emph{Index}~\citep{Fortune26,parallel26}. Index aims to compensate publishers, data providers, and independent creators when AI agents utilize their work, determining compensation by estimating each source's \emph{Shapley value} to reflect its marginal contribution. Further applications are discussed below.

\paragraph{Incentivizing data sharing in federated learning.} Incentivizing participation is also a central theme in \emph{federated learning}~\citep{Li20:Federated,Kairouz21:Advances}. A growing literature is examining incentive issues in such settings, where clients may withhold data due to privacy concerns, free-ride on others' contributions, or require compensation for participation~\citep{Karimireddy22:Mechanisms,Zhan21:Survey,Zeng21:Comprehensive,Kong22:Incentivizing,Yang23:Federated,Alaei26:Incentivizing}. This line of work is motivated by a similar participation friction: valuable contributions are supplied by decentralized clients with private participation costs or privacy costs. The focus, however, is typically on incentives within a learning protocol, data valuation, or privacy-aware participation. Our model abstracts away the learning protocol itself and shifts the focus to the design of payment rules for general submodular objectives.

\paragraph{Value-sharing in emerging AI platforms.} \citet{Wang24:Economic} propose a framework based on the Shapley value for compensating copyright owners whose data contributes to generative-AI outputs. More recently, \citet{Harris26:Incontext} study in-context credit assignment using the least core to reduce coalition-level undercompensation among creators whose intellectual property appears in a model's context window. These works focus on attribution and value division for realized AI-generated outputs or contexts. Our focus is complementary: we study payment rules as an \emph{ex ante} incentive device in decentralized systems, where agents decide whether to participate before the realized set of contributions is known.

\paragraph{Revenue sharing in music streaming.} Finally, our framework is also related to recent work on revenue sharing in music streaming platforms such as Spotify. A central question in this literature is how subscription revenue should be divided among artists. The two main rules are \emph{pro-rata} payments, where each artist is paid according to their share of total streams on the platform, and \emph{user-centric} payments, where each subscriber's payment is divided only among the artists that subscriber listens to. \citet{Alaei22:Revenue} compare these two rules and show that the \emph{pro-rata} rule can have advantages that are not captured by the usual criticism that it cross-subsidizes heavy users. In particular, they identify conditions under which \emph{pro-rata} sharing can be Pareto-superior to user-centric sharing. From a cooperative game-theoretic perspective, \citet{Bergantinos25:Revenue} give axiomatic foundations for these streaming payment rules, including characterizations of the \textit{pro-rata} and user-centric rules. These works divide realized subscription revenue among artists, whereas our model studies incentives for voluntary participation under a fixed budget and private costs.

\section{The basic model}

We now define our basic model in more detail. \Cref{fig:compensation-setting} contains an illustration. \Cref{sec:single-agent-combinatorial,sec:multi-agent-combinatorial} extend the scope of compensation design beyond the binary-action setting.

\subsection{Compensation design with binary participation}

We let $N=\{1,\ldots,n\}$ be the set of \emph{agents} (for example, content creators or data providers). The \emph{principal} (for example, the platform) has a budget $B>0$. The principal's goal is to maximize the value $v(S)$ of the participating set $S$. We use value-oracle access as the baseline oracle model: given a subset $S\subseteq N$, the oracle returns $v(S)$. This value can be interpreted as the aggregate utility generated by the participating agents, for example over a sequence of user queries. (Stronger oracle models, such as demand oracles, are discussed as a direction for future work in~\Cref{sec:conclusions}.)

Our blanket assumptions concerning the value function are gathered below; more general value functions are considered in \Cref{sec:beyond monotone submodular} and \Cref{sec:combinatorial}.

\begin{assumption}[Monotonicity and submodularity]
    \label{ass:blanket}
    We make the following assumptions concerning the value function $v: 2^N\to \R_{\geq 0}$. 
    \begin{itemize}
        \item \emph{normalization}: $v(\emptyset)=0$;
        \item \emph{monotonicity}: if $S \subseteq S' \subseteq N$, then $v(S) \leq v(S')$; and
        \item \emph{submodularity}: if $S \subseteq S' \subseteq N$ and $i \notin S'$, then
        \[
            v( S \cup \{i\}) - v(S) \geq v( S' \cup \{i\}) - v(S').
        \]
    \end{itemize}
\end{assumption}

In the basic version of our model, we assume that each agent $i \in [n]$ has a \emph{private cost} $c_i \geq 0$, and has two possible actions, which we denote by $\optin$ and $\optout$; we denote by $\cA_i \defeq \{\optin, \optout\}$ $i$'s action set. These modeling assumptions, which are in line with the targeted applications, are reversed relative to the contract design literature, where the cost is typically assumed to be known but the action unobserved. We generally operate in the \emph{large-market regime}~\citep{Mehta07:Adwords,Anari14:Mechanism}, wherein
\[
        \lambda \defeq \max_{i\in N}\frac{c_i}{B} \ll 1.
\]
\begin{center}
    

\begin{figure}[!t]
\centering
\resizebox{0.74\textwidth}{!}{%
\begin{tikzpicture}[
    font=\scriptsize,
    >=Latex,
    x=1cm,
    y=1cm,
    platformbox/.style={
        draw=navyBlue,
        thick,
        rounded corners=4pt,
        fill=lightRoyalBlue,
        minimum width=56mm,
        minimum height=10mm
    },
    contentcard/.style={
        draw=limeGreen!60!black,
        thick,
        rounded corners=2pt,
        fill=limeGreen!10,
        minimum width=16mm,
        minimum height=9mm,
        align=center
    },
    emptycard/.style={
        draw=gray!60,
        thick,
        dashed,
        rounded corners=2pt,
        fill=gray!7,
        minimum width=16mm,
        minimum height=9mm,
        align=center,
        text=gray!65!black
    },
    agentcircle/.style={
        draw=limeGreen!60!black,
        very thick,
        circle,
        fill=limeGreen!12,
        minimum size=8.5mm
    },
    outcircle/.style={
        draw=gray!60,
        very thick,
        circle,
        fill=gray!12,
        minimum size=8.5mm
    },
    agentlabel/.style={
        align=center,
        font=\scriptsize\bfseries,
        text=limeGreen!45!black
    },
    outlabel/.style={
        align=center,
        font=\scriptsize\bfseries,
        text=gray!65!black
    },
    costbox/.style={
        draw=gold!70!niceRed,
        thick,
        rounded corners=3pt,
        fill=white,
        text width=25mm,
        align=center,
        inner xsep=2pt,
        inner ysep=2pt
    },
    payment/.style={
        ->,
        thick,
        draw=navyBlue,
        shorten >=1pt
    },
    valuearrow/.style={
        ->,
        thick,
        draw=limeGreen!55!black,
        shorten >=1pt
    },
    privatearrow/.style={
        ->,
        thick,
        dashed,
        draw=gold!70!niceRed,
        shorten >=1pt
    },
    zeropay/.style={
        ->,
        thick,
        dashed,
        draw=gray!65,
        shorten >=1pt
    },
    noflow/.style={
        ->,
        thick,
        dashed,
        draw=gray!60,
        shorten >=1pt
    },
    separator/.style={
        draw=gray!65,
        semithick,
        dash pattern=on 5pt off 4pt
    },
    arrowlabel/.style={
        font=\scriptsize,
        fill=white,
        inner sep=1.1pt,
        rounded corners=1pt
    }
]
    \coordinate (col1) at (-4.55,0);
    \coordinate (col2) at (-1.05,0);
    \coordinate (dots) at (1.45,0);
    \coordinate (coln) at (4.35,0);

    \node[platformbox] (platform) at (-.55,4.48) {};
    \node[font=\scriptsize\bfseries, text=navyBlue, align=center] at (-1.58,4.52)
        {Platform};
    \node[font=\scriptsize, text=navyBlue, anchor=west] at (-.22,4.62)
        {budget $B$};
    \node[font=\scriptsize, text=navyBlue, anchor=west] at (-.72,4.36)
        {payment rule $p_i(S)$};

    \begin{scope}[shift={(-2.85,4.49)}, scale=.46]
        \draw[navyBlue, thick, fill=white, rounded corners=2pt]
            (-.90,-.48) rectangle (.90,.56);
        \draw[navyBlue, thick] (-.90,.28) -- (.90,.28);
        \fill[navyBlue] (-.68,.42) circle[radius=.035];
        \fill[navyBlue] (-.52,.42) circle[radius=.035];
        \fill[navyBlue] (-.36,.42) circle[radius=.035];
        \draw[navyBlue, thick, rounded corners=1pt]
            (-.58,-.25) rectangle (-.12,.08);
        \draw[navyBlue, thick, rounded corners=1pt]
            (.16,-.25) rectangle (.62,.08);
        \draw[navyBlue, thick] (-.42,-.06) -- (-.28,-.06);
        \draw[navyBlue, thick] (.32,-.06) -- (.46,-.06);
        \draw[navyBlue, thick] (0,-.48) -- (0,-.72);
        \draw[navyBlue, thick] (-.36,-.72) -- (.36,-.72);
    \end{scope}

    \draw[separator] (-5.55,3.50) -- (5.55,3.50);

    \node[contentcard] (content1) at ($(col1)+(0,2.43)$) {};
    \node[font=\scriptsize\bfseries, text=limeGreen!45!black] at ($(content1)+(0,.17)$)
        {content};
    \draw[limeGreen!55!black, thick] ($(content1)+(-.38,-.05)$) -- ($(content1)+(.22,-.05)$);
    \draw[limeGreen!55!black, thick] ($(content1)+(-.38,-.21)$) -- ($(content1)+(.34,-.21)$);
    \draw[limeGreen!55!black, thick] ($(content1)+(-.38,-.37)$) -- ($(content1)+(.03,-.37)$);

    \node[contentcard] (content2) at ($(col2)+(0,2.43)$) {};
    \node[font=\scriptsize\bfseries, text=limeGreen!45!black] at ($(content2)+(0,.17)$)
        {content};
    \draw[limeGreen!55!black, thick] ($(content2)+(-.38,-.05)$) -- ($(content2)+(.22,-.05)$);
    \draw[limeGreen!55!black, thick] ($(content2)+(-.38,-.21)$) -- ($(content2)+(.34,-.21)$);
    \draw[limeGreen!55!black, thick] ($(content2)+(-.38,-.37)$) -- ($(content2)+(.03,-.37)$);

    \node[font=\large, text=gray!70] at ($(dots)+(0,2.43)$) {$\cdots$};

    \node[emptycard] (contentn) at ($(coln)+(0,2.43)$) {no content};
    \draw[gray!65, thick] ($(contentn)+(-.38,-.30)$) -- ($(contentn)+(.38,.30)$);
    \draw[gray!65, thick] ($(contentn)+(.38,-.30)$) -- ($(contentn)+(-.38,.30)$);

    \node[agentcircle] (agent1) at ($(col1)+(0,1.02)$) {};
    \fill[limeGreen!55!black] ($(agent1)+(0,.08)$) circle[radius=.085];
    \draw[limeGreen!55!black, very thick, rounded corners=3pt]
        ($(agent1)+(-.24,-.23)$) .. controls ($(agent1)+(-.18,-.08)$) and
        ($(agent1)+(.18,-.08)$) .. ($(agent1)+(.24,-.23)$);
    \node[agentlabel, anchor=north] at ($(agent1.south)+(0,-.08)$)
        {Agent 1\\opt in};

    \node[agentcircle] (agent2) at ($(col2)+(0,1.02)$) {};
    \fill[limeGreen!55!black] ($(agent2)+(0,.08)$) circle[radius=.085];
    \draw[limeGreen!55!black, very thick, rounded corners=3pt]
        ($(agent2)+(-.24,-.23)$) .. controls ($(agent2)+(-.18,-.08)$) and
        ($(agent2)+(.18,-.08)$) .. ($(agent2)+(.24,-.23)$);
    \node[agentlabel, anchor=north] at ($(agent2.south)+(0,-.08)$)
        {Agent 2\\opt in};

    \node[font=\large, text=gray!70] at ($(dots)+(0,1.02)$) {$\cdots$};

    \node[outcircle] (agentn) at ($(coln)+(0,1.02)$) {};
    \fill[gray!65] ($(agentn)+(0,.08)$) circle[radius=.085];
    \draw[gray!65, very thick, rounded corners=3pt]
        ($(agentn)+(-.24,-.23)$) .. controls ($(agentn)+(-.18,-.08)$) and
        ($(agentn)+(.18,-.08)$) .. ($(agentn)+(.24,-.23)$);
    \node[outlabel, anchor=north] at ($(agentn.south)+(0,-.08)$)
        {Agent $n$\\opt out};

    \node[costbox] (cost1) at ($(col1)+(0,-.98)$)
        {private cost $c_1$\\[-1pt]
        {\tiny\textcolor{gold!65!niceRed}{incurred if opt in}}};
    \node[costbox] (cost2) at ($(col2)+(0,-.98)$)
        {private cost $c_2$\\[-1pt]
        {\tiny\textcolor{gold!65!niceRed}{incurred if opt in}}};
    \node[costbox] (costn) at ($(coln)+(0,-.98)$)
        {private cost $c_n$\\[-1pt]
        {\tiny\textcolor{gold!65!niceRed}{not incurred}}};

    \foreach \cost in {cost1,cost2,costn} {
        \draw[gold!70!niceRed, thick, fill=white, rounded corners=.5pt]
            ($(\cost.west)+(-.08,-.12)$) rectangle ($(\cost.west)+(.10,.10)$);
        \draw[gold!70!niceRed, thick]
            ($(\cost.west)+(-.05,.10)$) arc[start angle=180,end angle=0,radius=.075];
    }

    \draw[valuearrow] (agent1.north) -- (content1.south);
    \draw[valuearrow] (agent2.north) -- (content2.south);
    \draw[noflow] (agentn.north) -- (contentn.south);

    \coordinate (valuehub) at (-2.78,3.16);
    \draw[valuearrow] (content1.north) -- (valuehub);
    \draw[valuearrow] (content2.north) -- (valuehub);
    \draw[valuearrow] (valuehub)
        -- node[arrowlabel, text=limeGreen!45!black, xshift=-1.1cm,yshift=-0.2cm, pos=.56, right]
            {$v(S)$}
        ($(platform.south)+(-1.48,0)$);

    \draw[payment] ($(platform.west)+(0,-.16)$)
        .. controls (-5.96,3.94) and (-6.52,1.34) ..
        node[arrowlabel, text=navyBlue, xshift=-0.1cm, pos=.54, left] {$p_1(S)$}
        (agent1.west);
    \draw[payment] ($(platform.south)+(1.02,0)$)
        .. controls (.58,3.46) and (.30,1.34) ..
        node[arrowlabel, text=navyBlue, xshift=0.1cm, pos=.58, right] {$p_2(S)$}
        (agent2.east);
    \draw[zeropay] ($(platform.east)+(0,-.16)$)
        .. controls (5.66,3.96) and (6.74,1.35) ..
        node[arrowlabel, text=gray!65, xshift=0.2cm, pos=.44, right] {$0$ payment}
        (agentn.east);

    \draw[
        decorate,
        decoration={brace, mirror, amplitude=5pt},
        thick,
        draw=limeGreen!55!black
    ]
        ($(cost1.south west)+(.05,-.22)$) --
        node[
            below=5pt,
            font=\scriptsize\bfseries,
            text=limeGreen!45!black
        ]
            {$S=\mbox{opt-in agents}$}
        ($(cost2.south east)+(-.05,-.22)$);
\end{tikzpicture}%
}
\caption{Compensation design with binary participation. Agents independently
choose whether to opt in or opt out. Opt-in agents incur private costs, generate
content, form the realized set $S$, and receive payments $p_i(S)$ subject to
the budget $B$. Opt-out agents generate no content and receive payment zero. We later extend compensation design to combinatorial action sets (\Cref{sec:multi-agent-combinatorial}).}
\label{fig:compensation-setting}
\end{figure}
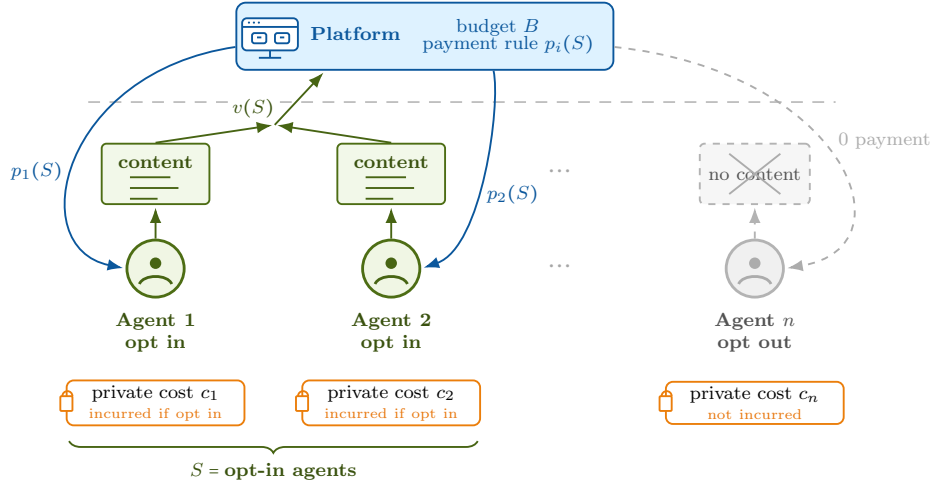

\end{center}

A joint action profile $(a_1, \dots, a_n)$ is in one-to-one correspondence with a set $S \defeq \{ i \in N : a_i = \optin \}$. The principal designs a \emph{payment rule} $p : 2^N \to \R_{\geq 0}^n$. This enforces \emph{limited liability}: the payment attached to each agent is nonnegative. The payment rule should be budget feasible such that $\sum_i p_i(S) \le B$ for all $S \subseteq N$. We typically assume that if $i \notin S$, it receives no payment, so its utility is zero. On the other hand, if $i \in S$, its utility is $p_i(S) - c_i$. Under a joint action $(a_1, \dots, a_n)$ associated with $S \subseteq N$, we denote by $u_i(S)$ the utility of agent $i$. Namely,
\[
    u_i(S) \defeq
    \begin{cases}
        0 & \text{if } i \notin S,\\
        p_i(S) - c_i & \text{otherwise}.
    \end{cases}
\]
In what follows, it will sometimes be convenient to identify a joint action profile by the associated subset $S$.

\paragraph{Equilibrium concepts.} Perhaps the most natural solution concept in this model is the \emph{pure} Nash equilibrium, formally defined below.

\begin{definition}[Pure Nash equilibrium]
    A \emph{pure Nash equilibrium} is a set $S \subseteq N$ such that no agent can strictly improve its utility by switching its participation decision:
    \[
        p_i(S) - c_i \geq 0 \text{ if } i \in S \text{ and } p_i(S \cup \{i\}) - c_i \leq 0 \text{ if } i \notin S.
    \]
\end{definition}
In general games, pure Nash equilibria may or may not exist (\emph{e.g.},~\citealp{Fabrikant04:Complexity}). Nevertheless, we shall show that the game induced by a suitable payment rule always admits a pure Nash equilibrium through a potential argument.

A more permissive solution concept is the \emph{coarse correlated equilibrium}~\citep{Moulin78:Strategically}, which is in turn a relaxation of \emph{correlated equilibria}~\citep{Aumann74:Subjectivity}.

\begin{definition}[Coarse correlated equilibrium]
    \label{def:CCE}
    An $\epsilon$-\emph{coarse correlated equilibrium} is a distribution $\mu \in \Delta(\cA_1 \times \dots \times \cA_n)$ such that for any agent $i \in [n]$ and action $a_i' \in \cA_i$,
    \[
        \E_{(a_1, \dots, a_n) \sim \mu} u_i(a_1, \dots, a_n) \geq  \E_{(a_1, \dots, a_n) \sim \mu} u_i(a_i', \vec{a}_{-i}) - \epsilon.
    \]
\end{definition}
Above, we used the standard notation $\vec{a}_{-i} \defeq (a_1, \dots, a_{i-1}, a_{i+1}, \dots, a_n)$. If $\mu$ in~\Cref{def:CCE} is a \emph{product} distribution, one obtains the usual notion of a mixed Nash equilibrium~\citep{Nash51:Non}.

\paragraph{Price of anarchy.} The goal of the principal is to maximize the value function subject to the budget constraint:
\begin{equation}
        \OPT = \max\left\{v(S): S \subseteq N,\ \sum_{i \in S} c_i \le B\right\}.        \label{eq:opt}
\end{equation}
For a game $\Gamma$, we let $\solcon(\Gamma) \neq \emptyset$ denote a set of solutions according to some equilibrium concept. The price of anarchy with respect to the solution concept $\solcon$ is defined as
\begin{equation}\label{eq:PoA}
     \PoA_{\solcon}(\cG) = \sup_{\Gamma \in \cG}
        \frac{\OPT(\Gamma)}{\inf \{ \E_{S \sim \mu}[v(S)] : \mu \in \solcon(\Gamma) \}},
\end{equation}
where $\cG$ is a set of instances. If $\OPT > 0$ and the denominator is zero, the ratio is to be interpreted as $+\infty$. We will write $\PoAPNE$ to denote the price of anarchy with respect to pure Nash equilibria (assuming existence of such equilibria) and $\PoACCE$ for the price of anarchy with respect to coarse correlated equilibria.

\subsection{Budget-share payment rules}
\label{sec:budget-share-rules}

The payment rules we consider can be described as shares of the total budget.

\begin{definition}[Budget-share payment rule]
A \emph{budget-share payment rule} assigns nonnegative shares $y_i(S)$ to selected agents and pays
\[
        p_i(S)=B y_i(S),
        \text{ where }
        \sum_{i\in S}y_i(S)\le 1
        \text{ for every } S\subseteq N.
\]
Agents outside $S$ receive no payment.
\end{definition}

\begin{definition}[Marginal-covering budget-share rule]
    \label{def:marg-cov}
For $\gamma \in (0,1]$, a budget-share rule is \emph{$\gamma$-marginal-covering} if, for every $S \subseteq N$ with $v(S)>0$ and every $i\in S$,
\[
        y_i(S)\ge
        \gamma \frac{v(S)-v(S\setminus\{i\})}{v(S)}.
\]
\end{definition}

The case $\gamma=1$ will be especially relevant, capturing the scenario where each selected agent's budget share covers their normalized marginal contribution.

\subsubsection{Shapley payments}

Perhaps the most natural payment rule is the one induced by the \emph{Shapley value}~\citep{Shapley53:Value}. We recall that for a nonempty coalition $S$ and an agent $i\in S$, $\phi_i(S)$ denotes the Shapley value of agent $i$ in the cooperative game obtained by restricting $v$ to $S$:
\begin{equation}
    \label{eq:Shapley}
    \phi_i(S)=
        \sum_{T\subseteq S\setminus\{i\}}
        \frac{|T|!(|S|-|T|-1)!}{|S|!}
        \bigl(v(T\cup\{i\})-v(T)\bigr).
\end{equation}
The Shapley value has many important applications, including profit sharing between internet service providers~\citep{Ma07:Internet}, measuring the influence of nodes in social networks~\citep{Narayanam10:Shapley}, identifying salient genes for biological functions~\citep{Moretti07:Class}, feature attribution in machine learning~\citep{Lundberg17:Unified}, and responsibility or importance measures in database queries and temporal logics~\citep{Mascle21:Responsibility,Khalil22:Complexity}.

In the context of compensation design, we define the \emph{Shapley payment rule} as
\begin{equation}
\label{eq:shapley-payment}
        \pSH_i(S)=
        \begin{cases}
        B\dfrac{\phi_i(S)}{v(S)} & \text{if } i\in S\text{ and }v(S)>0,\\[1.5ex]
        0 & \text{otherwise.}
        \end{cases}
\end{equation}
This payment rule is indeed budget feasible---in fact, it fully exhausts the budget---and satisfies the marginal-covering condition per~\Cref{def:marg-cov}.

\begin{restatable}[Efficiency and marginal-covering property of Shapley payments]{lemma}{efficShap}
\label{lem:shapley-efficiency-covering}
For every nonempty set $S$,
\[
        \sum_{i\in S}\phi_i(S)=v(S).
\]
Moreover, if $v$ is monotone and submodular, then for every $i\in S$,
\[
        \phi_i(S)\ge v(S)-v(S\setminus\{i\}).
\]
As a result, the Shapley payment rule~\eqref{eq:shapley-payment} is a $1$-marginal-covering budget-share rule.
\end{restatable}
For completeness, we provide the proof in~\Cref{appendix:proofs}.

\subsubsection{Marginal contribution payments}
\label{secsec:marginalcontr}

We also consider the \emph{marginal contribution payment rule}. For a subset $S \subseteq N$ and an agent $i \in S$, we let
\[
        \Delta_i(S) \defeq v(S)-v(S\setminus\{i\}).
\]
The marginal contribution payment rule is defined as
\begin{equation}
        \pMC_i(S)=
        \begin{cases}
        B\dfrac{\Delta_i(S)}{v(S)} & \text{if } i\in S\text{ and }v(S)>0,\\[1.5ex]
        0 & \text{otherwise.}
        \end{cases}                                      \label{eq:mc-rule}
\end{equation}
By monotonicity (\Cref{ass:blanket}), this guarantees limited liability. Just like~\eqref{eq:shapley-payment}, this payment rule is \emph{anonymous} (or non-discriminatory) in the sense that it uses the same formula for every agent without taking into account provider identities. It is also clearly \emph{cost oblivious}, as it only depends on the set of agents that decided to opt in together with each agent's marginal contribution. We denote by $\gameMC$ the binary-action $n$-player game induced by the marginal contribution payment rule defined in~\eqref{eq:mc-rule}.

We now show that this payment rule is indeed budget feasible. We begin by establishing a simple upper bound on the sum of the marginal contributions, which is a direct consequence of submodularity.

\begin{lemma}
\label{lem:sum-marginals}
For every subset $S\subseteq N$, $
        \sum_{i\in S}\Delta_i(S) \le v(S)$.
\end{lemma}

\begin{proof}
If $S=\emptyset$, the claim is immediate. Otherwise, if $k \defeq |S|$, we write $S=\{i_1,\ldots,i_k\}$ and $S_t=\{i_1,\ldots,i_t\}$. By convention, we take $S_0 \defeq \emptyset$. Since $S_{t-1}\subseteq S\setminus\{i_t\}$, submodularity yields
\[
        v(S_t)-v(S_{t-1})
        = v(S_{t-1}\cup\{i_t\})- v(S_{t-1})
        \ge v(S)-v(S\setminus\{i_t\})
        =\Delta_{i_t}(S).
\]
Summing over $t=1,\ldots,k$, the telescoping sum gives
\[
        v(S)=\sum_{t=1}^k \bigl(v(S_t)-v(S_{t-1})\bigr)
        \ge \sum_{i\in S}\Delta_i(S). \qedhere
\]
\end{proof}

This lemma readily implies budget feasibility.

\begin{proposition}[Budget feasibility]\label{prop:budget}
The marginal-contribution rule is budget feasible:
\[
        \sum_{i\in S} \pMC_i(S)\le B
        \quad \forall S \subseteq N.
\]
\end{proposition}

\begin{proof}
Consider any subset $S \subseteq N$. If $v(S)=0$, all payments are zero, so the claim follows. If $v(S)>0$, then by~\Cref{lem:sum-marginals},
\[
        \sum_{i\in S} \pMC_i(S)
        = B \frac{\sum_{i\in S}\Delta_i(S)}{v(S)}
        \le B,
\]
as claimed.
\end{proof}

This shows that the marginal-contribution rule is a legitimate budget-share rule. Moreover, by definition, it is $1$-marginal-covering in the sense of~\Cref{def:marg-cov}.

We conclude this section by pointing out that any pure Nash equilibrium $S$ under a budget-feasible payment rule is cost feasible. Indeed, if $i \in S$, then $p_i(S)\ge c_i$, for otherwise agent $i$ would prefer to opt out; summing over participating agents and using budget feasibility gives $\sum_{i\in S} c_i \le \sum_{i\in S} p_i(S)\le B$. Interestingly, this is \emph{not} the case under weaker notions, such as coarse correlated equilibrium, as discussed further in~\Cref{sec:additive-cce-ratio-sharp}.

\section{Pure Nash equilibrium}

The first natural question in our model is whether a pure Nash equilibrium exists. This section answers this in the affirmative for the game induced by the marginal contribution rule. In contrast, we show that the answer is negative for the game induced by the Shapley payment rule. Beyond existence, we go on to examine the complexity of computing a pure Nash equilibrium.

\subsection{Existence in the marginal contribution game}

We first show that pure Nash equilibria exist under the marginal contribution payment rule.

\begin{theorem}[Existence of pure Nash equilibria]
\label{thm:existence}
Let $\max_{i \in N} c_i < B$. In the game $\gameMC$ induced by the marginal contribution payment rule~\eqref{eq:mc-rule}, every sequence of strict better responses terminates. As a result, $\gameMC$ admits a pure Nash equilibrium.
\end{theorem}

\begin{proof}
We define the potential
\[
        \Phi(S) = v(S) \prod_{i\in S} \left( 1 - \frac{c_i}{B} \right).
\]
To handle zero-value states, we use the lexicographic potential
\[
        \Phi'(S)=\left(\Phi(S),-\sum_{i\in S}c_i\right).
\]
We will show that every strict unilateral improvement strictly increases $\Phi'$. First, consider an agent $i \notin S$ who decides to opt in. We can assume that $v(S\cup\{ i \}) > 0$, for otherwise opting in cannot yield a strict utility improvement. This deviation is strictly profitable if and only if
\[
        B\frac{v(S\cup\{i\})-v(S)}{v(S\cup\{i\})} > c_i \Longleftrightarrow \left(1-\frac{c_i}{B} \right) v(S\cup\{i\}) > v(S).
\]
Multiplying both sides by $\prod_{i' \in S} (1 - c_{i'}/B) > 0$, we get $
        \Phi(S\cup\{i\})>\Phi(S)$,
thereby increasing the first coordinate of $\Phi'$. 

Conversely, consider an agent $i \in S$ who decides to opt out. We use the shorthand notation $S_{-i} = S \setminus\{i\}$. Suppose first that $v(S)>0$. Opting out is strictly profitable if and only if
\[
        B \frac{v(S)-v(S_{-i})}{v(S)}<c_i \Longleftrightarrow  v(S_{-i})>
        \left(1-\frac{c_i}{B}\right)v(S).
\]
Multiplying both sides by $\prod_{i' \in S_{-i}} (1 - c_{i'}/B) > 0$, we get $\Phi(S_{-i})>\Phi(S)$.
So, every strictly profitable opt-out deviation strictly increases the first coordinate of $\Phi'$.

It remains to consider the case where $v(S)=0$ and $i\in S$ opts out. By monotonicity, $v(S_{-i})=0$, so $\Phi(S_{-i})=\Phi(S)=0$. Such a deviation is strictly profitable only if $c_i>0$, in which case
\[
        \sum_{i' \in S_{-i}} c_{i'} < \sum_{i' \in S} c_{i'}.
\]
Thus, the second coordinate of $\Phi'$ strictly increases.

Since the game has finitely many states, namely $2^n$, it follows that every sequence of strict better responses terminates. Any terminal state is a pure Nash equilibrium.
\end{proof}

\begin{remark}
The proof of~\Cref{thm:existence} in fact shows that, restricted to positive-value states, $\gameMC$ is an \emph{ordinal potential game}~\citep{Monderer96:Potential}. The second coordinate in $\Phi'$ is only used to handle opt-out deviations from zero-value states.
\end{remark}

\subsection{Nonexistence under Shapley payments}
\label{sec:shapley-counter}

In stark contrast, we show that a pure Nash equilibrium need not exist under the normalized Shapley payment rule $\pSH$ defined in~\eqref{eq:shapley-payment}. This nonexistence persists even though $\pSH$ is a $1$-marginal-covering budget-share rule (\Cref{lem:shapley-efficiency-covering}), and even in the large-market regime.

\subsubsection{A quadratic submodular family}
\label{sec:quadraticsub}

In our construction, it is convenient to use value functions of the form
\begin{equation}
\label{eq:sc-quadratic-value}
    v(S)=\sum_{i\in S}a_i-
    \sum_{\{i,i' \}\subseteq S}b_{i i' },
\end{equation}
where $b_{i i' }=b_{i' i} \geq 0 $.

\begin{lemma}[Marginals and Shapley values of a quadratic game]
\label{lem:sc-quadratic-shapley}
For the value function defined in~\eqref{eq:sc-quadratic-value}, the following identities hold.
\begin{enumerate}
    \item For $i\notin S$,
    \begin{equation}
    \label{eq:sc-quadratic-marginal}
        v(S\cup\{i\}) - v(S)=a_i - \sum_{i' \in S} b_{i i'} .
    \end{equation}
    As a result, if $b_{i i'} \geq 0$ for every $i, i' \in N$, then $v$ is submodular. Moreover, $v$ is monotone if
    \begin{equation}
    \label{eq:sc-quadratic-monotonicity-condition}
        a_i - \sum_{i' \in N \setminus\{i\}} b_{i i'} \geq 0
        \qquad\text{for every } i\in N.
    \end{equation}
    \item For every nonempty $S$ and every $i\in S$,
    \begin{equation}
    \label{eq:sc-quadratic-shapley}
        \phi_i(S)=a_i- \frac12 \sum_{i' \in S\setminus\{i\}}b_{i i'}.
    \end{equation}
\end{enumerate}
\end{lemma}

\begin{proof}
First, \eqref{eq:sc-quadratic-marginal} follows by subtracting the two expressions in~\eqref{eq:sc-quadratic-value}. If $S\subseteq S' \subseteq N \setminus\{i\}$, then
\[
\bigl(v(S\cup\{i\})-v(S)\bigr)
-
\bigl(v(S' \cup\{i\} )- v( S' ) \bigr)
=
\sum_{i' \in S' \setminus S} b_{i i' } \geq 0,
\]
which is precisely the diminishing marginal returns condition. Thus, $v$ is submodular. Moreover, the right-hand side of~\eqref{eq:sc-quadratic-marginal} is minimized at $S=N\setminus\{i\}$, so~\eqref{eq:sc-quadratic-monotonicity-condition} implies monotonicity.

For the Shapley value equation, we fix $S$ and $i\in S$. In a uniformly random permutation of $S$, let $P_i$ be the set of agents preceding $i$. The coefficient of a given $T\subseteq S\setminus\{i\}$ in~\eqref{eq:Shapley} is exactly the probability that $P_i=T$. By~\eqref{eq:sc-quadratic-marginal}, the marginal contribution of $i$ in that permutation is
\[
    a_i-\sum_{i' \in P_i} b_{i i'}.
\]
For each fixed $i' \in S\setminus\{i\}$, symmetry of a uniformly random ordering gives
\[
    \Pr(i' \in P_i) = \Pr(i' \text{ precedes }i)=\frac12.
\]
Taking expectations yields~\eqref{eq:sc-quadratic-shapley}.
\end{proof}

Our goal is to prove the following theorem.

\begin{theorem}[Nonexistence of pure Nash equilibria under Shapley payments]
\label{thm:shapley-no-pne}
For every $\bar\lambda>0$, there is an instance with four agents, a monotone submodular value function, and strictly positive costs satisfying $\max_{i \in N} c_i / B \leq \bar\lambda$ such that the game induced by $\pSH$ per~\eqref{eq:shapley-payment} has no pure Nash equilibrium.
\end{theorem}

In the proof, we will make use of the following simple lemma.

\begin{lemma}
\label{lem:sc-deviation-sign}
Suppose that $v(T)>0$ for every nonempty $T\subseteq N$.
\begin{enumerate}
    \item If $i\notin S$, then opting in is a strict improvement for $i$ at $S$ if and only if
    \begin{equation}
        \label{eq:first-ineq}
        \frac{\phi_i(S\cup\{i\})}{v(S\cup\{i\})}>\frac{c_i}{B}.
    \end{equation}
    \item If $i\in S$, then opting out is a strict improvement for $i$ at $S$ if and only if
    \begin{equation}
        \label{eq:second-ineq}
        \frac{\phi_i(S)}{v(S)}<\frac{c_i}{B}.
    \end{equation}
\end{enumerate}
\end{lemma}

\begin{proof}
If $i\notin S$, its current utility is zero. Its utility after opting in reads
\[
    B\left(
        \frac{\phi_i(S\cup\{i\})}{v(S\cup\{i\})}
        -\frac{c_i}{B}
    \right).
\]
This is strictly positive exactly when~\eqref{eq:first-ineq} holds. Similarly, if $i\in S$, its current utility is
\[
    B\left(\frac{\phi_i(S)}{v(S)}-\frac{c_i}{B}\right),
\]
whereas opting out gives zero utility. This means that opting out is strictly profitable exactly when~\eqref{eq:second-ineq} holds.
\end{proof}

\begin{proof}[Proof of~\Cref{thm:shapley-no-pne}]
We fix $\bar\lambda>0$ and set
\[
    t\coloneqq \min\{1,3\bar\lambda\},
\]
so $0 < t \leq 1$. Let $N=\{1,2,3,4\}$. The costs are chosen such that
\begin{equation}
\label{eq:sc-costs-parametric}
    \frac{c_1}{B}=\frac{t}{4},
    \quad
    \frac{c_2}{B}=\frac{t}{6},
    \quad
    \frac{c_3}{B}=\frac{t}{5},
    \quad
    \frac{c_4}{B}=\frac{t}{3}.
\end{equation}
Then $\max_{i \in N} c_i / B = t / 3 \leq \bar\lambda$, as claimed. We define the value function $v$ by~\eqref{eq:sc-quadratic-value} with the following coefficients. For the singleton terms,
\begin{equation}
\label{eq:sc-a-coefficients}
    a_1=\frac{t(t+9)}{3(t+12)}, \quad
    a_2=\frac{t(24-7t)}{6(t+12)}, \quad
    a_3=\frac{2(6-t-t^2)}{t+12}, \quad
    a_4=\frac{t(48-7t)}{6(t+12)}.
\end{equation}
For the pair coefficients,
\begin{equation}
\label{eq:sc-b-coefficients}
    b_{12}=\frac{t^2}{3(t+12)}, \quad b_{13}=0, \quad
    b_{14} =0, \quad
    b_{23} =\frac{t(24-19t)}{6(t+12)}, \quad
    b_{24} =\frac{5t^2}{3(t+12)}, \quad
    b_{34} = \frac{t(48-17t)}{6(t+12)}.
\end{equation}
By construction, all coefficients $a_i$ are strictly positive and all $b_{ij}$ are nonnegative for $0 <t \leq1$. Moreover, $
    v(N)=\sum_{i=1}^4 a_i -\sum_{1\leq i< i' \leq 4} b_{i i'} = 1$. We now verify that $a_i - \sum_{i' \in N \setminus \{i \} } b_{i i'} \geq 0$ for every $i \in N$. Using the definition of these coefficients in~\eqref{eq:sc-a-coefficients} and~\eqref{eq:sc-b-coefficients}, we have
\begin{equation*}
\begin{aligned}
    a_1 - b_{12} - b_{1 3} - b_{1 4}
        &=\frac{3t}{t+12},\\
    a_2-b_{12}-b_{23}-b_{24}
        &=0,\\
    a_3- b_{1 3} - b_{23} -b_{34}
        &=\frac{2(2t-3)(t-2)}{t+12},\\
    a_4 - b_{1 4} - b_{24} - b_{34}
        &=0.
\end{aligned}
\end{equation*}
All these quantities are nonnegative when $t \in (0, 1]$. \Cref{lem:sc-quadratic-shapley} therefore implies that $v$ is monotone and submodular. By construction, $v(\emptyset)=0$, so $v$ is also normalized. Moreover, every nonempty coalition has positive value: if $i\in S$, monotonicity and the fact that $a_i>0$ give
$v(S)\geq v(\{i\})=a_i>0$.

The analysis is split into two cases, depending on whether agent $3$ is part of the coalition or not.

\paragraph{Coalitions that omit agent $3$.}
First, we consider a coalition $S \subseteq N$ such that $3 \notin S$. We let $T=S\cup\{3\}$. By~\Cref{lem:sc-quadratic-shapley} and the nonnegativity of $b_{i i'}$ for any $i, i' \in N$, we have $\phi_3(T)\geq \phi_3(N)$. Moreover,
monotonicity and the fact that $v(N) = 1$ imply that $v(T) \leq 1$. As a result,
\begin{align*}
    \frac{\phi_3(T)}{v(T)}-\frac{c_3}{B}
    &\geq \phi_3(N)-\frac{t}{5}\\
    &=\frac{t^2-8t+12}{t+12}-\frac{t}{5}\\
    &=\frac{4(t^2-13t+15)}{5(t+12)}>0.
\end{align*}
By~\Cref{lem:sc-deviation-sign}, agent $3$ has a strict incentive to opt in at every coalition that omits it, so no such coalition is a pure Nash equilibrium.

\paragraph{Coalitions that contain agent $3$.} We now consider coalitions that contain agent $3$. We set
\[
    \epsilon\coloneqq \frac{t^2}{12(t+12)}>0.
\]
We claim that every one of the eight coalitions containing agent $3$ has a strict unilateral deviation. The following table lists a deviating agent together with the deviation sign per~\Cref{lem:sc-deviation-sign}. For an opt-in deviation, the displayed quantity is evaluated at the coalition after the agent joins, while for an opt-out deviation, it is evaluated at the current coalition.

\begin{center}
\renewcommand{\arraystretch}{1.85}
\setlength{\extrarowheight}{1.5pt}
\begin{tabular}{@{}lll@{}}
\toprule
Current coalition $S$ & Strict deviation & Deviation sign (\Cref{lem:sc-deviation-sign}) \\
\midrule
$\{3\}$
    & agent $1$ opts in
    & $\dfrac{\phi_1(\{1,3\})}{v(\{1,3\})}-\dfrac{c_1}{B}
        =\dfrac{t^2(5t+1)}{4(36+3t-5t^2)}>0$ \\
$\{1,3\}$
    & agent $2$ opts in
    & $\dfrac{\phi_2(\{1,2,3\})}{v(\{1,2,3\})}-\dfrac{c_2}{B}
        =\epsilon>0$ \\
$\{2,3\}$
    & agent $4$ opts in
    & $\dfrac{\phi_4(\{2,3,4\})}{v(\{2,3,4\})}-\dfrac{c_4}{B}
        =\dfrac{t^2}{24(6-t)}>0$ \\
$\{3,4\}$
    & agent $1$ opts in
    & $\dfrac{\phi_1(\{1,3,4\})}{v(\{1,3,4\})}-\dfrac{c_1}{B}
        =\epsilon>0$ \\
$\{1,2,3\}$
    & agent $1$ opts out
    & $\dfrac{\phi_1(\{1,2,3\})}{v(\{1,2,3\})}-\dfrac{c_1}{B}
        =-\epsilon<0$ \\
$\{1,3,4\}$
    & agent $4$ opts out
    & $\dfrac{\phi_4(\{1,3,4\})}{v(\{1,3,4\})}-\dfrac{c_4}{B}
        =-\epsilon<0$ \\
$\{2,3,4\}$
    & agent $2$ opts out
    & $\dfrac{\phi_2(\{2,3,4\})}{v(\{2,3,4\})}-\dfrac{c_2}{B}
        =-\dfrac{t^2}{24(6-t)}<0$ \\
$N$
    & agent $1$ opts out
    & $\dfrac{\phi_1(N)}{v(N)}-\dfrac{c_1}{B}
        =-\epsilon<0$ \\
\bottomrule
\end{tabular}
\end{center}

We include below the intermediate calculations that led to the above table. \Cref{lem:sc-quadratic-shapley} and direct substitution give the following values.
\begin{equation*}
    v(\{1,3\}) =\frac{36+3t-5t^2}{3(t+12)}, \quad
v(\{1,2,3\}) =1, \quad
v(\{2,3,4\}) =\frac{2(6-t)}{t+12}, \quad
v(\{1,3,4\}) =v(N)=1.
\end{equation*}
For the Shapley values, we have
{\small
\begin{align*}
\phi_1(\{1,3\}) &= \frac{t(t+9)}{3(t+12)}, \;
\phi_1(\{1,2,3\}) =\frac{t(t+18)}{6(t+12)}, \;
\phi_1(\{1,3,4\}) = \frac{t(t+9)}{3(t+12)}, \;
\phi_1(N) =\frac{t(t+18)}{6(t+12)}, \\
\phi_2(\{1,2,3\}) &= \frac{t(t+8)}{4(t+12)}, \;
\phi_2(\{2,3,4\}) =\frac{t(24-5t)}{12(t+12)}, \;
\phi_4(\{2,3,4\}) =\frac{t(48-7t)}{12(t+12)}, \;
\phi_4(\{1,3,4\}) =\frac{t(t+16)}{4(t+12)}. \;
\end{align*}
}
We conclude that~\Cref{lem:sc-deviation-sign} rules out every coalition containing agent $3$. As a result, we have excluded all possible $2^4$ action profiles, so the game has no pure Nash equilibrium.
\end{proof}

\begin{corollary}[A strict better-response cycle]
\label{cor:sc-cycle}
Every instance in the family of instances constructed in~\Cref{thm:shapley-no-pne} contains the strict better-response cycle
\[
\{1,3\}
\xrightarrow{\;2\text{ in}\;}
\{1,2,3\}
\xrightarrow{\;1\text{ out}\;}
\{2,3\}
\xrightarrow{\;4\text{ in}\;}
\{2,3,4\}
\xrightarrow{\;2\text{ out}\;}
\{3,4\}
\xrightarrow{\;1\text{ in}\;}
\{1,3,4\}
\xrightarrow{\;4\text{ out}\;}
\{1,3\}.
\]
In particular, the game admits no ordinal potential.
\end{corollary}

\begin{proof}
The six better-response transitions are verified (due to~\Cref{lem:sc-deviation-sign}), in order, by
\begin{align*}
\frac{\phi_2(\{1,2,3\})}{v(\{1,2,3\})}&>\frac{c_2}{B},
&\frac{\phi_1(\{1,2,3\})}{v(\{1,2,3\})}&<\frac{c_1}{B},
&\frac{\phi_4(\{2,3,4\})}{v(\{2,3,4\})}&>\frac{c_4}{B},\\
\frac{\phi_2(\{2,3,4\})}{v(\{2,3,4\})}&<\frac{c_2}{B},
&\frac{\phi_1(\{1,3,4\})}{v(\{1,3,4\})}&>\frac{c_1}{B},
&\frac{\phi_4(\{1,3,4\})}{v(\{1,3,4\})}&<\frac{c_4}{B}.
\end{align*}
These inequalities were established in the proof of~\Cref{thm:shapley-no-pne}, so the claim follows.
\end{proof}



\subsection{Complexity of pure Nash equilibria}

We next examine the complexity of computing pure Nash equilibria in $\gameMC$.

\subsubsection{Exponential better-response paths}

\label{sec:convergence-time}

\Cref{thm:existence} shows that every sequence of strict better responses terminates at a pure Nash equilibrium. Here, we show that the number of steps required for convergence is exponential in the number of agents.

The proof reduces from the standard \textsc{FLIP} dynamics---whereby exactly one node changes side---for weighted \textsc{MaxCut}. Specifically, here we are given a weighted graph $G=(V,E,w)$ with nonnegative integer edge weights. We let
\begin{equation}
    \label{eq:cutvalue}
    \cutvalue(X)=\sum_{\{i,j\}\in E}w_{ij}\mathbbm{1}\{|\{i,j\}\cap X|=1\}
\end{equation}
be the total weight of a cut $(X, V \setminus X)$, where $X \subseteq V$. A \textsc{FLIP} step flips one vertex if doing so strictly increases the cut weight. We will use the fact that weighted \textsc{MaxCut} admits instances and initial cuts with exponentially long \textsc{FLIP} improvement paths~\citep{Monien10:Power,Schaffer91:Simple}.

Our reduction proceeds as follows. For every vertex $u \in V$, we let $d_u = \sum_{v: \{u, v\} \in E} w_{u v}$, and choose $L \geq \max_{u \in V} d_u$. We create one agent for each vertex in $V$ and one additional agent $a$. For every $S \subseteq \{a \} \cup V$, the value function is defined as
\begin{equation}
    \label{eq:valuefun}
    v(S) = \frac{L}{\gamma}\mathbbm{1}\{a\in S\}
        + L|S\cap V|+\cutvalue(S\cap V),
\end{equation}
where $\gamma \in (0, 1)$. In terms of the costs, we set $c_a = 0$ and $c_i = \gamma$ for all $i \in V$. We first formalize that the value function defined in~\eqref{eq:valuefun} is indeed monotone and submodular. (Of course, $v(\emptyset) = 0$.)

\begin{lemma}
\label{lem:cut-valuation}
$v : 2^{ V \cup \{ a \} } \to \R_{\geq 0}$ defined in~\eqref{eq:valuefun} is monotone and submodular.
\end{lemma}

\begin{proof}
For a single edge $e=\{u, v\}$, we define
\[
        h_e(X)=\mathbbm{1}\{|\{u, v\}\cap X|=1\}.
\]
We verify that $h_e$ is submodular. Let $X \subseteq X' \subseteq V$ and $i \notin X'$. If $i \notin e$, then $h_e( X \cup \{i\} ) = h_e(X)$ and $h_e( X' \cup \{i\} ) = h_e(X')$. If $i = u$, then
\[
    h_e( X \cup \{i\} ) - h_e(X) =
    \begin{cases}
    1 & \text{if } v \notin X,\\
    -1 & \text{if } v \in X.
    \end{cases}
\]
and
\[
    h_e( X' \cup \{i\} ) - h_e(X') =
    \begin{cases}
    1 & \text{if } v \notin X',\\
    -1 & \text{if } v \in X'.
    \end{cases}
\]
If $h_e( X \cup \{i\} ) - h_e(X) = -1$, then $v \in X$, which in turn implies $v \in X'$. So, $h_e( X' \cup \{i\} ) - h_e(X') = -1 = h_e( X \cup \{i\} ) - h_e(X)$. In other words, the marginal gain can only decrease as $X$ grows. The case where $i = v$ is symmetric. We conclude that $h_e$ is submodular. As a result, the cut function $\cutvalue$ is also submodular since it is a nonnegative weighted sum of submodular functions. Moreover, the terms $ \frac{L}{\gamma} \mathbbm{1}\{a\in S\}$ and $L|S\cap V|$ are modular, so $v$ is indeed submodular as a sum of submodular functions.

It remains to verify monotonicity. Adding the agent $a$ can clearly never decrease $v$ since $L \geq 0$ and $\gamma > 0$. If $i \in V \setminus X$, then
\[
        \cutvalue(X\cup\{i\})-\cutvalue(X) 
        \ge -d_i,
\]
so adding $i$ changes $v$ by at least $L - d_i \ge 0$. This shows that $v$ is monotone, concluding the proof.
\end{proof}

We are now ready to establish the existence of exponential better-response paths.

\begin{theorem}[Exponential better-response paths]
\label{theorem:exp-br-lower}
There are instances of $\gameMC$ with a monotone submodular valuation function $v$ and $c_i<B$ for every $i\in N$ for which there exists a strict better-response path (from some initial state) that has length $2^{\Omega(n)}$.
\end{theorem}

In fact, this result holds even when $\lambda = \max_{i \in N} c_i/B \ll 1$, so the large-market assumption does not preclude this lower bound.

\begin{proof}[Proof of~\Cref{theorem:exp-br-lower}]
We let $B=1$. Let $G=(V,E,w)$ be such a weighted \textsc{MaxCut} instance with a strict \textsc{FLIP} path
\[
        X^{(0)},X^{(1)},\ldots,X^{(T)},
\]
where $T=2^{\Omega(|V|)}$ and each $X^{(t)}\subseteq V$ is one side of the cut. In particular, consecutive sets differ in exactly one vertex.

For a set $S=\{a\}\cup X$, with $k=|X|$, the first coordinate of the potential from the proof of~\Cref{thm:existence} is
\[
        v(\{a\}\cup X)\prod_{i \in X}(1-\gamma)
        =
        \left(\frac{L}{\gamma}+Lk+\cutvalue(X)\right)(1-\gamma)^k.
\]

Now, consider a \textsc{FLIP} step from $X$ to $X'=X \triangle\{i\}$ such that
\[
        \cutvalue(X')=\cutvalue(X)+\delta \text{ for some } \delta\ge 1.
\]
If $i\notin X$, then the corresponding agent opts in. The change in the first coordinate of the potential is
\begin{align*}
(1-\gamma)^{k+1} \left(\frac{L}{\gamma}+L(k+1)+\cutvalue(X)+\delta\right)- &(1 - \gamma)^k \left(\frac{L}{\gamma}+Lk+\cutvalue(X)\right) \\
&= (1 - \gamma)^k \left( (1-\gamma)\delta-\gamma\bigl(L(k+1)+\cutvalue(X)\bigr) \right).
\end{align*}
Therefore, this opt-in deviation strictly increases the potential---equivalently, strictly increases that agent's utility---whenever
\[
        \delta>\frac{\gamma}{1-\gamma}\bigl(L(k+1)+\cutvalue(X)\bigr).
\]
Conversely, if $i\in X$, then the corresponding agent opts out. The change in the first coordinate of the potential is
\begin{align*}
(1-\gamma)^{k-1} \left(\frac{L}{\gamma}+L(k-1) + \cutvalue(X)+\delta\right)
&-(1-\gamma)^{k} \left(\frac{L}{\gamma}+L k +\cutvalue(X)\right) \\
&= (1 - \gamma)^{k-1} \left( \delta+\gamma\bigl(Lk+\cutvalue(X)\bigr) \right) > 0.
\end{align*}
As a result, every cut-improving opt-out deviation strictly increases the potential function. Combining, if we choose $\gamma$ small enough so that
\[
        \frac{\gamma}{1-\gamma}
        \left(L(|V|+1)+\max_{X\subseteq V}\cutvalue(X)\right)<1,
\]
then every opt-in and opt-out deviation along the \textsc{FLIP} path strictly increases the potential. By the proof of~\Cref{thm:existence}, each such strict potential increase corresponds to a strict profitable switch in $\gameMC$. That is,
\[
        S^{(t)}=\{a\}\cup X^{(t)},
        \text{ for } t=0,1,\ldots,T,
\]
is a strict better-response path in $\gameMC$. The constructed game has $n=|V|+1$ agents, so the path length is $T=2^{\Omega(n)}$. Since $\lambda = \gamma$ and $\gamma$ can be chosen arbitrarily small, the lower bound holds even under an arbitrarily strong large-market condition.
\end{proof}

\subsubsection{\PLS-completeness}

Taking a step further, we show that computing a pure Nash equilibrium in $\gameMC$ is \PLS-complete. \PLS, which stands for polynomial local search, was introduced by~\citet{Johnson88:Easy}. 

\begin{theorem}[\PLS-completeness]
\label{thm:pls-complete}
Computing a pure Nash equilibrium of $\gameMC$ is \PLS-complete, even for monotone submodular valuation functions and costs satisfying $c_i<B$ for every $i\in N$.
\end{theorem}

\begin{proof}
Membership in $\PLS$ follows readily from the potential argument in the proof of~\Cref{thm:existence}.\footnote{This hinges on standard complexity assumptions concerning the representation of the value function, which we do not spell out here.} The feasible solutions are the subsets of $N$, the neighborhood corresponds to unilateral agent deviations, and a local optimum is a pure Nash equilibrium.

For the hardness, we reduce from the \PLS-complete problem of local \textsc{MaxCut} under the \textsc{FLIP} neighborhood~\citep{Schaffer91:Simple}. Given such an instance $G=(V,E,w)$, we use the value function in~\eqref{eq:valuefun} with $B=1$ and $\gamma>0$ small enough so that
\[
        \frac{\gamma}{1-\gamma}
        \left(L(|V|+1)+\max_{X\subseteq V}\cutvalue(X)\right)<1.
\]
For example, it suffices to use the upper bound $\max_{X \subseteq V} \cutvalue(X)\le \sum_{e\in E} w_e$ when choosing $\gamma$, so it can be represented with polynomially many bits.

We show that pure Nash equilibria of $\gameMC = \gameMC(G)$ induce locally optimal cuts. First, every pure Nash equilibrium contains the anchor $a$. Indeed, if $a \notin S$, then adding $a$ strictly increases the value, so opting in is strictly profitable for agent $a$ since $c_a = 0$. Next, we consider any set $S=\{a\}\cup X$ and a vertex $i\in V$. In the proof of~\Cref{theorem:exp-br-lower}, we established the following conditions. If $i \notin X$, then opting in is strictly profitable if and only if
\[
         \cutvalue(X\triangle\{i\})-\cutvalue(X) >
        \frac{\gamma}{1-\gamma}\bigl(L(|X|+1)+\cutvalue(X)\bigr).
\]
If $i \in X$, then opting out is strictly profitable if and only if
\[
        \cutvalue(X\triangle\{i\})-\cutvalue(X) + \gamma\bigl(L|X|+\cutvalue(X)\bigr) > 0.
\]
By our choice of $\gamma$, together with the fact that we have a pure Nash equilibrium of $\gameMC$, we have
\[
    \cutvalue(X\triangle\{i\})-\cutvalue(X) < 1.
\]
Since the weights are assumed to be integers, $\cutvalue(X\triangle\{i\})-\cutvalue(X)$ is also an integer, thereby implying $\cutvalue(X\triangle\{i\})-\cutvalue(X) \leq 0$ for any $i \in V$. So, flipping any vertex cannot increase the weight of the cut.
\end{proof}

\subsubsection{Pseudo-polynomial bound for additive valuation functions}
\label{subsec:additive-convergence}

We complement the foregoing hardness results by establishing a positive result when the value function is additive, meaning that $v(S)=\sum_{i\in S} v_i$, where $v_i \ge 0$. In this case, assuming $v(S)>0$, the marginal contribution payment rule becomes
\[
        \pMC_i(S) = B \frac{v_i}{v(S)}.
\]
In what follows, it will be convenient to define
\[
        \theta_i=\frac{Bv_i}{c_i}-v_i
        =
        v_i\frac{B-c_i}{c_i}>0,
\]
where we can assume that $c_i \in (0, B)$. If $i \notin S$, then $i$ strictly prefers to opt in if and only if
\[
        B \frac{v_i}{v(S)+v_i} > c_i
        \iff
        v(S) < \theta_i.
\]
Conversely, if $i\in S$, then $i$ strictly prefers to opt out if and only if
\[
        B\frac{v_i}{v(S)} < c_i
        \iff
        v(S) - v_i > \theta_i.
\]
Therefore, the pure Nash equilibria are exactly the sets $S$ satisfying
\[
        v(S)\ge \theta_i \quad \forall i\notin S
        \text{ and }
        v(S)-v_i\le \theta_i \quad \forall i\in S.
\]
For the case of additive valuations, we can derive the following potential function.

\begin{proposition}[Potential for additive valuations]
\label{prop:additive-potential}
For additive valuation functions, strict better-response dynamics strictly decreases the potential function
\[
        \Phi(S)
        =
        \sum_{i < i'} v_i v_{i'} \mathbbm{1}\{a_i = a_{i'} \}
        +
        \sum_{i=1}^n v_i ( W_{-i}-2\theta_i)a_i,
\]
where $a_i=\mathbbm{1}\{i\in S\}$ and $W_{-i}=\sum_{i' \ne i}v_{i'}$.
\end{proposition}

\begin{proof}
First, let us suppose that $i$ opts in. The first part of $\Phi$ changes by
\[
        v_i\left(
         \sum_{i' \neq i} v_{i'} \mathbbm{1} \{ i' \in S \} - \sum_{i' \neq i} v_{i'} \mathbbm{1} \{ i' \notin S \}
        \right)
        =
        v_i \left(2 \sum_{i' \ne i} v_{i'} \mathbbm{1}\{ i' \in S\} - W_{-i} \right),
\]
while the second part of $\Phi$ changes by $v_i(W_{-i} - 2 \theta_i)$. Thus,
\[
        \Phi(S\cup\{i\})-\Phi(S)
        =
        2v_i\left(\sum_{j\ne i}v_j\mathbbm{1}\{j\in S\}-\theta_i\right).
\]
A strict opt-in deviation occurs if and only if $\sum_{i' \ne i} v_{i'} \mathbbm{1} \{ i' \in S\} < \theta_i$, so the potential strictly decreases. Symmetrically, if $i$ decides to opt out, we have
\[
        \Phi(S\setminus\{i\})-\Phi(S)
        =
        -2v_i \left(\sum_{ i' \ne i} v_{i'} \mathbbm{1}\{i' \in S\}-\theta_i \right).
\]
A strict opt-out deviation occurs if and only if $\sum_{i' \ne i} v_{i'} \mathbbm{1}\{i' \in S\} > \theta_i$, so the potential again decreases.
\end{proof}

Now, let us assume that each value $v_i$ is integral and $W \defeq \sum_{i \in N} v_i$. We claim that~\Cref{prop:additive-potential} implies the following pseudo-polynomial upper bound on the length of any strict better-response path.

\begin{corollary}
    \label{cor:fastconvergence}
    Let $v$ be additive. If $v_i \in \N$ and $c_i \in (0, B)$ for each $i \in N$, every strict better-response path has length $O(W^2)$, where $W \defeq \sum_{i \in N} v_i$.
\end{corollary}

As a result, while general monotone submodular valuation functions can admit exponentially long better-response paths, additivity with polynomially bounded values yields a polynomial bound. A key aspect of~\Cref{cor:fastconvergence} is that the bound on the length of the better-response path does not depend on the costs.

\begin{proof}[Proof of~\Cref{cor:fastconvergence}]
    Every term $\sum_{i' \ne i}v_{i'} \mathbbm{1}\{i' \in S\}$ is an integer in $[0,W]$. Further, each threshold $\theta_i$ can be replaced by an equivalent threshold $\tilde \theta_i\in[-1,W+1]$ that is either an integer or a half-integer and preserves the sign of $\sum_{i' \ne i} v_{i'} \mathbbm{1} \{i' \in S\} -\theta_i$ over all integer values of $\sum_{i' \ne i} v_{i'} \mathbbm{1}\{ i' \in S\}$. This replacement does not change any strict better-response decision. With these equivalent thresholds, every strict step decreases $\Phi$ by at least $1$. Moreover, $\max_{S \subseteq N} \Phi(S)-\min_{S \subseteq N} \Phi(S)=O(W^2)$, so the claim follows.
\end{proof}

The case where $c_i = 0$ is not interesting: if $v_i > 0$, then agent $i$ has a strict dominant strategy to opt in, whereas if $v_i = 0$ that agent is strategically irrelevant.

\section{Price of anarchy bounds}

In this section, we derive price of anarchy bounds for both pure Nash equilibria and coarse correlated equilibria, primarily under the large-market condition $\lambda = \max_{i \in N} c_i / B \ll 1$.

\subsection{Pure Nash equilibria}

To begin with, we focus on pure Nash equilibria. Our bound will depend on the following \emph{cost concentration} term:
\begin{equation}
        \costcon(\lambda)=
        \max\left\{
        \sum_{i \in S } \frac{x_i }{1 - x_i}:
        0 \le x_i \le \lambda \text{ and } \sum_i x_i \le 1 \text{ for a finite } S
        \right\}.                                      \label{eq:costcon-def}
\end{equation}
An immediate upper bound is $\costcon(\lambda) \leq \frac{1}{1 - \lambda}$. More precisely, we can determine $\costcon(\lambda)$ exactly as follows. For $\lambda\in(0,1)$, let
\[
        q = \left\lfloor\frac{1}{\lambda}\right\rfloor
        \text{ and }
        r=1-q\lambda\in[0,\lambda).
\]
It is easy to see that $\costcon(\lambda)$ can be expressed in the following form.
\begin{restatable}{claim}{lamsplit}
    For $\costcon(\lambda)$ defined in~\eqref{eq:costcon-def}, we have
    \begin{equation}
        \costcon(\lambda)=q\frac{\lambda}{1-
        \lambda}+\mathbbm{1} \{r>0\} \frac{r}{1-r}.            \label{eq:costcon-formula}
\end{equation}
\end{restatable}

Indeed, the function $x\mapsto x/(1-x)$ is increasing and convex on $[0,1)$, so the maximum in~\eqref{eq:costcon-def} is attained by filling as many coordinates as possible (namely, $q$) at the upper bound $\lambda$, plus one remainder coordinate when $q < \frac{1}{\lambda}$. We formalize this in~\Cref{appendix:proofs}.

\begin{theorem}[Pure Nash equilibrium PoA bound]
\label{thm:poa-pne}
Let $\lambda = \max_{i \in N} c_i / B$. The game $\gameMC$ induced by the marginal contribution payment rule~\eqref{eq:mc-rule} satisfies
\[
        \PoAPNE \leq 1 + \costcon(\lambda) \leq \frac{2-
        \lambda}{1-
        \lambda}.
\]
In particular, as $\lambda\to 0$, every pure Nash equilibrium attains value at least $(1/2-o(1))\OPT$.\footnote{As we point out in~\Cref{appendix:marginalcovering}, this PoA bound can be extended to the class of $\gamma$-marginal-covering rules (\Cref{def:marg-cov}); however, a pure Nash equilibrium is not guaranteed to exist for that broader class of payment rules (\emph{cf.}~\Cref{thm:shapley-no-pne}).}
\end{theorem}

\begin{proof}
Let $S$ be a pure Nash equilibrium of $\gameMC$ and let $S^* \in\argmax\{v(T) : \sum_{i \in T} c_i \le B\}$. If $S = N$, it follows that $v(S) = \OPT$. Otherwise, let $i \notin S$. If $v(S) = 0$, it follows that $\OPT = 0$, so we can assume $v(S) > 0$. Since $i$ does not prefer to opt in, we have
\[
B\frac{ v(S\cup\{i\}) - v(S) }{v(S \cup \{i \} )} \le c_i.
\]
Rearranging and using the fact that $c_i < B$, we get 
\begin{equation}
    \label{eq:marg-value}
     v( S \cup \{i \} ) \leq v(S) \frac{B}{B - c_i} \Rightarrow v( S \cup \{i \} ) - v(S) \leq v(S) \frac{c_i}{B - c_i}.
\end{equation}
Now, let $x_i \defeq c_i / B$. Continuing from~\eqref{eq:marg-value},
\begin{equation}
    \label{eq:xbound}
    v( S \cup \{i \} ) - v(S) \leq \frac{x_i}{1 - x_i} v(S).
\end{equation}
By monotonicity and submodularity,
\[
        v(S^*) \le v(S \cup S^*)
        \le v(S) + \sum_{i \in S^* \setminus S} \bigl(v(S \cup \{i\}) - v(S) \bigr).
\]
Using~\eqref{eq:xbound},
\begin{equation}
    \label{eq:xcombbound}
    v(S^*)
        \le
        v(S)\left(1+
        \sum_{i \in S^* \setminus S}\frac{x_i}{1-x_i}\right).
\end{equation}
The vector $(x_i)_{i \in S^* \setminus S}$ satisfies $0 \le x_i \le \lambda$ and $\sum_{i \in S^* \setminus S} x_i \le \sum_{i \in S^* } x_i \le 1$ because $S^*$ is budget feasible (by definition). Therefore, combining~\eqref{eq:xcombbound} with the definition of $\costcon(\lambda)$ in~\eqref{eq:costcon-def}, we conclude that
\[
        v(S^*) \le \bigl(1+\costcon(\lambda)\bigr) v(S),
\]
and the claim follows.
\end{proof}

As a result, under a monotone and submodular valuation function, the marginal contribution payment rule guarantees a PoA of $1 + \costcon(\lambda) \leq \frac{2 - \lambda}{1 - \lambda}$, thereby approaching $2$ when $\lambda \to 0$.

\subsubsection{Tightness of the analysis}

We now show that~\Cref{thm:poa-pne} is tight even when the value function is additive: for any $\lambda \in (0, 1)$, there exists an instance where $\PoAPNE \geq 1 + \costcon(\lambda)$.

\begin{proposition}[Tightness of~\Cref{thm:poa-pne}]
\label{prop:tight}
For every $\lambda\in(0,1)$, there is an additive instance with $\max_{i \in N} c_i/B \le \lambda$ and a pure Nash equilibrium $S$ such that
\[
        \OPT = (1 + \costcon(\lambda)) v(S).
\]
\end{proposition}

\begin{proof}
Fix any $\lambda \in (0, 1)$. We let $B=1$, $q=\lfloor 1/\lambda\rfloor$, and $r=1-q\lambda$. We construct the following instance. First, there is an agent $a$ with cost $c_a = 0$ and value $v_a = 1$. Moreover, there are $q$ additional agents each of whom has a cost $\lambda$ and additive value $\lambda/(1 - \lambda)$. Similarly,
if $r>0$, create one additional agent with cost $r$ and additive value $r/(1 - r)$. We let $N$ denote the set of agents in this instance. The value function is additive, so that $v(S) = \sum_{i \in S} v_i$ for any $S \subseteq N$.

We first claim that the optimal solution selects all agents. This is indeed budget feasible since $\sum_{i \in N} c_i = q \lambda + r = 1 = B$, by definition of $r$. Moreover,
\begin{equation}
    \label{eq:opt-everyone}
     \OPT = v(N) = 1 + q \frac{\lambda}{1 - \lambda} + \mathbbm{1} \{ r > 0 \} \frac{r}{1 - r} = 1 + \costcon(\lambda).
\end{equation}
As a result,
\[
    \pMC_i (S \cup \{i\}) = B \frac{ v(S \cup \{i\}) - v(S)}{ v(S \cup \{i\})} = \frac{x/(1 - x)}{1/(1 - x)} = x = c_i,
\]
which means that opting in yields a utility of zero. That is, opting in is not a strict improvement. Similar reasoning applies with respect to the agent with cost $r$ (if that agent is part of the instance). We conclude that $S = \{a \}$ is a pure Nash equilibrium, and $v(S) = 1$. Combining with~\eqref{eq:opt-everyone}, this completes the proof.
\end{proof}

In the construction above outside agents are indifferent between opting in and opting out. One can make this lower bound agnostic to how ties are handled by making the value $x/(1 - x) - \eta$ for some arbitrarily small $\eta > 0$, where $x \in \{\lambda, r \}$.

\subsection{Lower bounds}

We complement our positive results with several lower bounds. First, we show that any cost-oblivious payment rule must incur a PoA of at least $2$ even when $\lambda \to 0$.

\subsubsection{Anonymous and cost-oblivious rules}

As a warm-up, we first treat the case where the payment rule is anonymous. We then extend the lower bound even when the payment rule is not anonymous (\Cref{thm:cost-oblivious-without-anonymity-lb}).

\begin{proposition}[Lower bound for anonymous and cost-oblivious rules]\label{thm:anonymous-lb}
No payment rule that is i) anonymous, ii) nonnegative, iii) budget feasible, and iv) cost oblivious can guarantee $\PoAPNE$ below $2$ as $\lambda \to 0$. This holds even when the value function is additive.
\end{proposition}

\begin{proof}
Consider any $k \in \N$ and any sufficiently small $\eta > 0$. Let $\lambda = \frac{1}{k+1} + \eta$. We construct an instance as follows. The budget is $B = 1$. There are $2k$ agents each with an additive value of $1$. We partition the $2k$ agents into
\[
        A=\{a_1,\ldots,a_k\}
        \text{ and }
        H=\{h_1,\ldots,h_k\}.
\]
For any $i \in [k]$, we define the costs as
\[
        c_{a_i}=0
        \text{ and }
        c_{h_i}=\frac{1}{k+1}+\eta.
\]
When $\eta < 1/(k(k+1))$, we have
\[
        \sum_{i=1}^k c_{h_i} = \frac{k}{k+1}+k\eta<1.
\]
So, the principal can afford to include all agents, for a total value of $2k$ (the value function is additive).

On the other hand, consider the set $S = A$. Each agent in $S$ has zero cost and incurs a nonnegative payment, so there is no incentive to opt out. If we consider the set $A \cup \{h_i \}$ that results from a deviation, each of the $k+1$ agents is identical from the perspective of a cost-oblivious rule. As a result, by anonymity, they must all receive an equal payment. By budget feasibility, each will receive at most $1/(k+1)$. Given that $c_{h_i} > 1/(k+1)$ for any $i \in [k]$, opting out is a best response for every such agent. As a result, $A$ is a pure Nash equilibrium, and $v(A) = k$. This concludes the proof.
\end{proof}

\subsubsection{Cost-oblivious rules without anonymity}
\label{sec:lower-bound-without-anonymity}

The preceding lower bound makes crucial use of anonymity. We now show that the barrier of $2$ also applies without anonymity.

\begin{theorem}[Lower bound without anonymity]
\label{thm:cost-oblivious-without-anonymity-lb}
Fix any $\lambda\in(0,1)$. Consider any deterministic payment rule that is nonnegative, budget feasible, cost oblivious, and guarantees the existence of a pure Nash equilibrium for every cost vector satisfying $\max_{i \in N} c_i/B \leq \lambda$. Then there is an instance with a normalized monotone additive value function $v$ and a cost vector satisfying $\max_{i \in N} c_i/B \leq \lambda$ for which the induced game has a pure Nash equilibrium $S$ such that $\OPT / v(S) \geq 2$. As a result, no such payment rule can guarantee $\PoAPNE < 2$.
\end{theorem}

\begin{proof}
We let
\[
    k=\left\lfloor\frac{1}{\lambda}\right\rfloor
    \text{ and }
    n=2k+1.
\]
We consider an instance with an agent set $N$ of cardinality $n$. The (additive) value function is defined as
\[
    v(S)=\frac{|S|}{n}
    \qquad\forall S \subseteq N.
\]
This function is clearly normalized, monotone,
and additive. First, we consider the cost vector
\[
    \bar c_i=\lambda B
    \quad \forall i\in N.
\]
By the assumed existence guarantee, the induced game has a pure Nash equilibrium $\bar S$. Every selected agent $i \in \bar S$ must satisfy $p_i(\bar S)\geq \lambda B$. Budget feasibility then gives
\[
    |\bar S|\lambda B
    \leq
    \sum_{i\in\bar S}p_i(\bar S)
    \leq B,
\]
so $|\bar S|\leq k$. For every agent $i \notin \bar{S}$, $p_i(\bar S\cup\{i \}) \leq \lambda B$. Since $|\bar S| \leq k$, at least $k+1$ agents lie outside $\bar S$. Let $T \subseteq N\setminus\bar S$ be any set such that $|T| = k$. We have
\[
    \sum_{i \in T} p_i(\bar{S} \cup \{i \}) \leq k \lambda B \leq B.
\]
We let
\[
    \eta = 
    \min\left\{
        \frac{B - \sum_{i \in T} p_i(\bar{S} \cup \{i \}) }{|\bar S|+k},
        \min_{i \in T}(\lambda B- p_i(\bar{S} \cup \{i \}) ),
        \lambda B
    \right\} \geq 0,
\]
and define a new cost vector $c$ by
\[
    c_i =
    \begin{cases}
        \eta & \text{if } i\in\bar S,\\
        p_i(\bar{S} \cup \{i \}) + \eta & \text{if } i\in T,\\
        \lambda B & \text{if } i\in N\setminus(\bar S\cup T).
    \end{cases}
\]
All costs are nonnegative and $\max_{i \in N} c_i/B=\lambda$. Moreover, the set $\bar S\cup T$ is budget feasible under $c$:
\[
    \sum_{i\in\bar S\cup T} c_i
    =
    |\bar S|\eta+\sum_{i\in T}( p_i(\bar{S} \cup \{i \}) +\eta)
    =
    \sum_{i \in T} p_i(\bar{S} \cup \{i \}) + (|\bar S|+k) \eta
    \leq B.
\]
Because the rule is cost oblivious, all payments under $c$ are exactly the same as under $\bar c$. We claim that $\bar S$ is a pure Nash equilibrium under the new costs. Indeed, for every $i\in\bar S$,
\[
    p_i(\bar S)-c_i
    =
    p_i(\bar S)-\eta
    \geq
    \lambda B-\eta
    \geq0,
\]
so selected agents do not prefer to opt out. For every $i \in T$,
\[
    p_i(\bar S\cup\{i\})-c_i
    =
    -\eta
    \leq 0,
\]
so agents in $T$ prefer to remain out. Finally, if
$i \in N \setminus(\bar S \cup T)$, then $c_i=\lambda B$ and
\[
    p_i(\bar S\cup\{ i \})- c_i = p_i(\bar S\cup\{ i \}) - \lambda B \leq 0.
\]
This shows that $\bar S$ is indeed a pure Nash equilibrium. What remains is to compare the value of $\bar S$ to the optimal one. Since $\bar S\cup T$ is feasible,
\[
    \OPT\geq v(\bar S\cup T)=\frac{|\bar S|+k}{n},
    \text{ whereas }
    v(\bar S)=\frac{|\bar S|}{n}.
\]
If $\bar S=\emptyset$, then $\OPT\geq k/n>0$ while $v(\bar S)=0$, so the ratio is $+\infty$. Otherwise, given that $1 \leq|\bar S|\leq k$, we have
\[
    \frac{\OPT}{v(\bar S)}
    \geq
    \frac{|\bar S|+k}{|\bar S|}
    =
    1+\frac{k}{|\bar S|}
    \geq2.
\]
This completes the proof.
\end{proof}

This result implies the following dichotomy. For every deterministic payment rule that is nonnegative, budget feasible, and cost oblivious, either it does not always guarantee existence of a pure Nash equilibrium, or there exists an instance with a pure Nash equilibrium whose welfare ratio is at least $2$.

\subsubsection{Beyond the large-market assumption}

We now turn our attention to settings beyond the large-market condition. \Cref{thm:poa-pne} provides a PoA bound that diverges when $\lambda \uparrow 1$. We observe that this is, in some sense, fundamental:

\begin{proposition}[Infinite PoA when $\lambda=1$]
\label{prop:lambda-one}
If $\lambda = \max_{i \in N} c_i/B = 1$, then every budget-feasible payment rule can have $\PoAPNE = +\infty$.
\end{proposition}

\begin{proof}
We consider a single agent with $c_1 = B$ and $v(\{1\})=1$. In this case, $\OPT = 1$. However, $S = \emptyset$ is in fact a pure Nash equilibrium. Indeed, if the agent were to join, budget feasibility dictates that the payment to that agent is at most $B = c_1$, so the utility of the agent in that case is at most $0$.
\end{proof}

This lower bound hinges on a high-value agent being indifferent between opting in and opting out. We next show that even if ties are broken in favor of the principal (all else being equal, the agent selects the action maximizing the principal's value), the price of anarchy of any deterministic payment rule is $\Omega(\sqrt{n})$.

\paragraph{An $\Omega(\sqrt{n})$ lower bound for principal-favoring tie breaking.} Unlike~\Cref{prop:lambda-one}, the lower bound below does not rely on adversarial tie breaking.

\begin{theorem}[Deterministic cost-oblivious lower bound]\label{thm:lower-deterministic-no-large-market}
For every $n \ge 3$ and every deterministic cost-oblivious budget-feasible payment rule $p$, there is an instance with $n$ agents and an additive value function $v$ such that the induced game has a strict pure Nash equilibrium $S$ satisfying
\[
  \frac{\OPT}{v(S)} \geq \sqrt{n - 1}.
\]
\end{theorem}

\begin{proof}
We fix $n\ge 3$. We write the agent set as
\[
  N=\{h,1,2,\ldots, n-1 \}.
\]
In what follows, $h$ will be the high-value agent. In particular, the value function is defined as
\[
  v(S)= \sqrt{n-1} \cdot \mathbbm{1} \{h\in S\}+|S \cap [n-1]|.
\]
This is an additive function with $v_h = \sqrt{n-1}$ and $v_i = 1$ for $i \in N \setminus \{h\}$. We further define $B = 1$.

We split the analysis into the following cases.

\paragraph{Case 1: some singleton is not paid the full budget.}
First, suppose that there exists an agent $i\in N$ such that $p_i(\{i\})<1$. Let $c_i\in (p_i(\{i\}),1]$ and set $c_{i'} = 2$ for each $i' \neq i$. Now, consider the set $S = \emptyset$. Since $p_i(\{i\}) - c_i < 0$, $i$ has no incentive to deviate. Moreover, if an agent $i' \neq i$ were to opt in, budget feasibility implies $p_{i'}( \{ i' \} ) \leq 1 < 2 = c_{i'}$, so this again results in a strictly worse outcome for that agent. This shows that $S = \emptyset$ is a pure Nash equilibrium, no matter how ties are handled.

However, $c_i \leq 1 = B$, so $\OPT > 0$ and the price of anarchy is $+\infty$. In other words, any deterministic cost-oblivious payment rule with finite price of anarchy must satisfy
\[
  p_i(\{i\})=1 \quad \forall i \in N.
\]
\paragraph{Case 2: every singleton is paid the full budget.} In view of the previous argument, for the rest of the proof, we can assume that
\[
  p_i(\{i\})=1\qquad\forall i\in N.
\]
We analyze outcomes that contain two agents.

\subparagraph{Case 2a: the high-value agent is not paid the full budget in some pair.}
Suppose that there exists $i \in [n-1]$ such that $p_h( \{h, i \} ) < 1$. Let $c_h\in (p_h(\{h,i\}),1]$ and $c_{i} = \epsilon$ for some $\epsilon \in (0, 1)$. Moreover, $c_{i'} = 2$ for all $i' \in [n-1] \setminus \{i \}$.

We claim that $S = \{ i \}$ is a pure Nash equilibrium, no matter how ties are broken. Indeed, $p_i( \{ i \} ) - c_i = 1 - \epsilon > 0$, so $i$ strictly prefers opting in. Moreover, agent $h$ does not have an incentive to join because $p_h( \{h, i \} ) - c_h < 0$, by definition of $c_h$. Finally, every other agent $i' \in [n-1] \setminus \{i \}$ does not want to join because its payment in any outcome is at most $1$ (the budget), whereas $c_{i'} = 2$.

The value in this equilibrium is $v( \{i \} ) = 1$. At the same time, $\{h\}$ is feasible because $c_h \leq 1$, and $v( \{h\} ) = \sqrt{n-1}$. This shows a lower bound of $\sqrt{n-1}$ on the price of anarchy.

\subparagraph{Case 2b: the high-value agent is paid the full budget in every pair.} In the remaining case, we assume that
\[
  p_h(\{h, i \}) = 1 \quad \forall i \in [n-1].
\]
By budget feasibility and nonnegativity of payments, this implies
\[
  p_i (\{h,i\})=0 \quad \forall i \in [n-1].
\]
We pick $\epsilon \in\left(0,\frac{1}{n}\right]$ and set
\[
  c_h=c_1=\cdots=c_{n-1}= \epsilon.
\]
This means that the entire set $N$ maintains budget feasibility since $ \sum_{i \in N} c_i = n \epsilon \leq 1 = B$. Yet, we claim that $S=\{h\}$ is a Nash equilibrium no matter how ties are broken. Agent $h$ prefers to opt in since $p_h( \{ h \} ) - c_h = 1 - \epsilon > 0$. For any other agent $i \in [n-1]$, we have
\[
  p_i(\{h,i\})-c_i = 0 - \epsilon < 0.
\]
This proves that $\{ h \}$ is indeed a pure Nash equilibrium, and has value $v( \{h\} ) = \sqrt{n-1}$. On the other hand, the entire set $N$, which is budget feasible, has value $v(N) = \sqrt{n-1} + n-1$. Therefore,
\[
    \frac{\OPT}{v(\{h\})} = \frac{ n - 1 + \sqrt{n-1} }{\sqrt{n-1}} \geq \sqrt{n - 1}.
\]
This completes the proof.
\end{proof}

\subsection{Constant PoA beyond the large-market assumption via randomization}
\label{sec:randomized-no-large-market}
In this subsection, we assume platform-favoring tie-breaking but work without any large-market assumption: agents may have costs as large as or larger than the budget. Despite the $\Omega(\sqrt{n})$ lower bound for deterministic payment rules (\Cref{thm:lower-deterministic-no-large-market}), we show that one can recover a constant PoA using a randomized payment rule. A randomized rule is a distribution over deterministic payment rules, and a deterministic payment rule is sampled and announced
before the agents choose whether to opt in. The performance guarantee is with
respect to the expectation over randomization and the worst
principal-favoring pure Nash equilibrium in each realized branch.


We consider a randomized rule over two deterministic branches. The first branch is a singleton-prize rule. For every nonempty $S\subseteq N$,
let $w(S)\in\argmax_{i\in S} v(\{i\})$ be selected according to a fixed public priority order. We define
\[
        p^{\mathrm{sing}}_i(S)
        =
        \begin{cases}
        B & \text{if }S\neq\emptyset\text{ and }i=w(S),\\
        0 & \text{otherwise.}
        \end{cases}
\]
The second branch is the marginal-contribution rule $\pMC$ from
\eqref{eq:mc-rule}. Both branches are budget feasible; the
budget feasibility of $\pMC$ is exactly \Cref{prop:budget}.

We first show that pure Nash equilibria exist in both branches. We note that the potential function argument in \Cref{thm:existence} requires $c_i < B$ for every agent and does not directly generalize. Nevertheless, we use a refined argument to show the existence of a pure Nash equilibrium.
\begin{proposition}[Existence of equilibria in the two branches]
\label{prop:no-large-market-existence}
Both branches above admit a principal-favoring pure Nash equilibrium.
\end{proposition}

\begin{proof}
For the singleton branch, if no agent is affordable (\emph{i.e.,} $c_i > B$ for all $i \in N$), then $\emptyset$ is an
equilibrium. Otherwise, let us choose among affordable agents with maximum singleton
value, the one with the highest priority, and call it $h$. The set $\{h\}$ is an
equilibrium: agent $h$ receives $B$ and does not want to leave, every affordable
agent who opts in loses the prize and receives zero, and every unaffordable agent
who opts in either receives zero or receives $B<c_i$.

For the marginal branch, first restrict the game to agents with $c_i<B$. The potential argument from \Cref{thm:existence}
with the following refinement
\[
        \Phi''(S):= \left(
        v(S)\prod_{i\in S}\left(1-\frac{c_i}{B}\right),
        v(S),
        -|S|
        \right),
\]
strictly increases along every strict better response and every
principal-favoring zero-utility move in this restricted game. Hence the
restricted game has a terminal state $S_L$.

If $v(S_L)>0$, then no agent with $c_i\geq B$ wants to enter the full game:
agents with $c_i>B$ cannot receive enough payment, while agents with $c_i=B$ receive strictly less than $B$ when joining a positive-value state. Thus $S_L$
is an equilibrium of the full marginal branch.

It remains to consider the case $v(S_L)=0$. We claim that $S_L = \emptyset$ since otherwise we can remove an agent from $S_L$ and strictly improve the third coordinate of $\Phi''$ while keeping the first two coordinates unchanged by monotonicity, contradicting the fact that $S_L$ is a termination state. Then we also have $v(\{i\}) = 0$ for every agent $i$ with $c_i < B$ since otherwise they could join $S_L = \emptyset$ and strictly improve the potential.
If there is an agent
$h$ with $c_h=B$ and $v(\{h\})>0$, then $\{h\}$ is a full-game equilibrium:
agent $h$ receives $B$ and stays by principal-favoring tie-breaking; every other
agent with cost at least $B$ has no profitable entry; and every agent with
$c_i<B$ has $v(\{i\})=0$, hence has zero marginal contribution to $\{h\}$ by
submodularity. If no such agent exists, then $S_L = \emptyset$ itself is a full-game
equilibrium.
\end{proof}

We then analyze the performance of the pure Nash equilibrium under the two deterministic rules.
Define the best affordable singleton value $M \defeq \max\{v(\{i\}): c_i\leq B\}$ with the convention that $M=0$ if no agent is individually affordable.

\begin{lemma}[Singleton branch]
\label{lem:no-large-market-singleton}
Every principal-favoring pure Nash equilibrium $S$ of the singleton-prize branch satisfies $v(S)\geq M$.
\end{lemma}

\begin{proof}
Suppose for contradiction that $v(S)<M$, and let $i^*$ be an affordable agent
with $v(\{i^*\})=M$. Then $i^*\notin S$, since otherwise monotonicity would give
$v(S)\geq M$. Moreover, for every $j\in S$,
\[
        v(\{j\})\leq v(S)<M=v(\{i^*\}).
\]
Thus, if $i^*$ opts in, it wins the singleton prize and receives $B$. Since
$c_{i^*}\leq B$, opting in gives nonnegative utility. If the utility is strictly
positive, this is a profitable deviation. If the utility is zero, then opting in
strictly increases the principal's value because
\[
        v(S\cup\{i^*\})\geq v(\{i^*\})=M>v(S),
\]
so principal-favoring tie-breaking again makes $i^*$ opt in. This contradicts the equilibrium condition.
\end{proof}

\begin{lemma}[Marginal branch]
\label{lem:no-large-market-marginal}
Let $S$ be a principal-favoring pure Nash equilibrium of the marginal branch,
and let $U\defeq v(S)$. Then
\[
        \OPT\leq M+2U.
\]
\end{lemma}

\begin{proof}
Let $S^*$ be an optimal budget-feasible set. If $\OPT=0$, the claim is immediate,
so assume $\OPT>0$.

We first note that $U>0$. Indeed, if $U=0$, then submodularity and $v(S^*)>0$
imply that some $i\in S^*$ has $v(\{i\})>0$. Since $i\in S^*$, we have $c_i\leq B$.
If $i$ opts in from $S$, it receives the full budget under $\pMC$ and the
principal's value strictly increases. This is either a strictly profitable
deviation or a principal-favoring tie, contradicting equilibrium.

For each $i\in S^*\setminus S$, write
\[
        d_i \defeq v(S\cup\{i\})-v(S)
        \text{ and }
        x_i \defeq \frac{c_i}{B}.
\]
Since $S^*$ is budget feasible, $x_i\in[0,1]$. The fact that $i$ does not opt in
at equilibrium gives
\[
        B\frac{d_i}{U+d_i}\leq c_i.
\]
When $x_i<1$, rearranging yields
\[
        d_i\leq \frac{x_i}{1-x_i}U.
\]
On the other hand, submodularity gives
\[
        d_i\leq v(\{i\})\leq M,
\]
where the last inequality uses $c_i\leq B$. Combining these two bounds gives,
for every $x_i\in[0,1]$,
\begin{equation}
        d_i\leq x_i(M+U).                         \label{eq:no-large-market-di}
\end{equation}
Indeed, for $x_i=1$ this follows from $d_i\leq M$. For $x_i<1$, if
$M\leq x_i(M+U)$ we again use $d_i\leq M$; otherwise
$x_i< M/(M+U)$, which implies
\[
        d_i \le \frac{x_i}{1-x_i}U\leq x_i(M+U).
\]
Finally, by monotonicity and submodularity,
\[
        v(S^*)
        \leq v(S\cup S^*)
        \leq v(S)+\sum_{i\in S^*\setminus S}\bigl(v(S\cup\{i\})-v(S)\bigr).
\]
Using \eqref{eq:no-large-market-di} and the budget feasibility of $S^*$,
\[
        v(S^*)
        \leq
        U+(M+U)\sum_{i\in S^*\setminus S}x_i
        \leq U+(M+U)
        =
        M+2U.
\]
This completes the proof.
\end{proof}

The randomized rule mixes the two branches as follows:
\[
        \begin{array}{ccl}
        \text{singleton-prize branch} & \text{with probability} & 1/3,\\[1mm]
        \text{marginal-contribution branch} & \text{with probability} & 2/3.
        \end{array}
\]

\begin{theorem}[Constant PoA without a large-market assumption]
\label{thm:no-large-market-randomized}
For every normalized monotone submodular value function $v$ and every budget
$B>0$, the randomized rule
above satisfies
\[
        \E[v(S)]\geq \frac{1}{3}\OPT
\]
under every branch-contingent principal-favoring pure Nash equilibrium.
The expected pure price of anarchy of this rule is therefore at most $3$.
\end{theorem}

\begin{proof}
Let $S^{\mathrm{sing}}$ be any principal-favoring pure Nash equilibrium in the
singleton branch, and let $S^{\mathrm{mc}}$ be any principal-favoring pure Nash
equilibrium in the marginal branch. Write $U=v(S^{\mathrm{mc}})$. By
\Cref{lem:no-large-market-singleton},
\[
        v(S^{\mathrm{sing}})\geq M,
\]
while \Cref{lem:no-large-market-marginal} gives
\[
        \OPT\leq M+2U.
\]
Therefore, the expected value of the randomized rule is at least
\[
        \frac{1}{3}v(S^{\mathrm{sing}})
        +
        \frac{2}{3}v(S^{\mathrm{mc}})
        \geq
        \frac{1}{3}M+\frac{2}{3}U
        =
        \frac{1}{3}(M+2U)
        \geq
        \frac{1}{3}\OPT. \qedhere
\]
\end{proof}





\subsection{Coarse correlated equilibria}
\label{sec:cce}

In light of the intractability of pure Nash equilibria (\Cref{thm:pls-complete}), an important question is whether coarse correlated equilibria also attain a constant fraction of the optimal value. We first give a simple constant-factor argument and then sharpen it to the asymptotically tight large-market bound.

In what follows, for $i \notin S$, it will be convenient to introduce the notation
\[
        x_i(S)=
        \begin{cases}
        \dfrac{v(S\cup\{i\})-v(S)}{v(S\cup\{i\})} & \text{if } v(S\cup\{i\})>0,\\
        0  & \text{if } v(S\cup\{i\})=0.
        \end{cases}
\]
In particular, if an agent $i$ opts in and the selected set is $S$, its payment under our marginal contribution rule will be $B x_i(S)$. We are now ready to obtain a PoA bound for CCEs.

\begin{theorem}[PoA bound for CCEs]
\label{thm:cce-submodular}
Let $\lambda=\max_{i \in N} c_i/B<1$. For any coarse correlated equilibrium $\mu \in \Delta(2^N)$ of $\gameMC$,
\[
        \OPT
        \le
        \left(2+\frac{1}{1-\lambda}\right)\mathbb E_{S\sim\mu}[v(S)].
\]
As a result,
\[
        \PoACCE
        \le
        2+\frac{1}{1-\lambda}.
\]
\end{theorem}

\begin{proof}
Let
\[
        \bar{v}=\mathbb E_{S\sim\mu}[v(S)].
\]
First, suppose that $\bar{v}=0$. Then $v(S)=0$ for any set $S$ in the support of $\mu$. If there is an agent $i$ with $v(\{i\}) > 0$, then $i \notin S$ for any set $S$ in the support of $\mu$. Moreover, opting in would give a payment of $B$ since $x_i(S) = 1$ for any $S$ in the support of $\mu$. Given that $B > c_i$, this contradicts the CCE condition. As a result, $v(\{i\})=0$ for every $i\in N$, which in turn implies that $\OPT = 0$. In the remainder of the proof, we can assume $\bar{v} > 0$.

For every agent $i\in N$, the CCE condition for the deviation in which $i$ always opts in gives
\begin{equation}
    \label{eq:optinginCCE}
    \mathbb E_{S \sim \mu} \left[
        \mathbbm{1}\{i\notin S\}
        \bigl(Bx_i(S)-c_i\bigr)
        \right]
        \le 0.
\end{equation}
Equivalently,
\begin{equation}
        \mathbb E_{S \sim \mu} \left[
        \mathbbm{1}\{i\notin S\}x_i(S)
        \right]
        \le
        \frac{c_i}{B}\Pr_{ S \sim \mu}[i\notin S].
        \label{eq:cce-no-entry-y}
\end{equation}

We next upper bound the singleton values. For an agent $i\in N$ let $v_i=v(\{i\})$. We claim that $v_i \le \bar{v}/(1-\lambda)$. The claim is trivial when $v_i = 0$. Now suppose $v_i>0$. In the event $i \notin S$, monotonicity gives $v(S\cup\{i\})\ge v_i$, so
\[
        x_i(S)
        =1-\frac{v(S)}{v(S\cup\{i\})}
        \ge
        1-\frac{v(S)}{v_i}.
\]
Combining this with~\eqref{eq:cce-no-entry-y}, we obtain
\[
        \Pr_{S \sim \mu}[i\notin S]-\frac{\mathbb E_{S \sim \mu} [\mathbbm{1}\{i\notin S\}v(S)]}{v_i}
        \le
        \frac{c_i}{B}\Pr_{S \sim \mu}[i\notin S].
\]
Therefore,
\begin{equation}
        \mathbb E_{S \sim \mu}[\mathbbm{1}\{i\notin S\}v(S)]
        \ge
        \left(1-\frac{c_i}{B}\right)\Pr_{ S \sim \mu }[i\notin S]v_i.
        \label{eq:absent-utility-lower}
\end{equation}
On the other hand, if $i\in S$, monotonicity gives $v(S) \ge v_i$, so
\begin{equation}
        \mathbb E_{S \sim \mu}[\mathbbm{1}\{i\in S\}v(S)]
        \ge
        \Pr_{ S \sim \mu }[i\in S]v_i.
        \label{eq:present-utility-lower}
\end{equation}
Adding~\eqref{eq:absent-utility-lower} and~\eqref{eq:present-utility-lower}, we get
\[
        \bar{v}
        \ge
        \left(\left(1-\frac{c_i}{B}\right)\Pr_{ S \sim \mu }[i\notin S]+\Pr_{ S \sim \mu }[i\in S]\right)v_i
        =
        \left(1-\frac{c_i}{B}\Pr_{S \sim \mu}[i\notin S]\right)v_i
        \ge
        \left(1-\frac{c_i}{B}\right)v_i.
\]
Rearranging,
\[
v_i
        \le
        \frac{\bar{v}}{1-c_i/B}
        \le
        \frac{\bar{v}}{1-\lambda}.
\]
By submodularity, it follows that for every set $S$ with $i\notin S$,
\[
        v(S\cup\{i\})-v(S)\le v(\{i\})=v_i,
\]
so
\begin{equation}
        v(S\cup\{i\})-v(S)\le \frac{\bar{v}}{1-\lambda}
        \quad\text{for every } i\notin S.
        \label{eq:uniform-marginal-bound}
\end{equation}
Now, let $S^* \in \argmax\{v(S) : \sum_{i\in S} c_i/B \le 1\}$. Summing~\eqref{eq:cce-no-entry-y} over $i\in S^*$ gives
\begin{equation}
        \mathbb E_{S \sim \mu} \left[
        \sum_{i\in S^*\setminus S}x_i(S)
        \right]
        \le
        \sum_{i\in S^*} \frac{c_i}{B}\Pr_{S \sim \mu}[i\notin S]
        \le
        \sum_{i\in S^*}\frac{c_i}{B}
        \le 1.
        \label{eq:sum-y-upper}
\end{equation}
For every set $S$, monotonicity and submodularity imply
\[
        \OPT
        =v(S^*)
        \le
        v(S\cup S^*)
        \le
        v(S)+\sum_{i\in S^*\setminus S}\bigl(v(S\cup\{i\})-v(S)\bigr).
\]
Therefore, for any $S$,
\begin{equation}
        \sum_{i\in S^*\setminus S}\bigl(v(S\cup\{i\})-v(S)\bigr)
        \ge \max\{0, \OPT - v(S)\} \eqqcolon
        [\OPT-v(S)]_+.
        \label{eq:deficit-lower}
\end{equation}
Using~\eqref{eq:uniform-marginal-bound}, for each $i\in S^*\setminus S$ we have
\[
        x_i(S)
        =
        \frac{v(S\cup\{i\})-v(S)}{v(S\cup\{i\})}
        \ge
        \frac{v(S\cup\{i\})-v(S)}{v(S)+M},
\]
where we defined $M \defeq \bar{v}/(1 - \lambda)$. Combining with~\eqref{eq:deficit-lower},
\begin{equation}
        \sum_{i\in S^*\setminus S}x_i(S)
        \ge
        \frac{[\OPT-v(S)]_+}{v(S)+M}.
        \label{eq:y-deficit-lower}
\end{equation}
Taking the expectation over $S \sim \mu$ in~\eqref{eq:y-deficit-lower} and using~\eqref{eq:sum-y-upper}, we obtain
\begin{equation}
        \mathbb E_{S \sim \mu} \left[
        \frac{[\OPT-v(S)]_+}{v(S)+M}
        \right]
        \le 1.
        \label{eq:key-cce-expectation}
\end{equation}
We now consider the function
\[
        h(t)=\frac{[\OPT - t]_+}{t + M}
        \quad t \ge 0.
\]
This function is convex on $[0, +\infty)$, so Jensen's inequality yields
\[
    \E_{S \sim \mu} [h( v(S) )] \geq h\left( \E_{S \sim \mu} v(S) \right) = h(\bar{v}).
\]
Combining this with~\eqref{eq:key-cce-expectation},
\[
    \frac{[ \OPT - \bar{v} ]_+}{ \bar{v} + M } = h(\bar{v}) \leq \E_{S \sim \mu} [h( v(S) )] = \mathbb E_{S \sim \mu} \left[
        \frac{[\OPT-v(S)]_+}{v(S)+M}
        \right] \leq 1.
\]
If $\bar{v} > \OPT$, the claimed bound is immediate. Otherwise, using the fact that $M = \bar{v}/(1 - \lambda)$,
\[
    \OPT \leq 2 \bar{v} + M \le \left( 2 + \frac{1}{1 - \lambda} \right) \bar{v} = \left( 2 + \frac{1}{1 - \lambda} \right) \E_{S \sim \mu} [ v(S) ],
\]
as claimed.
\end{proof}

\begin{corollary}[CCE PoA bound for $1$-marginal-covering rules]
\label{cor:cce-marginal-covering}
Let $y$ be the shares of a $1$-marginal-covering budget-share rule, so the induced payments are $p_i(S)=B y_i(S)$.
Let $\lambda=\max_{i\in N}c_i/B<1$. For any coarse correlated equilibrium $\mu$ of the game induced by this payment rule,
\[
        \OPT
        \le
        \left(2+\frac{1}{1-\lambda}\right)\mathbb E_{S\sim\mu}[v(S)].
\]
In particular, the same bound holds for the normalized Shapley payment rule.
\end{corollary}

\begin{proof}
We use the notation $x_i(S)$ from the proof of~\Cref{thm:cce-submodular}. If $\bar v=\mathbb E_{S\sim\mu}[v(S)]=0$ and $v(\{i\})>0$ for some agent $i$, then $i\notin S$ for every $S$ in the support of $\mu$ and $x_i(S)=1$. The marginal-covering condition implies $y_i(S\cup\{i\})\ge1$, so the deviation in which $i$ always opts in gives utility at least $B-c_i>0$, contradicting the CCE condition. Thus, $\OPT=0$ in the zero-value case.

Now suppose $\bar v>0$. For every $i$, the CCE condition for the deviation in which $i$ always opts in gives
\[
        \mathbb E_{S\sim\mu}\left[
        \mathbbm{1}\{i\notin S\}
        \bigl(B y_i(S\cup\{i\})-c_i\bigr)
        \right]
        \le 0.
\]
The marginal-covering condition gives $y_i(S\cup\{i\})\ge x_i(S)$ for every $S$ with $i\notin S$. Thus, \eqref{eq:cce-no-entry-y} continues to hold, and the remainder of the proof of~\Cref{thm:cce-submodular} applies directly. Finally, the normalized Shapley rule is $1$-marginal-covering due to~\Cref{lem:shapley-efficiency-covering}.
\end{proof}

\paragraph{A sharp large-market bound.} We next sharpen the argument by using a fixed marginal decomposition of the optimal set rather than bounding every optimal singleton. The result applies to the marginal-contribution rule, and more generally to any $1$-marginal-covering budget-share rule in the sense of~\Cref{def:marg-cov}.

\begin{theorem}[Improved CCE PoA bound]
\label{thm:cce-submodular-sharp}
Let $v$ be monotone and submodular, and let $\mu$ be a coarse correlated equilibrium of a $1$-marginal-covering budget-share rule. Then, if $\lambda < 1/2$,
\[
        \OPT
        \le
        \frac{2}{1-2\lambda}\mathbb E_{S\sim\mu}[v(S)].
\]
As a result,
\[
        \PoACCE
        \le
        \frac{2}{1-2\lambda}=2+o_\lambda(1).
\]
In particular, this bound applies to $\gameMC$ and to normalized Shapley payments.
\end{theorem}

\begin{proof}
In the proof, it will be convenient to extend the notation $x_i(S)$ by setting $x_i(S)=0$ whenever
$i\in S$. For a budget-share rule satisfying the $1$-marginal-covering condition, the CCE constraint with respect to the deviation in which
agent $i$ always opts in yields
\begin{equation}
        \E_{S\sim\mu}[x_i(S)]
        \le
        \frac{c_i}{B}\Pr_{S\sim\mu}[i\notin S]
        \le
        \frac{c_i}{B}.
        \label{eq:cce-entry-improved}
\end{equation}

Let
$S^*\in \argmax\{ v(S): \sum_{i \in S} c_i \le B\}$. We order the agents in $S^*$ in an arbitrary order as
$i_1,\ldots,i_m$, where $m = |S^*|$, and define
\begin{equation}
        w_{i_k}
        \defeq
        v(\{i_1,\ldots,i_k \}) - v(\{i_1,\ldots,i_{k-1}\}),
        \quad k = 1, \ldots, m.
        \label{eq:cce-greedy-weights}
\end{equation}
We write $\bar v\defeq \E_{S\sim\mu}[v(S)]$. By monotonicity, $w_i \ge0$ for every $i\in S^*$. Moreover, telescoping gives
\begin{equation}
        \sum_{i\in S^*} w_i
        = v(S^*) = \OPT.
        \label{eq:cce-weights-sum}
\end{equation}
For every $S\subseteq N$ and $i\in S^*$, we let $t_i(S)\defeq \min\{w_i, v(S\cup\{i\})-v(S) \}$. We first claim that, for every $S\subseteq N$,
\begin{equation}
        \sum_{i\in S^*}t_i(S)
        \ge
        \OPT-v(S).
        \label{eq:cce-greedy-mass}
\end{equation}
Let $P_k = \{i_1,\ldots,i_k \}$ and $P_0=\emptyset$. If $i_k \in S$, then $v(S\cup P_k)-v(S\cup P_{k-1}) = 0 = v(S\cup\{i_k\})-v(S)$. Otherwise,
submodularity gives
\[
        v(S\cup P_k)-v(S\cup P_{k-1})
        \le
        v(P_k)-v(P_{k-1})
        =
        w_{i_k}
\]
and
\[
        v(S\cup P_k)-v(S\cup P_{k-1})
        \le
        v(S\cup\{i_k\})-v(S).
\]
Combining, this means that $v(S\cup P_k)-v(S\cup P_{k-1}) \leq t_{i_k}(S)$ for every $k$. Summing over $k$ and telescoping,
\[
        \sum_{i\in S^*} t_i(S)
        =
        \sum_{k=1}^m t_{i_k}(S)
        \ge
        \sum_{k=1}^m \bigl(v(S\cup P_k)-v(S\cup P_{k-1})\bigr)
        =
        v(S\cup S^*)-v(S)
        \ge
        v(S^*)-v(S),
\]
which proves~\eqref{eq:cce-greedy-mass}. Next, we define
\[
        b_i(S)
        \defeq
        \begin{cases}
        \dfrac{t_i(S)}{v(S)+w_i} & \text{if } v(S)+w_i>0,\\[1.5ex]
        0 & \text{if } v(S)+w_i=0.
        \end{cases}
\]
We now claim that
\begin{equation}
        0\le b_i(S)\le x_i(S)
        \text{ and }
        t_i(S)=\bigl(v(S)+w_i\bigr)b_i(S).
        \label{eq:cce-truncated-share}
\end{equation}
The only non-immediate claim is the fact that $b_i(S) \leq x_i(S)$. Since $t_i(S) \le v(S \cup \{i \}) - v(S)$ and $t_i(S) \le w_i$, we have
\[
    v(S) (v(S \cup \{i \}) - v(S) - t_i(S) ) + (v(S \cup \{i \}) - v(S)) (w_i - t_i(S) ) \ge 0.
\]
Rearranging,
\[
        t_i(S) (v(S) + v(S \cup \{i \}) - v(S) )
        \le
        (v(S \cup \{i \}) - v(S)) ( v(S) + w_i ),
\]
which shows that $b_i(S) \leq x_i(S)$. Using~\eqref{eq:cce-entry-improved}, \eqref{eq:cce-truncated-share}, and the
budget feasibility of $S^*$,
\begin{equation}
        \E_{S \sim \mu} \left[ \sum_{i\in S^*}b_i(S) \right]
        =
        \sum_{i\in S^*} \E_{S \sim \mu}[b_i(S)]
        \le
        \sum_{i\in S^*} \E_{S \sim \mu}[x_i(S)]
        \le
        \sum_{i\in S^*}\frac{c_i}{B}
        \le
        1.
        \label{eq:cce-beta-bound}
\end{equation}
Similarly,
\begin{equation}
        \E_{S \sim \mu} \left[ \sum_{i\in S^*}w_i b_i(S) \right]
        =
        \sum_{i\in S^*}w_i\E_{S \sim \mu}[b_i(S)]
        \le 
        \sum_{i\in S^*}w_i\E_{S \sim \mu}[x_i(S)]
        \le
        \sum_{i\in S^*}w_i\frac{c_i}{B}.
        \label{eq:cce-rho-bound}
\end{equation}
Combining~\eqref{eq:cce-greedy-mass} and~\eqref{eq:cce-truncated-share}, every
realized set $S$ satisfies
\begin{equation}
        \OPT - v(S)
        \le
        \sum_{i\in S^*} t_i(S)
        =
        v(S) \sum_{i\in S^*}b_i(S) + \sum_{i\in S^*}w_i b_i(S).
        \label{eq:cce-pointwise-core}
\end{equation}
From~\eqref{eq:cce-pointwise-core},
\[
        \OPT - \left( v(S) + \sum_{i\in S^*}w_i b_i(S) \right)
        \le
        v(S) \sum_{i\in S^*}b_i(S)
        \le
        \left( v(S) + \sum_{i\in S^*} w_i b_i(S) \right) \sum_{i\in S^*}b_i(S).
\]
If $\OPT=0$, the theorem is immediate. Otherwise, it follows from~\eqref{eq:cce-pointwise-core} that $v(S) + \sum_{i\in S^*}w_i b_i(S) > 0$ for every realized $S$. As a result,
\[
        \sum_{i\in S^*}b_i(S)
        \ge
        h\left( v(S) + \sum_{i\in S^*}w_i b_i(S) \right),
        \text{ where }
        h(t) \defeq \frac{[\OPT-t]_+}{t}.
\]
The function $h$ is nonincreasing and convex on $(0,\infty)$. So, Jensen's inequality together with~\eqref{eq:cce-beta-bound} imply
\[
        1
        \ge
        \E\left[ \sum_{i\in S^*}b_i(S) \right]
        \ge
        \E\left[h\left( v(S) + \sum_{i\in S^*}w_i b_i(S) \right)\right]
        \ge
        h\left(\E\left[v(S) + \sum_{i\in S^*}w_i b_i(S)\right]\right).
\]
Moreover, by~\eqref{eq:cce-rho-bound},
\[
        \E\left[v(S) + \sum_{i\in S^*}w_i b_i(S)\right]
        \le
        \bar v + \sum_{i\in S^*} w_i\frac{c_i}{B}.
\]
Since $h$ is nonincreasing,
\[
        1
        \ge
        h\left(\bar v+\sum_{i\in S^*}w_i\frac{c_i}{B}\right).
\]
For $t>0$, the inequality $h(t)\le1$ is equivalent to $t \ge \OPT/2$. Hence, we conclude that
\[
        \bar v+\sum_{i\in S^*}w_i\frac{c_i}{B}
        \ge
        \frac{\OPT}{2},
\]
Finally, $c_i/B\le\lambda$ for every $i$, so
\[
        \sum_{i\in S^*}w_i\frac{c_i}{B}
        \le
        \lambda\sum_{i\in S^*}w_i
        =
        \lambda\OPT
\]
by~\eqref{eq:cce-weights-sum}. Rearranging yields the claim.
\end{proof}

\subsubsection{Additive valuation functions}

Finally, for the special case of additive valuation functions, we can recover the exact bound established for pure Nash equilibria.

\begin{theorem}[CCE PoA bound for additive valuations]
\label{thm:cce-additive}
Suppose that $v$ is additive and $\lambda<1$. Any coarse correlated equilibrium $\mu$ of $\gameMC$ satisfies
\[
        \OPT
        \le
        \bigl(1+\costcon(\lambda)\bigr)
        \mathbb E_{S\sim\mu}[v(S)].
\]
As a result,
\[
        \PoACCE
        \le
        1+\costcon(\lambda).
\]
\end{theorem}

This bound is tight under the marginal contribution payment rule in view of the tight pure Nash equilibrium instance constructed in~\Cref{prop:tight} (pure Nash equilibria are, of course, coarse correlated equilibria).

\begin{proof}[Proof of~\Cref{thm:cce-additive}]
As before, let $\bar{v}=\mathbb E_{S\sim\mu}[v(S)]$ and $S^*$ be an optimal budget-feasible set. We claim that, for every agent $i\in N$,
\begin{equation}
        \Pr_{S \sim \mu}[i\notin S]v_i
        \le
        \frac{c_i/B}{1-c_i/B}
        \mathbb E_{ S \sim \mu }[\mathbbm{1}\{i\notin S\}v(S)].
        \label{eq:additive-missing-value-bound}
\end{equation}
If $v_i=0$ or $\Pr_{S \sim \mu}[i\notin S]=0$, the claim is obvious, so we can assume that $v_i>0$ and $\Pr_{S \sim \mu}[i\notin S] > 0$. For $i\notin S$, additivity gives
\[
        x_i(S)=\frac{v_i}{v(S)+v_i}.
\]
As a result, by~\eqref{eq:optinginCCE},
\begin{equation}
        \mathbb E_{S \sim \mu} \left[
        \mathbbm{1}\{i\notin S\}
        \left(B\frac{v_i}{v(S)+v_i}-c_i\right)
        \right]
        \le 0.
        \label{eq:additive-cce-no-entry}
\end{equation}
Conditioning~\eqref{eq:additive-cce-no-entry} on the event $i\notin S$, we get
\[
        \mathbb E_{S \sim \mu} \left[
        \left.\frac{v_i}{v(S)+v_i}\,\right|\,i\notin S
        \right]
        \le \frac{c_i}{B}.
\]
The function $z\mapsto v_i/(z+v_i)$ is convex for $z\ge0$. By Jensen's inequality,
\[
        \frac{v_i}{\mathbb E_{S \sim \mu} [v(S)\mid i\notin S]+v_i}
        \le
        \mathbb E_{S \sim \mu} \left[
        \left.\frac{v_i}{v(S)+v_i}\,\right|\,i\notin S
        \right]
        \le \frac{c_i}{B}.
\]
Rearranging gives
\[
        v_i
        \le
        \frac{c_i/B}{1-c_i/B}\mathbb E_{S \sim \mu}[v(S)\mid i\notin S],
\]
and multiplying by $\Pr_{S \sim \mu}[i\notin S]$ proves~\eqref{eq:additive-missing-value-bound}. Now, using additivity,
\[
        \OPT
        =\sum_{i\in S^*}v_i = \E_{S \sim \mu} \left[ \sum_{i\in S^*}v_i \right]
        =
        \mathbb E\left[\sum_{i\in S^*\cap S}v_i\right]
        +
        \sum_{i\in S^*}\Pr[i\notin S]v_i.
\]
The first term in the right-hand side is at most $\bar{v}$. By~\eqref{eq:additive-missing-value-bound}, the second term is at most
\[
        \sum_{i\in S^*}\frac{c_i/B}{1-c_i/B}
        \mathbb E[\mathbbm{1}\{i\notin S\}v(S)]
        \le
        \bar{v}\sum_{i\in S^*}\frac{c_i/B}{1-c_i/B}.
\]
Since $S^*$ is feasible, $0\le c_i/B\le\lambda$ and $\sum_{i\in S^*}c_i/B\le1$. By the definition of $\costcon(\lambda)$,
\[
        \sum_{i\in S^*}\frac{c_i/B}{1-c_i/B}
        \le
        \costcon(\lambda).
\]
Therefore,
\[
        \OPT
        \le
        \bigl(1+\costcon(\lambda)\bigr)\bar{v}.
\]
\end{proof}

\subsection{On the price of stability of CCEs}
\label{sec:additive-cce-ratio-sharp}

We conclude this section by making a counter-intuitive observation: the \emph{price of stability}---defined similarly to the price of anarchy but with respect to the best equilibrium---of CCEs can be strictly smaller than 1. This means the best CCE can have better performance than the full-information benchmark. In particular, we provide a tight characterization for additive value functions. The reason why this happens is that CCEs may place positive probability on a coalition $S$ that violates the constraint $\sum_{i \in S} c_i \leq B$, even though the payment is budget feasible. Thus, a CCE can approximately implement a \emph{fractional} solution of the budget-constraint optimization problem and beats the full-information integral solution benchmark.

We begin with an explicit construction showing that the price of stability with respect to CCEs can be arbitrarily close to $1/2$!

\begin{theorem}[Two-agent sharpness construction]
\label{thm:additive-cce-sharp-example}
For every $\epsilon>0$, there is a two-agent instance with an additive value function and a coarse
correlated equilibrium $\mu$ of the marginal-contribution game such that
\begin{equation}
        \frac{\OPT}{\E_{S \sim \mu} v(S)} \leq
        \frac12+\epsilon.
        \label{eq:additive-cce-sharp-ratio}
\end{equation}
\end{theorem}

\begin{proof}
We let $B=1$ and choose a parameter $\theta\in\left(\frac12,\min\left\{1,\frac12+\epsilon\right\}\right)$. The instance contains two agents, $N=\{1,2\}$, with
\[
        v_1 = v_2 = 1
        \text{ and }
        c_1=c_2 = \theta.
\]
In other words, $v(S)=|S|$. Under the marginal-contribution rule, an agent $i$ that has opted in
receives payment $1$ if $S = \{i\}$ and payment $1/2$ if $S = N$. Now, we define a distribution $\mu \in \Delta(2^{N})$ as
\[
        \mu(\{1,2\})= \frac{1-\theta}{\theta} ,
        \quad
        \mu(\{1\})=\mu(\{2\})= \frac{2\theta-1}{2\theta},
        \quad
        \mu(\emptyset)=0.
\]
Since $\theta\in(1/2,1)$, this is indeed a valid distribution. We verify the CCE inequalities for agent $1$; the argument for agent $2$ is
symmetric. Agent $1$ obtains utility $1-\theta$ if $S = \{1\}$, utility
$1/2-\theta$ if $S = \{1,2\}$, and utility zero if $S = \{2\}$. Thus,
\begin{align*}
        \E_{S\sim\mu}[u_1(S)]
        =
        \frac{2\theta-1}{2\theta} (1-\theta)+ \frac{1-\theta}{\theta} \left(\frac12-\theta\right) = 0.
\end{align*}
Now, the deviation to always opt out also gives zero utility, so it is not profitable. Moreover, for the deviation to always opt in, the utility is the same in states where the agent is already selected, while for $S = \{2 \}$ the deviation yields a utility of $1/2 - \theta < 0$. Thus, this deviation is not profitable either. We conclude that $\mu$ is a coarse correlated equilibrium.

Each singleton is budget feasible and has value one, while the pair of agents is infeasible because the total cost is $2\theta > 1$. In other words, $\OPT = 1$. The expected value is
\[
        \E_{S\sim\mu}[v(S)] = \mu(\{1, 2\}) v(\{1,2\}) + \mu(\{1\}) v(\{1\}) + \mu(\{2\}) v(\{2\}) =
        2 \frac{1-\theta}{\theta} + 2 \frac{2\theta-1}{2\theta}
        =
        \frac1\theta.
\]
We conclude that
\[
        \frac{\OPT}{\E_{S\sim\mu}[v(S)]} = \theta,
\]
and the choice of $\theta \leq \frac{1}{2} + \epsilon$ gives~\eqref{eq:additive-cce-sharp-ratio}.
\end{proof}

In fact, we shall now show that the construction in~\Cref{thm:additive-cce-sharp-example} is tight for additive valuations. To begin with, we prove the following auxiliary lemma, pointing out that CCEs satisfy cost feasibility \emph{in expectation}.

\begin{lemma}[CCE implies cost feasibility in expectation]
\label{lem:additive-cce-expected-cost}
Let $p$ be any nonnegative budget-feasible payment rule and $\mu$ a
coarse correlated equilibrium of the induced game. Then
\[
        \E_{S\sim\mu}\left[\sum_{i\in S}c_i\right]\le B.
\]
\end{lemma}

\begin{proof}
For every agent $i \in N$, the fixed deviation in which $i$ always opts out gives
zero utility. Therefore, the CCE condition implies
\[
        \E_{S\sim\mu}[u_i(S)]\ge0.
\]
Summing over all agents,
\[
        \E_{S\sim\mu}\left[
        \sum_{i\in S}\bigl(p_i(S)-c_i\bigr)
        \right] \geq 0.
\]
As a result,
\[
        \E_{S\sim\mu}\left[\sum_{i\in S}c_i\right]
        \le
        \E_{S\sim\mu}\left[\sum_{i\in S}p_i(S)\right]
        \le B,
\]
where the last inequality follows from budget feasibility.
\end{proof}

We next point out that individually infeasible agents can never be in the support of a CCE.

\begin{lemma}[Individually infeasible agents are absent]
\label{lem:additive-cce-infeasible-absent}
In the setting of~\Cref{lem:additive-cce-expected-cost}, if $c_i>B$,
then
\[
        \Pr_{S\sim\mu}[i\in S]=0.
\]
\end{lemma}

\begin{proof}
Nonnegativity and budget feasibility imply $p_i(S)\le B$ whenever $i\in S$.
Thus, if $c_i>B$, then $u_i(S)=p_i(S)-c_i\le B-c_i<0$ whenever $i \in S$. If $i$ was selected with positive probability, its expected utility would be strictly negative, contradicting the CCE condition since that agent can always opt out.
\end{proof}

The key idea in the price of stability lower bound is that cost feasibility in expectation corresponds to a \emph{fractional relaxation} of knapsack, which can only be larger by a factor of $2$. For completeness, we include the simple proof below.

\begin{lemma}[Fractional knapsack is at most twice integral knapsack]
\label{lem:additive-cce-knapsack}
Let $I$ be a finite set of elements with values $v_i\ge0$ and costs
$0\le c_i\le B$. If
\[
        F
        =
        \max\left\{
        \sum_{i\in I} v_i x_i:
        \sum_{i\in I} c_i x_i \le B,\ 0\le x_i \le1
        \right\} \text{ and } \OPT
        =
        \max\left\{
        \sum_{i \in S } v_i:
        S \subseteq I,\ \sum_{i\in S} c_i \le B
        \right\},
\]
then $F \le2 \OPT$.
\end{lemma}

\begin{proof}
Let $x^*$ be an optimal extreme point of the fractional program. We claim that among elements
with positive cost, at most one coordinate of $x^*$ can be strictly between zero and
one. Indeed, if $i \ne i'$ satisfy $c_i, c_{i'} > 0$ and
$ 0 < x^*_{i}, x^*_{i'} < 1$, then $x^*$ would be (strictly) inside a line segment of cost-feasible points, contradicting extremality. Furthermore, we can choose $x^*$ so that every element $i$ with zero cost satisfies $x_i^* = 1$. Now, let $S^*=\{i \in I : x_i^*=1\}$. The set $S^*$ is budget feasible, so $\sum_{i\in S^*} v_i \le \OPT$. By virtue of our previous argument, there is at
most one additional element $i$ with $0 < x_i^* < 1$. If such an element exists, then $\{i\}$ is budget feasible because $c_i \le B$, so $v_i \le \OPT$. Combining, we have
\[
        F
        =
        \sum_{i' \in S^*} v_{i'} + x_i^*v_i
        \le
        2 \OPT,
\]
as claimed.
\end{proof}

We are now ready to show that~\Cref{thm:additive-cce-sharp-example} is tight.

\begin{theorem}
\label{thm:additive-cce-half-bound}
Let $v$ be additive and $p$ any nonnegative budget-feasible payment rule. If $\mu$ is a coarse correlated equilibrium of the induced game, then
\[
        \E_{S \sim \mu}[v(S)]
        \le
        2\OPT.
\]
\end{theorem}

\begin{proof}
For every agent $i$, we let $x_i \defeq \Pr_{S\sim\mu}[i\in S]$. By additivity, we have
\[
        \E_{S\sim\mu}[v(S)]
        =
        \sum_{i\in N}v_i x_i.
\]
Moreover, by~\Cref{lem:additive-cce-expected-cost},
\[
        \sum_{i\in N} c_i x_i = \E_{S \sim \mu} \left[ \sum_{i \in S} c_i \right] \le B.
\]
By~\Cref{lem:additive-cce-infeasible-absent}, $x_i=0$ whenever $c_i>B$. Thus,
the vector $(x_i)_{c_i\le B}$ is feasible for the fractional knapsack
relaxation over the individually feasible agents. Combining with~\Cref{lem:additive-cce-knapsack},
\[
        \E_{S \sim \mu}[v(S)]
        \le F
        \le 2 \OPT.
\]
This completes the proof.
\end{proof}

\section{Beyond monotone submodular}
\label{sec:beyond monotone submodular}
A natural question arising from our positive results so far is whether they can be extended to more general valuation functions beyond monotone submodular. Here, we consider XOS and subadditive value functions, both of which are general and include the submodular case as a special case.

\begin{definition}[XOS value function]
    \label{def:XOS}
    A function $v : 2^N \to \R_{\geq 0}$ is XOS (or, equivalently, fractionally subadditive) if there exists a finite set of additive functions $\{a^{(1)}, \dots, a^{(\ell)} \}$ such that $v(S) = \max_{1 \leq j \leq \ell} a^{(j)}(S)$ for any $S \subseteq N$.
\end{definition}
\begin{definition}[Subadditive value function]
    A function $v : 2^N \to \R_{\geq 0}$ is subadditive if $v(S\cup T) \le v(S) + v(T)$ for all $S ,T \subseteq N$.
\end{definition}
We recall the relationship between value functions in the complement-free hierarchy~\citep{Lehmann06:Combinatorial}:
\[
\mathrm{Additive}\; \subseteq \; \mathrm{Submodular}\; \subseteq \; \mathrm{XOS} \; \subseteq \; \mathrm{Subadditive}. 
\]

\subsection{PoA lower bounds for XOS valuations}

This subsection presents an information-theoretic obstruction for fractionally subadditive value functions (equivalently, XOS): even under a large-market condition, no payment rule that uses only polynomially many value queries can guarantee an $O(n^{1/2-\epsilon})$ price of anarchy with respect to CCEs.


As before, we assume a set $N$ comprising $n$ agents. The value function $v$ can be accessed through a value oracle. In what follows, for our lower bound construction, we take all costs to be equal to one; that is, $c_i = 1$ for any $i \in N$. Moreover, the budget is $B = \sqrt{n}$. For convenience, we assume that $n$ is a perfect square. This instance satisfies the large-market assumption. We can write the (offline) optimum as
\[
        \OPT(v) = \max\{v(S): |S|\le B \}.
\]
(Here we write $\OPT(v)$ to make explicit the dependence on $v$, as this will be relevant in our construction.) A payment rule $p$ assigns nonnegative payments $p_i(S)$ to agents within $S$ and zero to agents outside it. It is budget feasible if $\sum_{i\in S} p_i (S)\le B$ for any $S \subseteq N$. In this context, since each cost is taken to be one, the utility of agent $i$ can be expressed as
\[
        u_i(S)=
        \begin{cases}
        p_i(S)-1 & \text{if } i \in S,\\
        0 & \text{if } i\notin S.
        \end{cases}
\]

Our lower bound concerns payment rules that are \emph{oracle efficient}, in the sense of making only $\poly(n)$ value queries.

\begin{definition}[Oracle-efficient payment rule]
    \label{def:efficientpaymentrule}
A payment rule $\mathcal A$ is \emph{oracle efficient} if
\begin{enumerate}[label=(\roman*)]
    \item it produces a payment rule $p : 2^N \to \R^n_{\geq 0}$ by making at most $\poly(n)$ value queries;
    \item for every set $S \subseteq N$ and agent $i\in N$, the payment $p_i(S)$ can be determined using at most $\poly(n)$ additional value queries to $v$; and
    \item the resulting payment rule is budget feasible.
\end{enumerate}
The payment rule is allowed arbitrary computation.
\end{definition}

In fact, our lower bound applies even if the rule is allowed to depend on the private costs (all costs in the lower-bound instance are equal to $1$).

We say that $\mathcal A$ guarantees $\PoACCE \leq \rho(n)$ if any $\eta$-CCE $\mu$ of the induced game satisfies
\[
        \E_{S\sim\mu}[v(S)]\ge \frac{1}{\rho(n)} \OPT
\]
for any small enough $\eta = \poly(1/n)$.

\subsubsection{The hard XOS instance}

We rely on the lower bound of~\citet{Mirrokni08:Tight} concerning XOS valuations in the value-query model. For a small enough $\delta > 0$, we let 
\[
    h = \left\lfloor (1+\delta)n^{2\delta}\right\rfloor \text{ and } \beta=(1+\delta)n^{-1/2+\delta}.
\]
The base value function is defined as
\[
        v_0 (S) \defeq
        \max\left\{
        \max_{S' : |S'| \le h} |S \cap S'|,\ \beta |S|
        \right\}
        =
        \max\{\min\{|S|,h\},\ \beta |S|\}.
\]
Next, for a \emph{hidden} set $T\subset N$---which will be drawn uniformly at random among all sets of size $B$---we define the value function
\[
        v_T(S) \defeq \max\{v_0(S), |S \cap T| \}.
\]

It is clear that both $v_0$ and $v_T$ are XOS (\Cref{def:XOS}). The optimal solution (subject to the budget constraint) differs substantially depending on whether the underlying value function is $v_0$ or $v_T$.

\begin{claim}[Optimum gap]
\label{claim:opt-gap}
Let $B = \sqrt n$ and $c_i = 1$ for all $i \in N$. Then $\OPT(v_0)=O(n^{2\delta})$ whereas for any hidden set $T$ of size $B$, $\OPT(v_T)\ge \sqrt n$. As a result,
\[
        \frac{\OPT(v_T)}{\OPT(v_0)}
        \ge
        \Omega(n^{1/2-2\delta}).
\]
\end{claim}

\begin{proof}
Since every feasible set has size at most $B$, we have
\[
        v_0(S)\le \max\{h,\beta B \}.
\]
Moreover,
\[
        \beta B = (1+\delta)n^{-1/2+\delta}\cdot n^{1/2}
        =
        (1+\delta)n^\delta
        \le h
\]
for any sufficiently large $n$. This shows that $\OPT(v_0) \le h = O(n^{2\delta})$. On the other hand, for $v_T$, the hidden set $T$ is feasible since $|T| = B$ (by construction). Thus,
\[
        v_T(T) \ge |T\cap T|= B = \sqrt n.
\]
This completes the proof.
\end{proof}

\subsubsection{Indistinguishability by value queries}

The next lemma establishes the key hardness statement in the oracle setting. A query distinguishes $v_T$ from $v_0$ only if it has an unusually large intersection with the random hidden set $T$.

\begin{lemma}[One-query indistinguishability]
\label{lem:one-query}
For every fixed query set $Q\subseteq N$,
\[
        \Pr_T\bigl[v_T(Q)\ne v_0(Q)\bigr]
        \le
        \exp(-n^{\Omega(\delta)}).
\]
\end{lemma}

\begin{proof}
Let $X= |Q\cap T|$. Since $T$ is a uniformly random $B$-subset of $N$, $|Q\cap T|$ is a hypergeometric random variable with mean
\[
        \mu = \E[ | Q\cap T| ] = \frac{|Q| B}{n}=\frac{|Q|}{\sqrt n}.
\]
By the definition of $v_T$, the values differ only if
\[
        | Q\cap T| > v_0(Q)=\max\{\min\{|Q|,h\},\beta |Q|\}.
\]
If $|Q|<h$, then $v_0(Q)=|Q|$ and $| Q\cap T| \le |Q|$, so the probability is zero. If $h \le |Q|\le n^{1/2+\delta}$, then $\mu\le n^\delta$ and $v_0(Q) \ge h = \lfloor (1+\delta) n^{2 \delta} \rfloor $. As a result, the threshold is at least a factor $n^\delta$ larger than the mean. A standard Chernoff bound for hypergeometric random variables gives
\[
        \Pr[| Q\cap T| \ge v_0(Q)] \leq \Pr[| Q\cap T| \ge h] \leq
        \exp(-n^{\Omega(\delta)}).
\]
Finally, let $|Q|> n^{1/2+\delta}$. In this case,
\[
        v_0(Q)
        \ge
        \beta |Q|
        =
        (1+\delta)n^{-1/2+\delta}|Q|
        =
        (1+\delta)n^\delta\mu.
\]
A Chernoff bound again yields
\[
        \Pr[| Q\cap T| \ge v_0(Q)] \leq \Pr[| Q\cap T| > \beta |Q|]\le \exp(-n^{\Omega(\delta)}).
\]
Combining the three cases proves the lemma.
\end{proof}

\begin{corollary}[Adaptive indistinguishability]
\label{cor:adaptive-indist}
Let $\mathcal Q$ be any adaptive algorithm that makes at most $q(n)=\poly(n)$ value queries. When $T$ is drawn uniformly at random among all $B$-subsets of $N$, the probability that $\mathcal Q$ receives different transcripts from the oracles $v_0$ and $v_T$ is at most
\[
        q(n)\exp(-n^{\Omega(\delta)})=o(1).
\]
In particular, no algorithm that submits polynomially many value queries can distinguish $v_0$ from a random planted $v_T$ with constant probability.
\end{corollary}

\begin{proof}
The claim follows by~\Cref{lem:one-query} and a union bound over the $q(n)$ queries.
\end{proof}

\subsubsection{Implications for approximate CCE}

We next prove an auxiliary lemma concerning the expected cost in an approximate CCE\@. It is independent of the hard distribution; it only uses budget feasibility and the deviation in which an agent always opts out.

\begin{lemma}[Expected selected cost in an approximate CCE]
\label{lem:cce-size}
Let $p$ be any budget-feasible payment rule with unit costs, and let $\mu$ be an $\eta$-CCE of the induced game. Then
\[
        \E_{S\sim\mu}[|S|]
        \le B+n\eta.
\]
\end{lemma}

\begin{proof}
For each agent $i\in N$, the fixed deviation to opt out yields zero utility. Thus, the $\eta$-CCE condition yields
\[
        \E_{S\sim\mu}[u_i(S)]\ge -\eta.
\]
Under unit costs, we have
\[
        u_i(S)=\mathbbm{1}\{i\in S\}\bigl(p_i(S)-1\bigr).
\]
Summing over all $i\in N$ yields
\[
        \E_{S \sim \mu} \left[\sum_{i\in S}p_i(S)-|S|\right]\ge -n\eta.
\]
By budget feasibility, $\sum_{i\in S} p_i(S)\le B$ for every $S$. Therefore,
\[
        \E_{S \sim \mu}[|S|]\le B+n\eta.
\]
\end{proof}

This lemma allows us to upper bound the value of approximate CCEs on the base instance.

\begin{corollary}[Value of approximate CCEs on the base instance]
\label{cor:base-low-value}
Let $\eta\le 1/n$. For every $\eta$-CCE $\mu$ of any budget-feasible payment rule on the base instance $v_0$,
\[
        \E_{S\sim\mu}[v_0(S)]\le O(n^{2\delta}).
\]
\end{corollary}

\begin{proof}
For every set $S\subseteq N$, we have
\[
        v_0(S)=\max\{\min\{|S|,h\},\beta |S|\}\le h+\beta |S|.
\]
Taking expectations and using~\Cref{lem:cce-size},
\[
        \E[v_0(S)]
        \le
        h+\beta(B+n\eta).
\]
Since $B=\sqrt n$, $\beta B=(1+\delta)n^\delta$, and $n\eta\le1$, the right-hand side is $O(n^{2\delta})$.
\end{proof}

\subsubsection{Putting everything together}

We are now ready to prove our main theorem. The proof is based on a reduction: a sufficiently strong PoA guarantee for CCEs would enable a payment designer to distinguish between $v_0$ and $v_T$ with polynomially many queries.

\begin{theorem}[PoA lower bound for CCEs under XOS valuations]
\label{thm:xos-cce-lb}
Let $\epsilon > 0$ be a constant and $\eta(n)\le 1/n$. There is no oracle-efficient payment rule per~\Cref{def:efficientpaymentrule} that for every monotone XOS value function $v$, outputs a budget-feasible payment rule such that any $\eta(n)$-CCE $\mu$ of the induced binary-action game satisfies
\[
        \E_{S\sim\mu}[v(S)]
        \ge
        \frac{1}{C n^{1/2-\epsilon}}\OPT(v)
\]
for a universal constant $C>0$.

Equivalently, any payment-rule design algorithm that guarantees a PoA of $O(n^{1/2-\epsilon})$ with respect to CCEs for monotone XOS value functions must use exponentially many value queries. This holds even with known unit costs and $\lambda = 1 / \sqrt n \to 0$.
\end{theorem}

\begin{proof}
For the sake of contradiction, suppose that such a payment rule algorithm $\mathcal A$ exists. We will show how to construct a polynomial-query distinguisher $\mathcal D$ for the hard pair $(v_0,v_T)$.

The distinguisher is given value-oracle access to an unknown function $v$, promised to be either $v_0$ or $v_T$ for a uniformly random hidden set $T$ of size $B$. It invokes $\mathcal A$ with oracle access to $v$---using polynomially many preprocessing value queries---and obtains a budget-feasible payment rule $p^v$.

Next, $\mathcal D$ simulates repeated play in the induced binary-action game for $\poly(n,1/\eta)$ rounds using any standard bandit-feedback external-regret algorithm for each agent. In each round, to provide agent $i$ with the payoff of opting in against the realized actions of the other agents, the simulator computes $p_i^v(S_{-i}\cup\{i\})-1$, which requires polynomially many value queries (by assumption). The payoff of opting out is zero. Since there are $n$ agents, $R=\poly(n,1/\eta)$ rounds, and every payment computation uses $\poly(n)$ value queries, the total number of value queries remains polynomial.

By the external-regret guarantee, the empirical distribution $\mu$ of play is an $\eta$-CCE with constant probability. The distinguisher then computes
\[
        \overline v=\E_{S\sim \mu}[v(S)],
\]
which is the average of $v(S)$ over the realized outcomes, using one value query per outcome. It then outputs ``planted'' if $\overline v \ge n^{\epsilon/2}$, and ``base'' otherwise.

We analyze the two cases separately.

\paragraph{Base case.} First, suppose $v=v_0$. With constant probability, the empirical distribution $\mu$ is an $\eta$-CCE of a budget-feasible payment game. Thus, by~\Cref{cor:base-low-value},
\[
        \overline v
        =
        \E_{S\sim\mu}[v_0(S)]
        \le
        O(n^{2\delta}).
\]
Setting $\delta < \epsilon/4$, we have $\overline v < n^{\epsilon/2}$ for all sufficiently large $n$. As a result, the output of the distinguisher is correct.

\paragraph{Planted case.} In the planted case, let $v=v_T$. As before, $\mu$ is an $\eta$-CCE of the induced game with constant probability. By the assumed guarantee of $\mathcal A$, we have
\[
        \overline v
        =
        \E_{S \sim \mu}[v_T(S)]
        \ge
        \frac{1}{C n^{1/2-\epsilon}}\OPT(v_T).
\]
Moreover, by~\Cref{claim:opt-gap}, $\OPT(v_T)\ge \sqrt n$. Thus,
\[
        \overline v \ge \frac{1}{C}n^\epsilon.
\]
For all sufficiently large $n$, this is larger than $n^{\epsilon/2}$. As a result, the output of the distinguisher is again correct.

We conclude that $\mathcal D$ distinguishes $v_0$ from a random $v_T$ with a constant probability using only polynomially many value queries. This contradicts~\Cref{cor:adaptive-indist}. This completes the proof.
\end{proof}

\subsection{Non-monotone value functions}
\label{sec:non-monotone}

We next consider what happens when we relax monotonicity. Our previous analysis made crucial use of monotonicity. First, it guarantees that the marginal contribution rule is nonnegative, guaranteeing limited liability. Moreover, monotonicity implies that if a payment rule ensures that an equilibrium $S$ contains a target set $T$, then $v(S) \geq v(T)$; this can fail in the non-monotone setting, which means that a payment rule here must do more than attract a good set of agents---it must also exclude harmful ones.

We begin by pointing out a simple example showing that marginal payments can be negative without monotonicity.

\begin{example}
Let $N=\{a,b\}$ and define
\[
        v(\emptyset)=0,
        \quad
        v(\{a\})=1,
        \quad
        v(\{b\})=1,
        \quad
        v(\{a,b\})=\frac12.
\]
This function is nonnegative and submodular, since $v(\{a\})+v(\{b\})=2\ge \frac12=v(\{a,b\})+v(\emptyset)$. On the other hand, it is not monotone since $v(\{a,b\}) < v(\{b\})$. As a result, the payment prescribed by the marginal contribution rule~\eqref{eq:mc-rule} for agent $a$ is negative. In fact, as we shall see in~\Cref{remark:rawrule}, even if we allow negative payments, the marginal contribution payment rule can still have unbounded PoA.
\end{example}

A natural way to repair this issue is to pay only the positive part of the marginal contribution. Namely, if $\Delta_i(S) \defeq v(S) - v(S \setminus \{i \})$,
\begin{equation}
        p_i^+(S)
        \defeq
        \begin{cases}
        B\dfrac{[\Delta_i(S)]_+}{\sum_{i' \in S}[\Delta_{i'}(S)]_+}
        &\text{if }i\in S\text{ and }\sum_{i' \in S}[\Delta_{i'}(S)]_+>0,\\
        0
        &\text{otherwise.}
        \end{cases}
        \label{eq:positive-part-rule}
\end{equation}
This rule is nonnegative and budget feasible (because of submodularity). However, it does not guarantee bounded PoA. 

\begin{proposition}
\label{thm:positive-marginal-nonmonotone-lb}
For any $\lambda \in(0,1)$, there are instances with $\max_{i \in N} c_i / B \leq \lambda$ such that the payment rule~\eqref{eq:positive-part-rule} has a
pure Nash equilibrium of arbitrarily small value while $\OPT = 1$. As a result, $\PoAPNE = +\infty$.
\end{proposition}

\begin{proof}
We let $B=1$ and choose any $\eta \in (0,1)$. We consider an instance with two agents, $a$ and $b$, and costs $c_a = \lambda$ and $c_b = 0$, respectively. The value function is defined as
\[
        v (\emptyset) = 0,
        \quad
        v(\{a\})=1,
        \quad
        v(\{b\})=\eta,
        \quad
        v(\{a,b\})=\eta.
\]
This function is submodular because $v(\{a\})+v(\{b\})=1+\eta\ge \eta=v(\{a,b\})+v(\emptyset)$. (It is not monotone since adding $b$ to $\{a\}$ reduces the value from $1$ to $\eta$.)

The optimal solution selects $\{a\}$, which is budget feasible and has value $1$, so
$\OPT=1$. On the other hand, we claim that $S = \{b \}$ is a pure Nash equilibrium under the payment rule~\eqref{eq:positive-part-rule}. Since $\Delta_b(\{b\})=\eta$, agent
$b$ receives the entire budget under~\eqref{eq:positive-part-rule} while incurring zero cost, so it does not want to opt out. Furthermore, if agent $a$ were to join, the new set would be $\{a,b\}$. At that set, $\Delta_a(\{a,b\}) = v(\{a,b\})-v(\{b\}) = 0$ and $\Delta_b(\{a,b\}) = v(\{a,b\})-v(\{a\}) = \eta - 1 < 0$, in turn implying that both players receive zero payment. Since $c_a > 0$, this is not a profitable deviation for $a$. We conclude that $\{b\}$ is a pure Nash equilibrium. The value at that equilibrium is $\eta$, so the claim follows in view of the fact that $\OPT = 1$.
\end{proof}

\begin{remark}
    \label{remark:rawrule}
    It is interesting to point out that this same example shows that the original marginal contribution payment rule also suffers from the same issue, in that $S = \{b\}$ is still an equilibrium. Agent $b$ certainly does not have an incentive to opt out as it receives the entire budget. Further, if agent $a$ were to opt in, it would receive zero payment. The difference now is that, without taking the positive part, the payment can be negative.
\end{remark}

\subsubsection{Impossibility for a broader class}
\label{sec:nm-equal-marginal-lb}

The previous lower bound concerns a very specific payment rule. We next show that the obstruction is broader. In particular, we consider deterministic cost-oblivious payment rules that satisfy the following natural conditions. (We recall that $\Delta_i(S) = v(S) - v(S \setminus \{i \})$.)
\begin{enumerate}[label=(A\arabic*)]
    \item \emph{Equal treatment of equal marginal contributions:} for every value function
    $v$, subset $S\subseteq N$, and agents $i, i' \in S$,
    \[
            \Delta_i(S)=\Delta_{i'}(S)
            \quad\Longrightarrow\quad
            p_i(S)=p_{i'}(S).
    \]
    \item \emph{Positive compensation on uniformly productive coalitions:} for
    every nonempty $S\subseteq N$, if there is a $q>0$ such that
    $\Delta_i(S) = q$ for every $i\in S$, then
    \[
            \sum_{i\in S} p_i(S)>0.
    \]
\end{enumerate}
The last condition only rules out paying a total of zero to a coalition whose members all have the same strictly positive marginal contribution.

For an integer $m\ge2$, we consider $m$ agents $L=\{\ell_1,\ldots,\ell_m\}$ and one high-value agent $h$, so $N=L\cup\{h\}$. For every $S \subseteq L$, we define
\begin{equation}
        v(S)=\frac{|S|}{m^2}
        \text{ and }
        v(S\cup\{h\})=1-\frac{|S|}{m}.
        \label{eq:nm-equal-marginal-value}
\end{equation}

\begin{lemma}
\label{lem:nm-equal-marginal-value-properties}
The function in~\eqref{eq:nm-equal-marginal-value} is normalized,
nonnegative, submodular, and non-monotone. Moreover, $0\le v(S) \le 1$ for every $S \subseteq N$.
\end{lemma}

\begin{proof}
We verify diminishing marginal returns. For any agent $\ell \in L$, the marginal value of adding $\ell$ to a set not containing $h$ is $1/m^2$, whereas the marginal value of adding $\ell$ to a set containing $h$ is $-1/m$. Next, the marginal value of adding $h$ to a set $S \subseteq L$ is
\[
        v(S \cup\{h\})-v(S)
        =
        1-\frac{|S|}{m}-\frac{|S|}{m^2},
\]
which is decreasing in $|S|$. This shows that $v$ is submodular. Finally, $v(\{h\})=1$ but $v(\{h,\ell_1\})=1-1/m<1$, so $v$ is not monotone.
\end{proof}

In the following proof, two coalitions will be important. First,
\begin{equation}
        v(L)=\frac1m
        \text{ and }
        \Delta_\ell(L)
        =
        v(L)-v(L\setminus\{\ell\})
        =
        \frac1{m^2}
        \text{ for all } \ell\in L.
        \label{eq:nm-equal-marginal-at-L}
\end{equation}
Second, at the full coalition $N$,
\begin{equation*}
        v(N)=0
        \text{ and }
        v(N\setminus\{h\})=v(L)=\frac1m,
\end{equation*}
and, for every $\ell\in L$,
\begin{equation*}
        v(N\setminus\{\ell\})
        =
        v((L\setminus\{\ell\})\cup\{h\})
        =
        \frac1m.
\end{equation*}
As a result,
\begin{equation}
        \Delta_i(N)=-\frac1m
        \text{ for every }i\in N.
        \label{eq:nm-equal-marginal-at-N}
\end{equation}

\begin{theorem}[Lower bound under equal-marginal treatment]
\label{thm:nm-equal-marginal-lb}
Fix any $\lambda\in(0,1)$, and let $p$ be any deterministic cost-oblivious, nonnegative, and budget-feasible payment rule satisfying \emph{(A1)--(A2)}. For every integer $m\ge \max\left\{2,\left\lceil\frac1\lambda\right\rceil\right\}$, there is an instance with $n=m+1$ agents, a normalized nonnegative non-monotone submodular value function, and strictly positive costs satisfying $0 < c_i \leq \lambda B$ for every $i \in N$, where $\max_{i \in N} c_i / B = \lambda$, such that the induced game has a strict pure Nash equilibrium $L$ and
\[
        \frac{\OPT}{v(L)}=m=n-1.
\]
\end{theorem}

\begin{proof}
We consider the value function in~\eqref{eq:nm-equal-marginal-value}. By~\eqref{eq:nm-equal-marginal-at-L}, (A1), and (A2), there is a common number $q > 0$ such that
\[
        p_\ell(L) = q
        \text{ for every } \ell \in L.
\]
Budget feasibility gives $m q \leq B$, so $q \leq B/m$. Next, by~\eqref{eq:nm-equal-marginal-at-N} and (A1), there is a common number $q' \ge 0$ such that
\[
        p_i(N) = q'
        \qquad\text{ for every } i \in N.
\]
Using budget feasibility again, $(m+1) q' \le B$, so $q' \leq B/(m+1)$. Since $m \ge \lceil1/\lambda\rceil$, we have $q' \leq B/(m+1) < \lambda B$. Now, we choose the private costs as follows.
\begin{equation*}
        c_\ell=\frac q2 \text{ for } \ell\in L
        \text{ and }
        c_h=\lambda B.
\end{equation*}
We have
\[
        c_\ell=\frac q2
        \le \frac{B}{2m}
        \le \frac{\lambda B}{2}
        <\lambda B.
\]
We claim that $L$ is a strict pure Nash equilibrium. Indeed, each agent $\ell\in L$ obtains
\[
        p_\ell(L)-c_\ell
        =
        q-\frac q2
        =
        \frac q2
        >
        0.
\]
If $h$ joins, the resulting coalition is
$N$, so
\[
        p_h(N) - c_h
        =
        q' - \lambda B
        <
        0.
\]
This shows that $L$ is a strict pure Nash equilibrium. However, the singleton $\{h\}$ is feasible because $c_h = \lambda B < B$, and $v(\{h\})=1$. By~\Cref{lem:nm-equal-marginal-value-properties}, no coalition has value exceeding
$1$, so $\OPT=1$. On the other hand, $v(L)=1/m$ by
\eqref{eq:nm-equal-marginal-at-L}, so
\[
        \frac{\OPT}{v(L)}
        =
        m
        =
        n-1. \qedhere
\]
\end{proof}

\subsubsection{A broader impossibility with zero-cost harmful agents}

The preceding lower bounds still impose structure on the payment rule. We next point out that harmful zero-cost agents can always make the price of anarchy arbitrarily large. In other words, the principal needs a way to exclude harmful agents in order to guarantee a finite price of anarchy.

\begin{proposition}[No finite PoA without exclusion of zero-cost harmful agents]
\label{prop:zero-cost-harmful-impossibility}
For any $\lambda \in (0,1)$ and  nonnegative budget-feasible payment rule, there exists a two-agent instance with a submodular value function and $\max_{i \in N} c_i / B \leq \lambda$ such that there is a pure Nash equilibrium with zero value whereas $\OPT = 1$. As a result, the price of anarchy is unbounded.
\end{proposition}

\begin{proof}
We let $B=1$. There are two agents, $a$ and $b$, with costs $c_a = \lambda$ and $c_b = 0$. The submodular value function is defined as
\[
        v(\emptyset)=0,
        \quad
        v(\{a\})=1,
        \quad
        v(\{b\})=0,
        \quad
        v(\{a,b\})=0.
\]
The optimal solution is $\{a\}$, so $\OPT=1$. Now, let $p$ be any nonnegative budget-feasible payment rule. We consider two cases.

\paragraph{Case 1: $p_a(\{a,b\})\le c_a$.}
Consider $S=\{b\}$. In this case, $S = \{b\}$ is a Nash equilibrium. Agent $b$ has no incentive to opt out since its payment is nonnegative and incurs zero cost. Moreover, since $p_a(\{a,b\})\le c_a$ (by assumption), $a$ has no profitable deviation. The value of this Nash equilibrium is $v(\{b\})=0$.

\paragraph{Case 2: $p_a(\{a,b\})>c_a$.}
In this case, $S = \{ a, b\}$ is a Nash equilibrium. Agent $a$ has strictly positive utility since $p_a(\{a,b\})>c_a$ (by assumption), as does agent $b$ since it incurs zero cost and the payment is nonnegative. The value of this Nash equilibrium is $v(\{a,b\})=0$.

In either case, there is a pure Nash equilibrium with zero value, whereas
$\OPT=1$.
\end{proof}

For completeness, \Cref{appendix:targeted-implementation} contains a complementary positive result beyond cost-oblivious payment rules, based on what we refer to as targeted implementation.

\section{Combinatorial compensation design}\label{sec:combinatorial}

The key premise so far is that each agent has \emph{binary} actions: they can either opt in or opt out. 
In many cases, however, the agent's action set is more complex. For example, a YouTube content creator could create multiple videos, and each video has a private cost. In this section, we introduce the \emph{combinatorial compensation design} problem, in which each agent chooses a subset from multiple actions. In contrast to the binary-action setting, we show that the combinatorial compensation design problem is non-trivial even in the single-agent setting in \Cref{sec:single-agent-combinatorial}. We further study the even more general problem of multi-agent combinatorial compensation design in \Cref{sec:multi-agent-combinatorial}.


\subsection{Single-agent}
\label{sec:single-agent-combinatorial}


We first consider the setting where there is a single agent with a finite set of actions $A=\{1,\ldots,m\}$. The
agent chooses an arbitrary subset $S\subseteq A$. Each action $e\in A$ has a \emph{private} cost $c_e>0$, and costs are additive:
$
  c(S)=\sum_{e\in S}c_e.
$
The principal has a budget $B > 0$ and a normalized monotone value function
$v:2^A\to \R_{\geq 0}$ with $v(\emptyset)=0$. The principal designs a compensation rule $p: 2^{A} \rightarrow \mathbb{R}_{\ge 0}$ that pays $p(S)$ to the agent if they choose $S \subseteq A$. We require the
compensation rule to be \emph{budget feasible}, meaning that
\[
  0\leq p(S)\leq B \qquad \forall S\subseteq A.
\]
Given the compensation rule $p$ and the cost vector $c$, the agent chooses a subset
$S_p(c)$ satisfying
\[
  S_p(c)\in\argmax_{S \subseteq A}\{p(S)-c(S)\}.
\]

The principal aims to design a compensation rule that maximizes its value. The full-information benchmark is
\[
  \OPT \defeq \max\{v(T): T\subseteq A,\ c(T)\leq B\}.
\]


We assume \emph{platform-favoring tie-breaking}: among all subsets maximizing
the agent's payoff $p(S)-c(S)$, the agent chooses one with maximum principal
utility $v(S)$. This assumption matters only at indifference points, for example,
when a feasible set has payoff exactly zero.

\subsubsection{Deterministic rules}
Unlike the binary-action setting, where a deterministic compensation rule achieves a constant PoA even in the multi-agent setting with a submodular value function, any deterministic rule incurs $\Omega(m)$ PoA even in single-agent combinatorial compensation design, and even for additive values. The lower bound holds even when $c_e \le B$ for all actions $e \in A$. This lower bound is tight: under the same assumption, one can achieve a PoA of $m$ even for subadditive values by paying the agent $B$ only when they choose an action $e^* \in \argmax_{e \in A} v(\{e\})$, which guarantees $v(\{e^*\})\ge v(A)/m \ge \OPT/m$.

\begin{theorem}[$\Omega(m)$ lower bound for deterministic rules]
\label{thm:single-agent-lower-deterministic}
Fix \(m\ge 2\), \(B>0\), and the additive value function $v(S)=\frac{|S|}{m}$. The price of anarchy for any deterministic rule is at least $m$ even when the costs satisfy $c_e < B$ for any $e \in A$.
\end{theorem}

\begin{proof}
Let \(p:2^A\to[0,B]\) be any deterministic budget-feasible cost-oblivious compensation rule.
Fix any $\eta\in\left(0,\frac{1}{m}\right)$. We consider the following two uniform cost vectors:
\[
        \bar{c}_e=(1-\eta)B
        \quad\forall e\in A,
\]
and
\[
        \underbar{c}_e =\frac{B}{m}
        \quad\forall e\in A.
\]
We denote by $\PoAPNE(p; c)$ the price of anarchy with respect to pure Nash equilibria induced by the payment rule $p$ and costs $\{ c_e \}_{e \in A}$. We will prove that
\[
        \max\{\PoAPNE(p;\bar{c}),\PoAPNE(p; \underbar{c})\}\ge m.
\]

We write the agent's utility under a cost vector \(c\) as $u_c(S)=p(S)-c(S)$. We first consider the high-cost instance \(\bar{c}\). Every singleton is budget feasible, and
no set of size at least \(2\) is budget feasible. Hence,
\[
        \OPT(\bar{c})=\frac{1}{m}.
\]
Because \(\eta<1/m\le 1/2\), every set \(S\) with \(|S|\ge 2\) has negative utility, so it cannot be chosen in equilibrium in the high-cost instance. The only possible sets in equilibrium are \(\emptyset\) and singletons. If the agent chooses \(\emptyset\), then the principal obtains a value \(0\), in which case the price of anarchy is infinite. It remains to consider the case where under $p$ and $\bar{c}$, the agent chooses a singleton. We write this singleton as \(\{e\}\). Since \(\{e\}\) is chosen under \(\bar{c}\),
\[
        p(\{e\})-(1-\eta)B
        =
        u_{\bar{c}}(\{e\})
        \ge
        u_{\bar{c}}(\emptyset)
        =
        p(\emptyset).
\]
Thus,
\begin{equation}
\label{eq:singleton-nearly-full}
        p(\{e\})\ge p(\emptyset)+(1-\eta)B.
\end{equation}

We now consider the low-cost instance \(\underbar{c}\). The full set \(A\) has cost $\underbar{c}(A)=m\cdot \frac{B}{m}=B$, so it is budget feasible. Thus,
\[
        \OPT(\underbar{c})=v(A)=1.
\]
For the singleton set \(\{e\}\) identified above, using
\eqref{eq:singleton-nearly-full}, we have
\begin{equation}
\label{eq:singleton-low-utility}
    u_{\underbar{c}}(\{e\})
        = p(\{e\})-\frac{B}{m} \ge
        p(\emptyset)+(1-\eta)B-\frac{B}{m} >
        B-\frac{2B}{m},
\end{equation}
where the last inequality holds since $\eta < \frac{1}{m}$. On the other hand, every set \(S\) with \(|S|\ge 2\) satisfies
\[
        u_{\underbar{c}}(S)
        =
        p(S)-|S|\frac{B}{m}
        \le
        B-|S|\frac{B}{m}
        \le
        B-\frac{2B}{m}.
\]
It thus follows that the agent will choose a singleton under \(\underbar{c}\). Since every singleton has value \(1/m\), the principal obtains value \(1/m\). Thus,
\[
        \PoAPNE(p;\underbar{c})
        =
        \frac{\OPT(\underbar{c})}{1/m}
        =
        \frac{1}{1/m}
        =
        m.
\]
This completes the proof.
\end{proof}


\subsubsection{Randomized threshold-prize rule}
Given the lower-bound results for deterministic rules, even in the additive-value setting, it is then interesting to investigate whether \emph{randomized} compensation rules would have improved PoA bounds. We consider randomized compensation rules that are distributions over deterministic compensation
rules. The random seed is drawn and publicly announced before the agent chooses
a subset. The performance of such a rule is the expected principal's value over
this public randomization and the induced platform-favoring best-response of the
agent. Our main result shows that there is a randomized compensation rule with PoA of $O(\log m)$ even for subadditive values.

Throughout this section, we assume the value function $v$ is normalized, monotone, and subadditive: $v(S \cup S') \leq v(S)+v(S'), \forall S, S' \subseteq A$. We let $U \defeq v(A)$ and recall that $\OPT$ is the full-information benchmark. If $\OPT=0$, the problem is trivial. So we can assume that $U \ge \OPT >0$. We set the parameter
\[
  L \defeq \lceil \log_2 m\rceil.
\]
We propose the following \emph{randomized threshold-prize rule}. It first samples an index $r\in\{0,1,\ldots,L\}$ uniformly at random and publicly
announces it. It then sets a threshold
\[
  \tau_r \defeq \frac{U}{2^r}
\]
and selects the deterministic payment rule
\[
  p_r(S) \defeq
  \begin{cases}
  B & \text{if } v(S) \geq \tau_r,\\
  0 & \text{if } v(S)<\tau_r.
  \end{cases}
\]
This rule is clearly cost oblivious and budget feasible. The idea is that since $\OPT \in [U/m, U]$, we can use a geometric scaled grid of size $O(\log m)$ to guess the range of $\OPT$ and get a value of $O(\OPT)$ with probability at least $\Omega(1/\log m)$.

\begin{theorem}
\label{thm:single-agent-upper-basic}
Assume that $v$ is normalized, monotone, and subadditive. Suppose further that $c_e\leq B$ for
all actions $e\in A$. The randomized threshold-prize rule satisfies
\[
  \E[v(S)]\geq \frac{1}{2(\lceil \log_2 m \rceil+1)}\OPT.
\]
\end{theorem}

\begin{proof}
Let $S^* \in\argmax \{v(S): c(S) \leq B \}$ be an optimal feasible set and
$\OPT \defeq v(S^*)$. Since every singleton
is feasible (by assumption), $\OPT \geq v(\{e\})$ for every element $e$. By subadditivity,
\[
  U=v(A)\leq \sum_{e\in A} v(\{e\})\leq m\OPT.
\]
Moreover, monotonicity yields $\OPT \leq U$. Thus,
\[
  \frac{U}{m}\leq \OPT\leq U.
\]
Because $2^L\geq m$, the smallest threshold satisfies $\tau_L\leq U/m\leq \OPT$.
Now, let $r^*$ be the smallest index such that $\tau_{r^*}\leq \OPT$. If $r^*=0$,
then $\tau_{r^*}=U\geq \OPT$, so $\tau_{r^*} = \OPT$. If $r^*>0$, then $\tau_{r^*-1}>\OPT$, so
\[
  \tau_{r^*}=\frac{1}{2}\tau_{r^*-1}>\frac{\OPT}{2}.
\]
In either case,
\[
  \tau_{r^*}\geq \frac{\OPT}{2}
  \text{ and }
  \tau_{r^*}\leq \OPT.
\]

Consider now the event $r=r^*$. The optimal set $S^*$ exceeds the payment threshold because
$v(S^*) = \OPT \geq \tau_{r^*}$, so the agent can obtain payoff
\[
  p_{r^*}(S^*)-c(S^*) = B - c(S^*) \geq 0.
\]
As a result, under platform-favoring tie-breaking, the agent chooses a qualifying set in this branch. In this case, the principal obtains value at least
\[
  \tau_{r^*}\geq \frac{\OPT}{2}.
\]
Finally, the branch $r^*$ occurs with probability $1/(L+1)$. Since values are nonnegative in
any other event, we conclude that
\[
  \E[v(S)]\geq \frac{1}{L+1}\cdot \frac{\OPT}{2}. \qedhere
\]
\end{proof}

We will now provide an improvement to~\Cref{thm:single-agent-upper-basic} under a large-market condition. Namely, we assume that the principal knows an upper bound
\[
  c_e \leq \lambda B\qquad \forall e\in A
\]
for some $\lambda\in(0,1]$. We set
\[
  k \defeq \min\left\{m,\left\lfloor \frac{1}{\lambda}\right\rfloor\right\}
  \text{ and }
  q \defeq \left\lceil\frac{m}{k}\right\rceil.
\]
By the cost bound, every set comprising at most $k$ actions is budget feasible. We will use a similar threshold-prize rule, but now set
\[
  L_\lambda \defeq \lceil \log_2 q\rceil
  \text{ and }
  \tau_r=\frac{U}{2^r}.
\]
where $r \in \{0, 1, \dots, L_\lambda \}$ is again sampled uniformly at random.

\begin{theorem}
\label{thm:single-agent-upper-large-market}
Assume that $v$ is normalized, monotone, and subadditive. Suppose further that
$c_e \leq \lambda B$ for all $e \in A$. The randomized threshold-prize rule
with $L_\lambda=\lceil \log_2 q\rceil$ satisfies
\[
  \E[v(S)]\geq
  \frac{1}{2(L_\lambda+1)}\OPT,
  \text{ where }
  q=\left\lceil\frac{m}{\min\{m,\lfloor1/\lambda\rfloor\}}\right\rceil.
\]
In particular, the approximation factor is
\[
  O\left(\log q+1\right)
  =O\left(\log(m \lambda+2)\right).
\]
\end{theorem}

\begin{proof}
We partition $A$ arbitrarily into $q$ blocks $A_1, \ldots, A_q$, each of size at most $k$. Since every action has cost at most $\lambda B$ (by assumption) and $k\leq 1/\lambda$, it follows that every block is budget feasible. Therefore,
\[
  \OPT \geq \max_{\ell\in[q]} v(A_\ell).
\]
By subadditivity,
\[
  U=v(A) \leq \sum_{\ell=1}^q v(A_\ell)
  \leq q \max_{\ell \in [q]} v(A_\ell)
  \leq q \OPT.
\]
Together with $\OPT\leq U$, this gives
\[
  \frac{U}{q}\leq \OPT\leq U.
\]
The rest of the proof is analogous to the proof of~\Cref{thm:single-agent-upper-basic}, with $m$ replaced by $q$.
\end{proof}

\begin{remark}[Extreme case of~\Cref{thm:single-agent-upper-large-market}]
If $m\lambda\leq 1$, then $k=m$ and $q=1$. In this extreme case, the whole action set $A$ is budget feasible, so $A$ is optimal by monotonicity.
\end{remark}

\subsubsection{A matching lower bound for additive values}

We now show that the logarithmic dependence obtained in~\Cref{thm:single-agent-upper-basic} is unavoidable for cost-oblivious rules, even when
$v$ is additive.

In what follows, we normalize $B=1$. We will consider the additive value function
\[
  v(S)=|S|.
\]
For a scalar $t > 0$, consider the uniform cost vector
\[
  c_e=t \quad \forall e\in A.
\]
The optimum subject to the budget constraint reads
\[
  \OPT(t) = \max\{|S|: t|S| \leq 1\} = \min\left\{m,\left\lfloor\frac{1}{t}\right\rfloor\right\}.
\]

Now, we consider a cost-oblivious payment rule $p:2^A\to[0,1]$. We define
\[
  p_k=\max\{p(S): |S|=k\},\qquad k=0,1,\ldots,m,
\]
where by convention $p_0 = p(\emptyset)$. At a cost level $t$, the agent will choose a cardinality
\[
  D_p(t) \in \argmax_{k \in \{0,\ldots,m\}} \{p_k-t k\},
\]
where ties are broken in favor of the largest $k$, which is the platform-favoring
choice for $v(S)=|S|$. The proof of our lower bound will make use of the following bound on the demand integral.

\begin{lemma}[Demand integral bound]
\label{lem:single-agent-demand-integral}
For every deterministic cost-oblivious payment rule $p$ and every interval
$0<a<b$,
\[
  \int_a^b D_p(t) dt\leq 1.
\]
\end{lemma}

\begin{proof}
We define
\[
  g(t) \defeq \max_{k\in\{0,\ldots,m\}}\{p_k - t k\}.
\]
The function $g$ is the pointwise maximum of finitely many affine functions, so it
is convex and piecewise linear. At every point $t$ at which the maximizer is
unique, $g'(t) = - D_p(t)$. Since the set of breakpoints is finite, we have
\[
  \int_a^b D_p(t) dt
  =g(a)-g(b).
\]
Since $0\leq p_k\leq 1$ for all $k$, we have $g(a)\leq 1$. Since $k=0$ gives payoff
$0$, we also have $g(b)\geq 0$, and the claim follows.
\end{proof}

\begin{theorem}[Logarithmic lower bound for additive values]
\label{thm:single-agent-lower-basic}
For any randomized cost-oblivious payment rule, there exists a uniform
positive cost vector $c_e = t \leq 1$ such that, for the additive function
$v(S)=|S|$, the expected selected value is at most an
$O(1/\log m)$ fraction of $\OPT(t)$.
\end{theorem}

\begin{proof}
Consider a randomized payment rule, which is a
distribution over deterministic payment rules $p_\omega$. For a realized rule
$p_\omega$ and uniform cost $t$, let $D_\omega(t)$ denote the selected
cardinality under platform-favoring tie-breaking. The expected principal value is $\E_\omega[D_\omega(t)]$.

If the randomized rule guarantees a $\rho$-fraction of the optimum for every
uniform cost $t \in [1/m,1/2]$, then for every such $t$,
\[
  \E_\omega[D_\omega(t)] \geq \rho\left\lfloor\frac{1}{t}\right\rfloor.
\]
Since $t \leq 1/2$,
\[
  \left\lfloor\frac{1}{t} \right\rfloor\geq \frac{1}{2t}.
\]
Therefore,
\[
  \E_\omega[D_\omega(t)]\geq \frac{\rho}{2t}
  \qquad \forall t \in[1/m,1/2].
\]
Integrating over the cost $t$ and applying \Cref{lem:single-agent-demand-integral} to each
realized deterministic rule,
\[
  1 \geq
  \E_\omega\left[\int_{1/m}^{1/2} D_\omega(t) dt \right]
  =
  \int_{1/m}^{1/2} \E_\omega[D_\omega(t)] dt
  \geq
  \frac{\rho}{2} \int_{1/m}^{1/2}\frac{dt}{t}.
\]
Thus,
\[
  1\geq \frac{\rho}{2}\log\left(\frac{m}{2}\right),
\]
so
\[
  \rho\leq \frac{2}{\log(m/2)}.
\]
We conclude that the approximation factor is at least $\Omega(\log m)$.
\end{proof}

An analogous argument gives the matching dependence under the large-market parameter of~\Cref{thm:single-agent-upper-large-market}.

\begin{corollary}
\label{cor:single-agent-lower-large-market}
Let $\lambda\in[2/m,1/2]$. For any randomized cost-oblivious payment
rule, there exists a uniform positive cost vector $c_e = t$ with $t \leq \lambda$ such that, for the additive function $v(S)=|S|$, the expected
selected value is at most an $O(1/\log(m\lambda))$ fraction of the optimum.
\end{corollary}

\begin{proof}
We follow the proof of~\Cref{thm:single-agent-lower-basic}, but
we now integrate over $t \in [1/m,\lambda]$. If a $\rho$-fraction approximation held
for every such $t$, then
\[
  1
  \geq
  \int_{1/m}^{\lambda} \E_\omega[D_\omega(t)] dt
  \geq
  \frac{\rho}{2}\int_{1/m}^{\lambda} \frac{dt}{t}
  =
  \frac{\rho}{2}\log(m\lambda),
\]
so $\rho \leq 2 / \log(m \lambda)$.
\end{proof}

\subsection{Multiple agents}
\label{sec:multi-agent-combinatorial}

We finally consider the multi-agent version of the combinatorial-action model. Here there is a finite ground set $A$ of atomic actions, partitioned across $n$ agents: $A=A_1\sqcup A_2 \sqcup \cdots \sqcup A_n$, where $A_i$ is the set of actions controlled by agent $i$. Let $m=|A|$. Agent
$i$ chooses an arbitrary subset $S_i\subseteq A_i$, and the realized action set is $S=\bigcup_{i=1}^n S_i$. Each atomic action $e\in A$ has a private cost $c_e \ge 0$. As before, it is assumed that costs are additive, so
agent $i$'s cost from choosing $S_i$ is $c(S_i) = \sum_{e\in S_i}c_e$. The principal has a normalized, monotone, submodular value function
$v:2^A\to \R_{\geq 0}$. A payment rule assigns a nonnegative payment $p_i(S_1,\ldots,S_n)$ to each agent. It is budget feasible if
\[
        \sum_{i=1}^n p_i(S_1,\ldots,S_n)\le B
        \qquad\text{for every profile }(S_1,\ldots,S_n).
\]
The utility of agent $i$ is $p_i(S_1,\ldots,S_n)-c(S_i)$. The full-information benchmark is
\[
        \OPT
        =
        \max\left\{v(S): S \subseteq A,\ \sum_{e \in S} c_e \le B \right\}.
\]

As in the single-agent model, we use principal-favoring tie-breaking: whenever
an agent has multiple payoff-maximizing choices against the other agents' choices, it chooses one that maximizes the principal's value. This assumption is only used in the singleton-prize branch below, and only at exact indifference when an action costs exactly $B$.

Since the multi-agent setting generalizes the single-agent setting, we know that it is impossible to beat $\Omega(\log m)$ (\Cref{thm:single-agent-lower-basic}). However, whether it is possible to achieve a PoA of $O(\log m)$ in the multi-agent setting is not straightforward. The multi-agent setting is harder because the payment rule must not only incentivize high-quality actions within each agent's action set, but also manage competition among multiple agents, all without knowing the private costs. As the main result in this section, we show in \Cref{thm:multi-agent-upper} that it is still possible to achieve $O(\log m)$.

\subsubsection{A randomized payment rule}

The randomized rule combines ideas from both the single-agent combinatorial and multi-agent binary-action settings. It has two types of branches. The singleton prize branch is similar to that in \Cref{thm:no-large-market-randomized} to guarantee a lower bound of value. The second type of branch is a \emph{capped marginal-contribution branch}, which incentivizes high-quality contributions from the agents. Specifically, we first sample a target value on a geometric scale intended to approximate $\OPT$. Given such a target value $R$, the mechanism then uses the marginal contribution payment based on the \emph{capped payment potential} $F_R(S):= B \min \{\frac{v(S)}{R},1\}$. Here, the cap $R$ incentivizes agents to contribute value at least $R$, while the marginal contribution structure ensures sufficient competition among agents. We show that when $R$ is a constant-factor lower estimate of $\OPT$, this branch achieves constant PoA. Randomization over the geometric scale thus gives a PoA of $O(\log m)$. We provide a detailed proof below. 


\paragraph{The singleton-prize branch.} First, we let
\[
        M=\max_{e\in A}v(\{e\})
        \text{ and }
        e^*\in\argmax_{e\in A}v(\{e\}).
\]
Let $i(e^*)$ denote the owner of $e^*$. In the singleton-prize branch, the
principal pays
\[
        p_i^{\mathrm{sing}}(S_1,\ldots,S_n)
        =
        \begin{cases}
        B & \text{if } i=i(e^*) \text{ and }e^*\in S_i,\\
        0 & \text{otherwise.}
        \end{cases}
\]
This branch is clearly budget feasible. The following lemma follows directly.

\begin{lemma}[Singleton branch]
\label{lem:multi-singleton-branch}
If every individual action is affordable, $c_e\le B$ for all $e\in A$, then in every pure Nash equilibrium of the singleton-prize branch, the
realized set has value at least $M$.
\end{lemma}

\paragraph{The capped marginal-potential branches.}
Next, for a scale parameter $R>0$, we define the \emph{capped payment potential}
\[
        F_R(S)=B\min\left\{\frac{v(S)}{R},1\right\}.
\]
The branch with scale $R$ pays each agent its final marginal contribution to
$F_R$; namely
\[
        p_i^R(S_1,\ldots,S_n)=F_R(S)-F_R(S_{-i}),
\]
where $S=\bigcup_{j=1}^n S_j$ and $S_{-i}=\bigcup_{i' \ne i} S_{i'}$. In what follows, we show that the capped marginal-potential compensation rule is budget feasible and that the induced game among agents always admits a pure Nash equilibrium. We first show that the capped potential $F_R$ is submodular.

\begin{lemma}[Capped potential is submodular]
\label{lem:capped-potential-submodular}
For every $R>0$, the function $F_R$ is normalized, monotone, and submodular.
\end{lemma}

\begin{proof}
Normalization and monotonicity are immediate since $v$ is assumed to be normalized and monotone. To prove submodularity, let $S \subseteq S' \subseteq A$ and
$e \notin S'$. Since $v$ is monotone and submodular,
\[
        v(S) \le v(S') \text{ and }
        v(S\cup\{e\})-v(S) \ge v(S' \cup\{e\})-v(S') \ge 0.
\]
Now, let $\phi(x)=B\min\{x/R,1\}$. The function $\phi$ is nondecreasing and concave, so it follows that
\[
        \phi(v(S) + v(S\cup\{e\})-v(S)  ) - \phi(v(S))
        \ge
        \phi(v(S') + v(S' \cup\{e\})-v(S') )-\phi( v(S') ).
\]
This is exactly
\[
        F_R(S\cup\{e\})-F_R(S)
        \ge
        F_R(S' \cup\{e\})-F_R(S'),
\]
so $F_R$ is submodular.
\end{proof}

The capped potential then satisfies budget feasibility.

\begin{lemma}[Budget feasibility]
\label{lem:capped-branch-budget}
For every $R>0$, the capped marginal-potential rule is budget feasible: $\sum_{i=1}^n p^R_i(S_1,\ldots,S_n)\le B$ for any $(S_1, \dots, S_n)$.
\end{lemma}

\begin{proof}
Let $S=\bigcup_{i=1}^n S_i$. We order the nonempty blocks as $S_{i_1},\ldots,S_{i_k}$, and let
$T_\ell=S_{i_1}\cup\cdots\cup S_{i_{\ell-1}}$. By submodularity,
\[
        F_R(S)-F_R(S\setminus S_{i_\ell})
        \le
        F_R(T_\ell\cup S_{i_\ell})-F_R(T_\ell).
\]
Summing over $\ell$ telescopes:
\[
        \sum_{\ell=1}^k\bigl(F_R(S)-F_R(S_{-i_\ell})\bigr)
        \le
        F_R(S)-F_R(\emptyset)
        =
        F_R(S)
        \le B,
\]
so the total payment is at most $B$.
\end{proof}

\begin{lemma}[Pure equilibrium existence]
\label{lem:capped-branch-pne-existence}
For every $R>0$, the capped marginal-potential branch admits a pure Nash
equilibrium.
\end{lemma}

\begin{proof}
We define
\[
        \Phi_R(S_1,\ldots,S_n) \defeq F_R(S)-\sum_{e\in S}c_e.
\]
If agent $i$ changes from $S_i$ to $S'_i\subseteq A_i$, holding all other agents
fixed, then the change in its utility is
\[
\begin{aligned}
\left[F_R(S_{-i}\cup S'_i)-F_R(S_{-i}) - c(S'_i)\right]&
-
\left[F_R(S)-F_R(S_{-i})-c(S_i)\right]\\
&\qquad =
F_R(S_{-i}\cup S'_i)-F_R(S)-c(S'_i)+c(S_i),
\end{aligned}
\]
which is exactly the change in $\Phi_R$. As a result, this is a finite potential game.
\end{proof}

We are now ready to justify the use of capped potentials: if one identifies a suitable scale $R$, then every pure Nash equilibrium of the capped branch has a high value unless the optimum is essentially captured by a large singleton.

\begin{lemma}[Fixed-scale guarantee]
\label{lem:multi-fixed-scale}
Assume that every individual action is affordable. Let $S$ be any pure Nash
equilibrium of the capped marginal-potential branch with scale $R$. If $M=\max_{e\in A}v(\{e\})$, then
\[
        v(S) \ge \min\{R-M,\OPT-R\}.
\]
\end{lemma}

\begin{proof}
If $v(S) \ge R-M$, the conclusion is immediate, so in the remainder of the analysis we can assume instead that $v(S) < R-M$. Let $S^*$ be an optimal budget-feasible set, so $v(S^*)=\OPT$ and
$\sum_{e\in S^*}c_e\le B$.

We fix any action $e \in S^*\setminus S$, which belongs to some agent $i$. This agent can unilaterally add $e$ to its chosen set. Since $S$ is a pure Nash
equilibrium, this deviation cannot be profitable, so
\begin{equation}
    \label{eq:diff-F}
    F_R(S\cup\{e\})-F_R(S)\le c_e.
\end{equation}
By submodularity and normalization,
\[
        v(S\cup\{e\})-v(S)\le v(\{e\})\le M.
\]
Since $v(S)<R-M$, we have $v(S\cup\{e\})<R$. Therefore, adding $e$ does not hit the cap in $F_R$, so
\[
        F_R(S\cup\{e\})-F_R(S)=\frac{B}{R}\bigl(v(S\cup\{e\})-v(S)\bigr).
\]
Combining with~\eqref{eq:diff-F},
\[
        v(S\cup\{e\}) - v(S) \le \frac{R}{B} c_e.
\]
Using monotonicity and submodularity,
\[
\begin{aligned}
        \OPT
        =
        v(S^*)
        &\le v(S\cup S^*)\\
        &\le v(S)+\sum_{e\in S^*\setminus S}\bigl(v(S\cup\{e\})-v(S)\bigr)\\
        &\le v(S)+\frac{R}{B}\sum_{e\in S^*}c_e\\
        &\le v(S)+R.
\end{aligned}
\]
Rearranging, we get $v(S) \ge \OPT-R$. Combining the two cases gives the claim.
\end{proof}

This lemma implies that identifying a good scale $R$ suffices to upper bound the price of anarchy, as we point out below.

\begin{corollary}[A good scale gives a constant PoA]
\label{cor:multi-good-scale}
If
\[
        M<\frac{\OPT}{4} \text{ and } \frac{\OPT}{2}<R\le \frac{3\OPT}{4},
\]
then every pure Nash equilibrium of the $R$-branch satisfies $v(S)\ge \OPT/4$. If instead $M \ge \OPT/4$, then the singleton-prize branch has value at least
$\OPT/4$ in every pure Nash equilibrium.
\end{corollary}

\begin{proof}
If $M<\OPT/4$ and $\OPT/2<R\le 3\OPT/4$, then
\[
        R-M>\frac{\OPT}{2}-\frac{\OPT}{4} = \frac{\OPT}{4} \text{ and } \OPT-R\ge \frac{\OPT}{4},
\]
so~\Cref{lem:multi-fixed-scale} gives $v(S)\ge \OPT/4$. The singleton case follows
from~\Cref{lem:multi-singleton-branch}.
\end{proof}

\subsubsection{Unknown scale and the logarithmic approximation}

The challenge, of course, is that the principal does not know $\OPT$ since costs are assumed to be private. We now show that one can obtain a logarithmic approximation by randomly guessing the value scale. This mirrors the algorithm behind~\Cref{thm:single-agent-upper-basic} concerning the single-agent (combinatorial) setting.

First, if every $c_e \le B$ for
all $e\in A$, then every singleton is feasible and, by monotonicity and
subadditivity of $v$,
\[
        \frac{v(A)}{m}\le \OPT\le v(A).
\]
More broadly, if the principal is given a large-market parameter $\lambda\in(0,1]$ such that $c_e \le \lambda B$ for all $e \in A$, we can define
\begin{equation}
    \label{eq:q}
    k = \min\left\{m,\left\lfloor \frac{1}{\lambda}\right\rfloor\right\}
        \text{ and }
        q = \left\lceil \frac{m}{k}\right\rceil.
\end{equation}
We can then partition $A$ arbitrarily into $q$ blocks, each of size at most
$k$. Every block has total cost at most $B$. Moreover, since $v$ is subadditive,
\[
        v(A)\le \sum_{\ell=1}^{q}v(A_\ell).
\]
This means that some feasible block has value at least $v(A)/q$, so
\begin{equation*}
    \frac{v(A)}{q} \le \OPT\le v(A).
\end{equation*}
In what follows, in the general case where a large-market parameter is not given, we can set $q = m$. We can then tune the upper bound as
\begin{equation}\label{eq:L}
    L = \left\lceil \log_{3/2}\left(\frac{4q}{3}\right)\right\rceil,
\end{equation}
so that the scales are given by
\begin{equation}
    \label{eq:scales}
    R_r=\frac{v(A)}{(3/2)^r},
        \quad
        r=0,1,\ldots,L.
\end{equation}
In this context, we consider a randomized payment rule that chooses uniformly from the following $L+2$ branches:
\[
        \text{singleton-prize branch},
        R_0\text{-branch},R_1\text{-branch},\ldots,R_L\text{-branch}.
\]
The scales constructed in~\eqref{eq:scales} have the following property.
\begin{lemma}[Scale coverage]
\label{lem:multi-scale-coverage}
There exists $r\in\{0,1,\ldots,L\}$ such that
\[
        \frac{\OPT}{2} \le R_r \le \frac{3\OPT}{4}.
\]
\end{lemma}

\begin{proof}
Since $\OPT\le v(A)$, the initial scale $R_0=v(A)$ is at least $\OPT$. Since $L \ge \log_{3/2}(4q/3)$, we have
\[
        R_L=\frac{v(A)}{(3/2)^L}\le \frac{3v(A)}{4q}\le \frac{3\OPT}{4},
\]
where the last inequality uses $v(A)/q\le \OPT$. Let $r$ be the smallest index
such that $R_r\le 3\OPT/4$. If $r > 0$, we have $R_r\le 3\OPT/4$ and $R_{r-1} > 3 \OPT/4$. Thus,
\[
        R_r=\frac{2}{3}R_{r-1}>\frac{\OPT}{2}.
\]
The case $r=0$ can only arise if $\OPT=0$, in which case the theorem follows trivially.
\end{proof}

We are now ready to state our main guarantee.

\begin{theorem}[Logarithmic approximation]
\label{thm:multi-agent-upper}
Assume that $v$ is normalized, monotone, and submodular, and every individual
action is affordable. The randomized payment rule described above is cost oblivious and budget feasible. For every cost vector and every choice of pure Nash equilibrium in each realized branch,
\[
        \E[v(S)] \ge
        \frac{\OPT}{4(L+2)} \geq \frac{\OPT}{O(\log q+1)},
\]
where $L$ and $q$ are defined in~\eqref{eq:L} and~\eqref{eq:q}, respectively. In particular, the PoA is always at most $O(\log m)$.
\end{theorem}

\begin{proof}
The singleton-prize branch is clearly budget feasible, while the capped branches are budget feasible by~\Cref{lem:capped-branch-budget}. The fact that the payment rule is cost oblivious is immediate from the definitions.

Now, let $M=\max_{e\in A}v(\{e\})$. If $M\ge \OPT/4$, then the singleton-prize branch---which is chosen with probability $1/(L+2)$---has equilibrium value at
least $\OPT/4$. Since all other branches have nonnegative values, we have
\[
        \E[v(S)]\ge \frac{1}{L+2}\cdot \frac{\OPT}{4}.
\]
If $M<\OPT/4$, then by~\Cref{lem:multi-scale-coverage} there exists a scale $R_r$ satisfying
\[
        \frac{\OPT}{2} \le R_r\le \frac{3\OPT}{4},
\]
so in that branch every pure Nash equilibrium has value at least $\OPT/4$ by~\Cref{cor:multi-good-scale}. Since this branch is chosen with probability $1/(L+2)$, the same lower bound on the expected value follows.
\end{proof}

As we saw earlier, this logarithmic dependence is unavoidable even in the single-agent model (\Cref{thm:single-agent-lower-basic}).

\section{Conclusions and future research}
\label{sec:conclusions}

We introduced compensation design, a framework for incentivizing high-quality contributions in decentralized systems. Unlike budget-feasible mechanism design, compensation design does not hinge on centralized allocation. Instead, in the spirit of utility design, participation is shaped through a (cost-oblivious) payment rule. Our results paint a comprehensive picture of what is possible in this model. Most notably, we established that a simple marginal contribution payment rule attains a sharp constant-factor price of anarchy for monotone submodular value functions under the large-market assumption, and this is so even with respect to coarse correlated equilibria. Moreover, there are obstacles to extending these positive results beyond monotone submodular objectives. Finally, logarithmic price of anarchy guarantees are possible when agents have combinatorial action sets.

Several directions remain open for future research. Our oracle lower bound for XOS valuations applies to value-query access; it would be interesting to understand what is possible in compensation design with stronger oracle access, such as \emph{demand oracles}. Moreover, our model treats the budget as a fixed, exogenous parameter. In some cases, the budget can be adjusted dynamically based on participation and the quality of the contributions generated~\citep{Deng26:Forging}. Extending compensation design to such settings is an important direction.

\section*{Acknowledgments}

For some of the proofs, the authors used Gemini 3.1 Pro and ChatGPT 5.5 Pro to generate candidate ideas or formalize proof sketches, which the authors then verified and edited.

\bibliography{ref}

\clearpage

\appendix

\section{Omitted proofs}
\label{appendix:proofs}

This section contains the proofs omitted from the main body.

\efficShap*
\begin{proof}
The identity $\sum_{i\in S}\phi_i(S)=v(S)$ is the standard efficiency property of the Shapley value. For the inequality, monotonicity and submodularity imply that every marginal contribution appearing in~\eqref{eq:Shapley} is at least the final marginal contribution $v(S)-v(S\setminus\{i\})$. Hence,
\[
    \phi_i(S)
        \geq
        \bigl(v(S)-v(S\setminus\{i\})\bigr)
        \sum_{T\subseteq S\setminus\{i\}}
        \frac{|T|!(|S|-|T|-1)!}{|S|!}
        =
        v(S)-v(S\setminus\{i\}),
\]
where the last equality uses that the Shapley coefficients sum to one. The two displayed properties imply budget feasibility and $1$-marginal-covering for~\eqref{eq:shapley-payment}.
\end{proof}

\lamsplit*

\begin{proof}
    Let \(f(x)=x/(1-x)\). We first note that the constraint \(\sum_i x_i\le 1\) is tight at any maximum. Indeed, if \(\sum_i x_i<1\), then we can
improve the objective \(\sum_i f(x_i)\) by increasing some coordinate (or adding a new coordinate) that is strictly below \(\lambda\).

Moreover, we claim that, at a maximum $\vec{x}$, at most one coordinate can lie
strictly between \(0\) and \(\lambda\). For the sake of contradiction, suppose that there are two coordinates with values $a$ and $b$ that satisfy \(0<a\le b<\lambda\). We replace the pair $(a, b)$ with $(a - \epsilon, b + \epsilon)$ for $\epsilon = \min\{a,\lambda-b\}$. The total mass remains unchanged and the new pair still belongs to \([0,\lambda]^2\). 

Now, we define
\[
        g(t)=f(a-t)+f(b+t) \quad t \in [0, \epsilon].
\]
Since
\[
        f'(x)=\frac{1}{(1-x)^2}
\]
is strictly increasing on \([0,1)\) and \(b+t > a-t\) for all
\(t\in (0,\epsilon]\), we have $g'(t)= -f'(a-t)+f'(b+t)>0$. Thus, $g(\epsilon) > g(0)$, which contradicts the fact that the original point $\vec{x}$ is a maximum. We conclude that at most one nonzero coordinate can be strictly smaller than $\lambda$. In other words, every maximizer is of the form
\[
        \underbrace{\lambda,\ldots,\lambda}_{q \text{ times}},\ 1 - q \lambda,\ 0,\ldots,0,
\]
where $q = \lfloor 1/\lambda \rfloor$. This completes the proof.
\end{proof}

\section{Other cost-oblivious and anonymous payment rules}
\label{appendix:Shapley}

This section examines other cost-oblivious and anonymous payment rules beyond the marginal contribution rule defined in~\Cref{secsec:marginalcontr}.

\subsection{Marginal covering suffices for bounded PoA}
\label{appendix:marginalcovering}

We first extend~\Cref{thm:poa-pne} under the marginal-covering condition introduced in~\Cref{def:marg-cov}.

\begin{theorem}[PoA bound for marginal-covering rules]
\label{thm:marginal-covering}
Consider any $\gamma$-marginal-covering budget-share rule with shares $y_i(S)$ and $\lambda = \max_{i \in N} c_i / B < \gamma$. For every pure Nash equilibrium $S$,\footnote{This does not presuppose the existence of a pure Nash equilibrium. If no pure Nash equilibrium exists, the theorem becomes vacuous.}
\[
        \OPT
        \le
        (1+\costcon_\gamma(\lambda)) v(S),
\]
where
\[
        \costcon_\gamma(\lambda)
        =
        \max\left\{
        \sum_i \frac{ x_i }{\gamma - x_i}:
        0 \le x_i \le \lambda,\ \sum_i x_i \le 1
        \right\}.
\]
In particular,
\[
        \PoAPNE \le
        1+\frac{1}{\gamma-
        \lambda}.
\]
\end{theorem}

\begin{proof}
Let $S$ be a pure Nash equilibrium. If $i \notin S$, we have
\[
        B y_i(S\cup\{i\}) \le c_i.
\]
By the $\gamma$-marginal-covering condition (\Cref{def:marg-cov}),
\[
        \frac{c_i}{B} \geq y_i (S \cup \{i \})
        \ge
        \gamma\frac{v(S\cup\{i\})-v(S)}{v(S\cup\{ i \})}.
\]
Using the fact that $\lambda < \gamma$ and rearranging,
\[
        v(S\cup\{i\})-v(S) \le \frac{c_i/B }{\gamma-c_i/B } v(S).
\]
The rest of the argument is analogous to the proof of~\Cref{thm:poa-pne}.
\end{proof}

\subsection{Equal splitting can have unbounded PoA}

Finally, for completeness, we point out that the simple equal-split rule---whereby $p_i(S) = B/|S|$ for any $i \in S$---can have unbounded price of anarchy.

\begin{proposition}[Equal splitting has unbounded PoA]
\label{prop:equal-split}
For any $\lambda \in (0, 1)$, the equal-split payment rule can have unbounded PoA with respect to pure Nash equilibria even when the value function is additive and $\max_{i \in N} c_i / B \le \lambda$.
\end{proposition}

\begin{proof}
We set $B=1$ and choose any integer $k > 1/\lambda$. We construct an instance as follows. There are $k$ low-value agents $1, 2,\dots, k$ each with cost zero and (additive) value $1/k^2$. There is an additional high-value agent $h$ with cost $\lambda$ and (additive) value $M$, where $M \gg 1$.

We first claim that $S = \{1, 2, \dots, k \}$ is a pure Nash equilibrium. Indeed, each low-value agent receives a positive payment and does not incur any cost, so there is no incentive to opt out. If $h$ opts in, its payment under the equal-split rule will be $\frac{1}{k+1} < \lambda = c_h$, so opting out is a best response for that agent. We conclude that $S$ is a pure Nash equilibrium, and $v(S) = 1/k$. However, selecting the entire set of agents is budget feasible, so $\OPT \geq M$. This completes the proof.
\end{proof}

\section{Targeted implementation}
\label{appendix:targeted-implementation}

The main challenge in compensation design is that the agents' costs are private. Otherwise, there is a simple approach based on what we refer to as \emph{targeted implementation}, which in fact succeeds even under non-monotone (submodular) objectives. The idea here is that one can implement a target set by excluding agents outside of it. This section spells out this approach.

Specifically, let $T\subseteq N$ be a target set; one can think of $T$ as the output of an
approximation algorithm for non-monotone submodular maximization subject to a
knapsack constraint. Suppose further there are strict bonuses $b_i > 0$ such that
\begin{equation}
        \sum_{i\in T}(c_i+b_i)\le B;
        \label{eq:target-budget}
\end{equation}
we will further justify this assumption in~\Cref{sec:strictbonuses}. One can then define the payment rule
\begin{equation}
        p_i(S)
        =
        \begin{cases}
        c_i+b_i & \text{if } i\in S\cap T,\\
        0 & \text{otherwise.}
        \end{cases}
        \label{eq:targeted-rule}
\end{equation}
This rule is budget feasible by~\eqref{eq:target-budget}.

\begin{theorem}[Targeted implementation with outsider exclusion]
\label{thm:targeted-implementation}
Suppose that every non-target agent has strictly positive cost: $c_i > 0$ for all $i \notin T$. Under the payment rule~\eqref{eq:targeted-rule}, every pure Nash equilibrium is
exactly $T$. Moreover, every coarse correlated
equilibrium is supported on the single set $T$.
\end{theorem}

\begin{proof}
Every target agent $i\in T$ has a strict dominant action to opt in: opting in secures a utility $b_i>0$, while opting out gives zero utility. Moreover, every non-target
agent $i \notin T$ has a strict dominant action to opt out: opting in gives
utility $-c_i < 0$, while opting out gives zero utility. Thus, the unique dominant
action profile is $T$, and the claims for pure Nash equilibria and coarse correlated equilibria follow readily.
\end{proof}

In particular, if the target set $T$ has a certain approximation ratio relative to the optimum, the payment rule introduced above inherits that.

\begin{corollary}[Algorithmic PoA for non-monotone objectives]
\label{cor:nonmonotone-algorithmic-poa}
Suppose $v(T)\ge \alpha\OPT$ and~\eqref{eq:target-budget} holds. If all
non-target agents have positive costs, then the rule
in~\eqref{eq:targeted-rule} satisfies
\[
        \PoAPNE
        =
        \PoACCE
        \le \frac{1}{\alpha}.
\]
\end{corollary}

\subsection{Large market slack for strict bonuses}
\label{sec:strictbonuses}

To account for strict bonuses when identifying the target set, an obvious approach is to work with a reduced budget. The following lemma bounds the loss from reducing the budget even without monotonicity.

\begin{lemma}[Reduced-budget comparison]
\label{lem:nonmonotone-reduced-budget}
Let $v$ be normalized, nonnegative, and submodular, but not necessarily
monotone. If $\OPT_B=\max\left\{v(S):\sum_{i\in S} c_i \le B\right\}$ and $\max_{i \in N} c_i\le \lambda B$, then for any $\delta\ge 0$ with
$\delta+\lambda<1$,
\[
        \OPT_{(1-\delta)B}
        \ge
        (1-\delta-\lambda)\OPT_B.
\]
\end{lemma}

\begin{proof}
Let $S^*$ be an optimal set with respect to the budget $B$. If $\sum_{i\in S^*}c_i \le (1-\delta)B$, then $S^*$ itself is feasible for the smaller budget and the claim is immediate.

We can thus assume that $\sum_{i\in S^*}c_i>(1-\delta)B$. Because
$S^*\setminus\{i\}$ is feasible under $B$, optimality of $S^*$ implies
\[
        v(S^*)\ge v(S^*\setminus\{i\})
        \quad \text{for every }i\in S^*.
\]
This means that each final marginal contribution of an item in $S^*$ is nonnegative. By
submodularity, this implies that for any $S \subseteq S^*\setminus\{i\}$,
\[
        v(S \cup\{i\})-v(S)
        \ge
        v(S^*)-v(S^*\setminus\{i\})
        \ge 0.
\]

We now fix an arbitrary order of the elements in $S^*$, and let $a_i$ denote the marginal
contribution of $i$ in this order. Then, $a_i \ge0$ and
\[
        \sum_{i\in S^*} a_i = v(S^*).
\]
Moreover, for every $S \subseteq S^*$, submodularity gives
\begin{equation}
    \label{eq:sub-ais}
    v(S) \ge \sum_{i\in S} a_i.
\end{equation}
We now sort the elements of $S^*$ by nonincreasing density $a_i/c_i$, with zero-cost
positive-value items placed first. Let $P$ be the longest prefix with total cost
at most $(1-\delta)B$. Since each element has cost at most $\lambda B$ and the total cost of $S^*$ exceeds $(1-\delta)B$, the prefix $P$ has cost at least $(1-\delta)B-\lambda B=(1-\delta-\lambda)B$. The density ordering implies
\[
        \sum_{i\in P}a_i
        \ge
        \frac{\sum_{i\in P}c_i}{\sum_{i\in S^*}c_i}
        \sum_{i\in S^*}a_i
        \ge
        (1-\delta-\lambda)v(S^*),
\]
where we used $\sum_{i\in S^*}c_i\le B$. Combining with~\eqref{eq:sub-ais},
\[
        v(P)\ge \sum_{i\in P}a_i
        \ge
        (1-\delta-\lambda)v(S^*)
        =
        (1-\delta-\lambda)\OPT_B.
\]
Since $P$ is feasible for the budget $(1-\delta)B$, the claim follows.
\end{proof}

\begin{corollary}[Reduced-budget implementation]
\label{cor:nonmonotone-reduced-budget-implementation}
Let $T$ be an $\alpha$-approximate solution for the non-monotone submodular
knapsack problem under budget $(1-\delta)B$, so that $v(T)\ge \alpha\OPT_{(1-\delta)B}$. Suppose further that all non-target agents have strictly positive costs and that strict bonuses are such that $\sum_{i\in T}b_i\le \delta B$, so that they fit into the budget. Then the rule~\eqref{eq:targeted-rule} satisfies
\[
        \PoAPNE, \PoACCE
        \le
        \frac{1}{\alpha(1-\delta-\lambda)}.
\]
\end{corollary}

The corollary can be paired with any available approximation algorithm for non-monotone submodular maximization under a knapsack constraint; the resulting approximation factor enters through the parameter $\alpha$.

\end{document}